\documentclass[A4paper,openany,oneside]{article}
\usepackage{geometry}
\usepackage{float}
\usepackage{amssymb}
\usepackage{amsmath}
\usepackage{amsfonts}
\numberwithin{figure}{section}
\usepackage{graphicx}
\usepackage{enumerate} %多种样式的 enumerate 编号
\usepackage{bm} %希腊字母加粗
\usepackage{xcolor} % \textcolor{red}{123}
\usepackage{amsthm}
\newtheorem{theorem}{Theorem}[section]

\newtheorem{lemma}[theorem]{Lemma}

\newtheorem*{remark}{Remark}

\usepackage{hyperref}
\usepackage{xpatch}
\xpatchcmd{\proof}{\itshape}{\normalfont\proofnamefont}{}{}
\newcommand{\proofnamefont}{\bfseries}
\usepackage{mathtools}

\let\oldforall\forall
\let\forall\undefined
\DeclareMathOperator{\forall}{\oldforall}
\let\oldexists\exists
\let\exists\undefined
\DeclareMathOperator{\exists}{\oldexists}

\usepackage{indentfirst}
\usepackage{setspace}

\makeatletter
\def\th@plain{%
	\thm@notefont{}% same as heading font
	\itshape % body font
}
\def\th@definition{%
	\thm@notefont{}% same as heading font
	\normalfont % body font
}
\makeatother
\usepackage{amsmath}
\usepackage{amssymb}
\usepackage{cite}

\usepackage{cases}
\usepackage{geometry}
\usepackage{graphicx}
\usepackage{subfigure}
\usepackage{epsfig}
\usepackage{verbatim}
\usepackage{color}
\usepackage[sort&compress]{natbib}
\allowdisplaybreaks[4]

\title{Exactly or Approximately Wasserstein Distributionally Robust Estimation According to Wasserstein Radii Being Small or Large  }
\author{Xiao Ding$^1$,~Enbin Song$^1$,~Dunbiao Niu$^2$,~Zhujun Cao$^1$,~Qingjiang Shi$^{3,4}$ \\
\\
{\small $^1$The College of Mathematics, Sichuan University, Chengdu, China } \\
{\small (e-mail: xiaoding\_sc@163.com, e.b.song@163.com, zhujuncao\_sc@163.com)} \\
{\small $^2$The Department of Control Science and Engineering, College of Electronics and Information Engineering, }\\
{\small National Key Laboratory of Autonomous Intelligent Unmanned Systems,} \\ 
{\small Tongji University, Shanghai, China  (e-mail: dunbiaoniu\_sc@163.com) } \\
{\small $^3$The School of Software Engineering, Tongji University, Shanghai, China (e-mail: shiqj@tongji.edu.cn)} \\
{\small $^4$Shenzhen Research Institute of Big Data, Shenzhen, China  }
}
\date{}

\begin{document}

	\maketitle
	
{\bf{Abstract:}}
This paper primarily considers the robust estimation problem under Wasserstein distance constraints on the parameter and noise distributions in the linear measurement model with additive noise,
which can be formulated as an infinite-dimensional nonconvex minimax problem.
We prove that the existence of a saddle point for this problem is equivalent to that for a finite-dimensional minimax problem, and give a 
counterexample demonstrating that the saddle point may not exist. 
Motivated by this observation, we present a verifiable necessary and sufficient condition whose parameters can be derived from a convex problem and its dual.
Additionally, we also introduce a simplified sufficient condition, which intuitively indicates that when the Wasserstein radii are small enough, the saddle point always exists.
In the absence of the saddle point, we solve an finite-dimensional nonconvex minimax problem, obtained by restricting the estimator to be linear. Its optimal value establishes an upper bound on the robust estimation problem, 
while its optimal solution yields a robust linear estimator.
Numerical experiments are also provided to validate our theoretical results.


\section{Introduction}
Robust estimation is a classical methodology widely employed in statistics \cite{eldar2004robust,CALAFIORE2001573,el1997robust} 
and engineering \cite{Bertsimas2007Constrained,Kan2005power,vorobyov2003robust} 
to deal with parameter uncertainty arising from limited observable data, noise, outliers, and measurement errors.
Within this framework, decision makers are typically assumed to have only partial information about uncertain parameters.
Based on this partial information, an appropriate uncertainty set can be constructed:
for deterministic parameters, the set contains all possible values \cite{ Eldar2004regret}; 
for random parameters, it can be defined by the family of admissible distributions or by the structured uncertainties about some statistics \cite{ Eldar2006Robust}. 
By minimizing the worst-case risk over the uncertainty set, the estimator can be derived that exhibit relative insensitivity to deviations of the actual model from the assumed model \cite{ben1998robust,ben2000robust}.
However, it also implies that the performance of robust estimator depends heavily on the characterization of uncertainties.

An important research direction in robust estimation for random parameters, known as distributionally robust estimation, is designed to minimize the worst-case risk within a specified distributional uncertainty set.
A common assumption in this framework is that the unknown true distribution is proximate to a predefined nominal distribution in some sense, where the nominal distribution captures known characteristics and can often be chosen as the empirical distribution derived from observed samples. 
Numerous metrics are employed to quantify this proximity, 
including Kullback-Leibler (KL) divergence \cite{levy2004robust}, Wasserstein distance \cite{nguyen2023bridging}, $\phi$-divergence \cite{wang2022distributionally}, moments-based similarity \cite{wang2021robust}, and $\epsilon$-contamination \cite{wang2022distributionallyuo}. 

Under this assumption, the distributionally robust estimation problem can be formulated as an infinite-dimensional minimax problem. 
Existing work has explored related topics under various models.
In \cite{levy2004robust}, the authors derive the optimal least-squares estimator for the least-favorable distribution within a KL-divergence ball centered on a given nominal joint distribution of parameter and observation.
This result has been later applied to state-space models to obtain robust Kalman filtering, yielding significant practical implications \cite{levy2013robust,zorzi2017robust}.
Following this, \cite{niu2023marginal} extends this result to the multi-sensor setting by imposing the same KL divergence constraint on each marginal joint distribution of the parameter and the single-sensor observation.
Additionally, \cite{dytso2019MMSE} considers the linear observation model with additive noise and obtains the %least-favorable 
minimum mean square error estimator corresponding to the least favorable distribution in a KL-divergence ball on the parameter distribution.
Typically, these studies identify an optimal solution to a finite-dimensional auxiliary problem and establish its role as a saddle point solution to the infinite-dimensional minimax problem. 
On the other hand, for the Wasserstein ambiguity sets,\cite{shafieezadeh2018wasserstein} considers the distributionally robust estimation problem under a Wasserstein distance ball on the joint distribution, reformulating the minimax problem as an semidefinite programming (SDP) problem via its Nash equilibrium. 
Subsequently, \cite{wang2022distributionally,wang2021robust,wang2022distributionallyuo} considered similar distributionally robust estimation problems defined by more types of metrics in the state space model.

However, as the model complexity increases, such as the linear measurement model with additive noise by assuming separate uncertainties in both the parameter and the noise, which is a common scenario in practical dynamic systems, the distributionally robust estimation problem becomes nonconvex and the existence of saddle points is not always guaranteed.
In this paper, we mainly focus on a Wasserstein distributionally robust estimation (WDRE) problem in the linear measurement model with additive noise, where the parameter and the noise are mutually independent and their true distributions are
unknown, and each constrained in a Wasserstein-distance ball.
This problem has been considered in \cite{nguyen2023bridging} and the authors prove that it is equivalent to a finite-dimensional SDP problem when a saddle
point exists.
However, an example is proposed in this paper to illustrate that the saddle point solution to this problem does not necessarily exist, motivating the persuit for the conditions that allow for its existence.
On the other hand, if the saddle point is absent, the optimal value derived from the SDP problem can merely provide a lower bound for the optimal value of the WDRE problem. 
Consequently, we must either tackle the original infinite-dimensional nonconvex robust estimation problem directly or establish an upper bound by relaxing the minimax problem.

Our main contributions are summarized as follows:
\begin{itemize}
	
	\item We first prove that the saddle point solution of the WDRE problem exists if and only if that of a finite-dimensional minimax problem exists.
	This finite-dimensional problem is formulated by restricting the distribution to be Gaussian and the estimator to be linear. 
	Then we establish a connection among the optimal values of this finite-dimensional problem, the WDRE problem, and their corresponding problems obtained by interchanging the minimization and maximization.
	Furthermore, we provide an example to illustrate that the saddle point solution of the WDRE problem may not exist.
	
	\item We present a necessary and sufficient condition for the existence of the saddle point solution for the WDRE problem, 
	which can be precisely characterized by determining its parameters through the resolution of a convex problem and its dual.
	Moreover, to simplify the assessment of the existence of the saddle point solution, we also provide a straightforward sufficient condition, 
	which indicates that when the Wasserstein radii of the uncertainty sets are small enough, the saddle point always exists.
	
	\item In the absence of the saddle point, we focus on a robust estimation problem with the linear estimator, which can be formulated as a finite-dimensional nonconvex problem.
	%that arises from restricting the estimator to be a linear function. 
	The optimal value of this problem, which serves as an upper bound for the WDRE problem, can be demonstrated to be equivalent to the optimal value of a SDP problem.
	%obtained through a tight SDP relaxation.
	Furthermore, leveraging the primal and dual optimal solutions of this  SDP problem, we construct the optimal solution to the finite-dimensional nonconvex problem, which yields a robust linear estimator.
\end{itemize}

The rest of this paper is organized as follows.
Section 2 formally presents the WDRE problem.
Section 3 provides a theoretical analysis for the existence of the saddle point solution to the WDRE problem.
In cases where the saddle point is absent, Section 4 addresses an upper bound problem that restricts the estimator to a linear function, which provides a robust linear estimator.
Subsequently, Section 5 presents numerical experiments designed to validate the theoretical results.
Finally, Section 6 concludes this paper.

{\textbf{Notation}}:
Let $\mathbb{R}^n$ be the $n$ dimensional real vector space and 
$\mathbb{R}^{n \times m}$ be the $n \times m$ dimensional real matrix space. 
The notation $I_n$ stands for the identity matrix in $\mathbb{R}^{n \times n}$ and
$\bf{0}$ denotes the zero matrix of proper dimension.
For any $A \in \mathbb{R}^{n \times n}$, we use $A^-$ and $A^\dagger$ to denote the generalized inverse and the Moore-Penrose pseudo-inverse of the matrix $A$.
The notation $\mbox{Tr}(A)$ denotes the trace of the matrix $A$ and
the null space of $A$ is denoted by $\mbox{Null}(A)$.
For any $A, B \in \mathbb{R}^{n \times n}$, the inner product of $A$ and $B$ is denoted by $\left\langle A,B \right\rangle =\mbox{Tr}(AB)$ and $A \perp B$ means $\left\langle A,B \right\rangle =0$.
Let $\mathbb{S}^n$ be the space of symmetric matrices in $\mathbb{R}^{n \times n}$ and $\mathbb{S}^n_+$ be the cone of positive semidefinite matrices in $\mathbb{S}^n$.
For any $A \in \mathbb{S}^n$, we use $\lambda_{\min}(A)$ and $\lambda_{\max}(A)$ to denote the minimum and maximum eigenvalues of $A$, respectively.
For any $A \in \mathbb{S}^n_+$, the notation $A^{\frac{1}{2}}$ denotes the unique positive semidefinite square root of $A$.
For any $A,B \in \mathbb{S}^n $, the relation $A \succ B$ $(A \succeq B)$ means that $A-B$ is positive definite (semidefinite).
%Let $\mathcal{F}$ be the class of measurable functions.
For a measurable function of two variables $f(x,y)$, we use $\bigtriangledown_x f (x,y)$ to denote the partial derivative of $f(x,y)$ with respect to $x$.
The notation $\| \cdot \|$ denotes the Euclidean norm.  
%The convex hull of a set $\Omega$ is denoted by $\mbox{conv}\left\lbrace \Omega \right\rbrace$.
A normal distribution with mean vector $\mu$ and covariance matrix $\Sigma$ is denoted by $\mathcal{N}(\mu,\Sigma)$.

\section{Problem Formulation}
Consider the linear measurement model with additive noise 
\begin{equation} y=H x +w, \end{equation}
where $x \in \mathbb{R}^n$ is the unknown parameter, $y \in \mathbb{R}^m$ is the observation, $H \in \mathbb{R}^{m \times n}$ is the known observation matrix and the noise $w \in \mathbb{R}^m$ is independent of $x$.
%Let $\mathcal{F}$ be the class of measurable functions and 
Let a Lebesgue measurable function $f:\mathbb{R}^m \to \mathbb{R}^n$ be an estimator that estimates parameter $x$ from the noisy observation $y$ and  
\begin{equation}
\mathcal{F} \triangleq \left\{ f \left| f:\mathbb{R}^m \to \mathbb{R}^n \text{ is a Lebesgue measurable function}  \right\} \right. 
\label{fc} \end{equation}
be the set of all estimators. 
Moreover, let $\mathcal{P}_d$ denote the set of probability distributions of a random variable on $\mathbb{R}^d$ with finite second order moments.
Then for a given joint distribution of $x$ and $w$ denoted by $P \in \mathcal{P}_{n+m}$, the mean square error (MSE) obtained by the estimator $f$ is denoted by
\begin{equation}
\mbox{mse}(f,P) \triangleq \int_{\mathbb{R}^n \times \mathbb{R}^m} \left\| f(Hx+w)-x\right\| ^2 \mathrm{d} P(x,w).
\label{mse} \end{equation}
Furthermore, the minimum mean square error (MMSE) of the distribution $P$, as determined by the optimal estimator corresponding to $P$, is provided by
$$\mbox{mmse}(P) \triangleq 
\inf_{f \in \mathcal{F}} \int_{\mathbb{R}^n \times \mathbb{R}^m} \left\| f(Hx+w)-x\right\| ^2 \mathrm{d} P(x,w).$$

Similar to \cite{nguyen2023bridging}, assume that the marginal distributions $P_x$ of the parameter $x$ and $P_w$ of the noise $w$ are unknown. However, given the nominal distributions $\hat{P}_x$ and $\hat{P}_w$, the uncertainty about $P_x$ and $P_w$ can be quantified by their corresponding bounded Wasserstein distances from $\hat{P}_x$ and $\hat{P}_w$.
In this paper, we focus exclusively on the type-2 Wasserstain distance, which is defined for two distributions $P$ and $\hat{P}$ on $\mathbb{R}^d$ as follows:
$$W_2\left( P,\hat{P} \right) \triangleq \inf_{\pi \in \Pi \left( P, \hat{P} \right) }
\left( \int_{\mathbb{R}^d \times \mathbb{R}^d} \left\| \xi_1-\xi_2\right\| ^2 \pi(\mathrm{d}\xi_1,\mathrm{d}\xi_2) \right) ^\frac{1}{2},$$
where $\Pi \left( P, \hat{P} \right)$ denotes the set of all joint distributions of the random variables $\xi_1$ and $\xi_2$ that have marginal distributions $P$ and $\hat{P}$, respectively.
Specifically, for the given nominal Gaussian marginal distributions $\hat{P}_x$ and $\hat{P}_w$, 
due to the independence of $x$ and $w$,
the nominal joint distribution of $x$ and $w$ is $\hat{P}=\hat{P}_x \times \hat{P}_w$, where the product of two distributions means that for any $(x,w)\in \mathbb{R}^n \times \mathbb{R}^m$, $\hat{P}(x,w)=\hat{P}_x(x) \hat{P}_w(w)$.
Taking the Wasserstain radii $\rho_x \geq 0$ and $\rho_w \geq 0$, we assume that the true joint distribution $P$ belongs to the following set
\begin{equation}
%P \in 
\mathbb{B}(\hat{P})\triangleq \left\{ 
P_x \times P_w \in \mathcal{P}_n \times \mathcal{P}_m
\left|
\begin{aligned}
&W_2\left( P_x,\hat{P}_x \right)  \leq \rho_x, \
\hat{P}_x=\mathcal{N}\left(\hat{\mu}_x,\hat{\Sigma}_x \right) \\
&W_2\left( P_w,\hat{P}_w \right) \leq \rho_w,\ 
\hat{P}_w=\mathcal{N}\left( \hat{\mu}_w,\hat{\Sigma}_w \right) 
\end{aligned} \right\} \right.,
\label{set} \end{equation}
where $\mathcal{N}(\mu,\Sigma)$ denotes a normal distribution with mean vector $\mu$ and covariance matrix $\Sigma$.
We intend to find the optimal estimator of the least favorable distribution in set (\ref{set}), i.e., 
considering the following robust estimation problem
\begin{equation}
\inf_{f \in \mathcal{F}} \sup_{P \in \mathbb{B}(\hat{P})}
\mbox{mse} \left( f,P \right),  
\label{minmax} \tag{WDRE}  \end{equation}
where the objective function is given by (\ref{mse}), 
the constraint sets are defined by (\ref{fc}) and (\ref{set}), respectively, 
and the following assumption will be made in this paper.

{\bf{Assumption 1.} }

i) The Wasserstein radii $\rho_x>0$ and $\rho_w>0$. 

ii) The nominal covariance
matrices $\hat{\Sigma}_x$ and  $\hat{\Sigma}_w$ are positive semidefinite.

%Theoretical Analysis on
\section{The Existence of the Saddle Point Solution for (\ref{minmax})}
In this section, we devote to a theoretical analysis on the existence of the saddle point solution for the robust estimation problem (\ref{minmax}). 
To begin with, we demonstrate that the saddle point solution for (\ref{minmax}) exists if and only if a finite-dimensional minimax problem has a saddle point solution. 
Subsequently, we show that the saddle point solution does not always exist for (\ref{minmax}) by means of a counterexample. 
Following this, we present a necessary and sufficient condition for the existence of the saddle point solution. 
Furthermore, we propose a straightforward sufficient condition,
which does not require complex calculations and allows us to ascertain the existence of the saddle point solution when it is satisfied. 

\subsection{An Finite-dimensional Problem Related to the Existence of the Saddle Point Solution for (\ref{minmax})}

We first glance at problem (\ref{minmax}) as an optimization problem with infinite-dimensional variables being the distribution $P$ and the estimator $f$.
Notice that the objective function of (\ref{minmax}) is linear in $P \in \mathbb{B}(\hat{P})$ for a given $f$, and convex in $f \in \mathcal{F}$ for a given $P$. 
However, according to the definition of $\mathbb{B}(\hat{P})$ in (\ref{set}), the elements in the  constraint set of the joint distribution $\mathbb{B}(\hat{P})$ depends on the elements in the two convex Wasserstein balls, which is generally not convex. 
Consequently, it is hard to apply current theoretical results (such as Sion's minimax theorem) to directly deduce the existence of the saddle point solution for (\ref{minmax}).

In order to discuss the existence of the saddle point solution for (\ref{minmax}), 
we first formulate a finite-dimensional minimax problem. 
Specifically, we restrict the distribution to the set consisting only of Gaussian distributions
\begin{equation}
\mathbb{B}_\mathcal{N}\left( \hat{P}\right)  \triangleq
\left\lbrace  P_x \times P_w \in \mathcal{P}_n \times \mathcal{P}_m
\left| \begin{aligned}
&W_2\left( P_x,\hat{P}_x \right)  \leq \rho_x, \
\hat{P}_x=\mathcal{N}\left(\hat{\mu}_x,\hat{\Sigma}_x \right),
P_x=\mathcal{N}\left(\mu_x,\Sigma_x \right) \\
&W_2\left( P_w,\hat{P}_w \right) \leq \rho_w,\ 
\hat{P}_w=\mathcal{N}\left( \hat{\mu}_w,\hat{\Sigma}_w \right),
P_w=\mathcal{N}\left(\mu_w,\Sigma_w \right) 
\end{aligned}  \right\rbrace. \right. 
\label{Guass class} \end{equation}
Correspondingly, the estimator is limited to the set of linear functions 
\begin{equation}
\mathcal{F_L} \triangleq \left\lbrace f \in \mathcal{F} \left| 
\exists A \in \mathbb{R}^{n \times m}, b \in \mathbb{R}^n 
\mbox{ with } f(y)=Ay+b, \ \forall y \in \mathbb{R}^m \right\rbrace. \right.
\label{linear class} \end{equation}
Then we introduce an auxiliary problem
\begin{equation}
\inf_{f \in \mathcal{F_L}} 
\sup_{P \in \mathbb{B}_\mathcal{N}\left( \hat{P}\right) }
\mbox{mse} \left( f,P \right),  
\label{minmaxG} \tag{LG-WDRE} \end{equation}
where the objective function is given by (\ref{mse}), and the constraint sets are defined by (\ref{linear class}) and (\ref{Guass class}), respectively.
Essentially, (\ref{minmaxG}) is a finite-dimensional robust estimation problem, because the linear estimator $f$ can be parameterized by the matrix $A$ and the vector $b$,
and the Gaussian distributions can also be fully described by their mean vectors and covariance matrices.
Next, we shall explore the equivalence between the existence of the saddle point solution in (\ref{minmax}) and that in (\ref{minmaxG}).

First, we review a property of Wasserstain distance. 
%which will be used in the proof of the following theorem.
\begin{lemma}{(\citep[Theorem 4]{kuhn2019wasserstein})} \label{WDp}
	Assume that the nominal distribution is Guassian, i.e., $P_0=\mathcal{N}(\mu_0,\Sigma_0)$.
	For an arbitrary distribution $P$ with mean vector $\mu_P$ and covariance matrix $\Sigma_P$, we have
	$$W_2 \left( P , P_0 \right)  \geq 
	W_2 \left( \mathcal{N} \left( \mu_P,\Sigma_P\right) , P_0 \right).$$
\end{lemma}

Lemma \ref{WDp} shows that for any distribution in $\mathbb{B}(\hat{P})$ 
%whose marginal nominal distributions are Guassian
, the Gaussian distribution defined by its mean vector and covariance matrix is still in  $\mathbb{B}(\hat{P})$. 
Then we have the following lemma.

\begin{lemma} \label{sort}
	Suppose that Assumption 1 holds. The optimal values of (\ref{minmax}), (\ref{minmaxG}) and their corresponding problems obtained by exchanging the minimization and maximization are ordered as follows:
	\begin{equation}
	\sup_{P \in \mathbb{B}_\mathcal{N}(\hat{P})} 
	\inf_{f \in \mathcal{F_L}} \text{mse}(f,P) =
	\sup_{P \in \mathbb{B}(\hat{P})} 
	\inf_{f \in \mathcal{F}} \mbox{mse}(f,P) \leq 
	\inf_{f \in \mathcal{F}}
	\sup_{P \in \mathbb{B}(\hat{P})} \mbox{mse}(f,P) \leq
	%\inf_{f \in \mathcal{F_L}} 
	%\sup_{P \in \mathbb{B}(\hat{P})} \mbox{mse}(f,P) =
	\inf_{f \in \mathcal{F_L}} 
	\sup_{P \in \mathbb{B}_\mathcal{N}(\hat{P})} \mbox{mse}(f,P).
	\label{lemma 2.2} \end{equation}
\end{lemma}

\begin{proof}
	We divide the proof into three parts, corresponding to the equality or inequalities in (\ref{lemma 2.2}), respectively.
	
	{\bf Step 1:} We begin by proving the equality in (\ref{lemma 2.2}).
	
	For any $P \in \mathbb{B}(\hat{P})$, let  $P_\mathcal{N}$ denote a Gaussian distribution $\mathcal{N}(\mu_{P},\Sigma_{P})$ which has the same mean vector $\mu_{P}$ and covariance matrix  $\Sigma_{P}$ as $P$. 
	Then $P_\mathcal{N}$  belongs to $\mathbb{B}(\hat{P})$ according to Lemma \ref{WDp}.
	Define 
	$$\mathbb{B}_P \hspace*{-1mm} \left( \hat{P}\right) \hspace*{-1mm}  \triangleq \hspace*{-1mm}
	\left\lbrace P_\mathcal{N}=\mathcal{N}\left( \mu_{P},\Sigma_{P}\right)  \left| \begin{aligned}
	&\exists P \in \mathbb{B}\left( \hat{P}\right) \\ &\mbox{ with mean vector } \mu_{P} \\ & \mbox{ and covariance matrix }\Sigma_{P} \end{aligned}
	\right\rbrace  %\subseteq \mathbb{B}_\mathcal{N}\left( \hat{P}\right) 
	\right. \hspace*{-1mm} . $$
	Then we have $\mathbb{B}_P (\hat{P}) \subseteq \mathbb{B}_\mathcal{N}(\hat{P})$.
	Combined with $\mathbb{B}_\mathcal{N}(\hat{P}) \subseteq \mathbb{B}_P( \hat{P}) $ obtained from the definition of $\mathbb{B}_\mathcal{N}( \hat{P}) $, we derive $\mathbb{B}_\mathcal{N}(\hat{P})= \mathbb{B}_P(\hat{P}) $.
	
	On the other hand, let the linear function $f_\mathcal{N}$ denote the optimal estimator for $P_\mathcal{N}$. 
	%which is also the optimal linear estimator for $P$ \citep[Lemma 1.1]{nmdef2012networked}. 
	Then we have
	\begin{equation}
	\mbox{mmse}(P_\mathcal{N})=
	\mbox{mse}(f_\mathcal{N},P_\mathcal{N})=
	\mbox{mse}(f_\mathcal{N},P) \geq
	\inf_{f \in \mathcal{F}} \mbox{mse} (f,P) =
	\mbox{mmse}(P),
	\label{ieqc} \end{equation}
	where the second equality holds because for a linear estimator $f(y)=Ay+b$, the MSE given by
	$$\mbox{mse}(f,P)=\int_{\mathbb{R}^n \times \mathbb{R}^m}
	\left\| x-(AHx+Aw+b)\right\| ^2 \mathrm{d}P\left( x,w\right) $$ 
	depends only on the first and second moments of the distribution $P$.  Thus, we obtain
	$$ \sup_{P \in \mathbb{B}(\hat{P})} 
	\inf_{f \in \mathcal{F}} \mbox{mse}(f,P)=
	\sup_{P \in \mathbb{B}(\hat{P})} \mbox{mmse}(P) \leq 
	\sup_{P_\mathcal{N} \in \mathbb{B}_P(\hat{P})} \mbox{mmse}(P_\mathcal{N}) =
	\sup_{P \in \mathbb{B}_\mathcal{N}(\hat{P})} \mbox{mmse}(P)= 
	\sup_{P \in \mathbb{B}_\mathcal{N}(\hat{P})} 
	\inf_{f \in \mathcal{F}} \text{mse}(f,P), $$
	where the inequality follows from (\ref{ieqc}). %and the second inequality is due to $\{P_\mathcal{N}=\mathcal{N}(\mu_{P},\Sigma_{P}) | P \in \mathbb{B}(\hat{P}) \} \subseteq \mathbb{B}_\mathcal{N}(\hat{P})$.
	On the other hand, since $\mathbb{B}_\mathcal{N}(\hat{P}) \subset \mathbb{B}(\hat{P})$, we derive
	$$\sup_{P \in \mathbb{B}_\mathcal{N}(\hat{P})} 
	\inf_{f \in \mathcal{F}} \text{mse}(f,P) \leq
	\sup_{P \in \mathbb{B}(\hat{P})} 
	\inf_{f \in \mathcal{F}} \mbox{mse}(f,P).$$
	Consequently, we have $$\sup_{P \in \mathbb{B}_\mathcal{N}(\hat{P})} 
	\inf_{f \in \mathcal{F}} \text{mse}(f,P) =
	\sup_{P \in \mathbb{B}(\hat{P})} 
	\inf_{f \in \mathcal{F}} \mbox{mse}(f,P)$$ 
	and then the equality in (\ref{lemma 2.2}) arises from the optimality of the linear estimator under the Gaussian distribution.
	
	{\bf Step 2:} The first inequality in (\ref{lemma 2.2}) is evident due to weak duality.
	
	{\bf Step 3:} Finally, we shall prove the last inequality in (\ref{lemma 2.2}).
	
	Since $\mathcal{F_L} \subset \mathcal{F}$, we have
	\begin{equation} \inf_{f \in \mathcal{F}}
	\sup_{P \in \mathbb{B}(\hat{P})} \mbox{mse}(f,P) \leq
	\inf_{f \in \mathcal{F_L}} 
	\sup_{P \in \mathbb{B}(\hat{P})} \mbox{mse}(f,P).
	\label{2.2.3} \end{equation}
	For any $P \in \mathbb{B}(\hat{P})$ with the mean vector $\mu_{P}$ and covariance matrix  $\Sigma_{P}$, the Gaussian distribution 
	$P_\mathcal{N} = \mathcal{N}(\mu_{P}, \Sigma_{P}) $ 
	belonging to $\mathbb{B}(\hat{P})$ according to Lemma \ref{WDp}.
	Then since MSE is only related to the first and second moments of the distribution $P$ for a linear estimator $f \in \mathcal{F_L}$, it holds that 
	$\mbox{mse}(P, f) = \mbox{mse}(P_\mathcal{N}, f)$, 
	which further implies that 
	\begin{equation}
	\inf_{f \in \mathcal{F_L}} 
	\sup_{P \in \mathbb{B}(\hat{P})} \mbox{mse}(f,P) =
	\inf_{f \in \mathcal{F_L}} 
	\sup_{P \in \mathbb{B}_\mathcal{N}(\hat{P})} \mbox{mse}(f,P).
	\label{rlr} \end{equation}
	Thus, the last inequality in (\ref{lemma 2.2}) holds by combining (\ref{2.2.3}) and (\ref{rlr}).
\end{proof}

\begin{remark}
	Formula (4.1) in \cite{nguyen2023bridging} presents a conclusion similar to the one above. However, in Lemma \ref{sort}, we focus on the relationship between the optimal values of (\ref{minmax}), (\ref{minmaxG}) and their corresponding problems obtained by exchanging the minimization and maximization. We further prove the equality of the optimal values of the two maximin problems, which is key to proving the following theorem.
\end{remark}

With the help of Lemma \ref{sort}, we can now establish the equivalence between the existence of the saddle point solution in (\ref{minmax}) and that in (\ref{minmaxG}).
\begin{theorem} \label{equp}
	Suppose that Assumption 1 holds. The saddle point solution of (\ref{minmax}) exists if and only if the saddle point solution of (\ref{minmaxG}) exists.
\end{theorem}

\begin{proof}
    If the saddle point solution of (\ref{minmaxG}) exists, it is easy to show that the saddle point solution of (\ref{minmax}) exists from Lemma \ref{sort}.
    
    Conversely, 
    if $\sup_{P \in \mathbb{B}(\hat{P})} 
    \inf_{f \in \mathcal{F}} \mbox{mse}(f,P)$ admits an optimal solution $(f^*,P^*)$, 
    it follows from (\ref{ieqc}) that 
    $\mbox{mmse}(P^*) \leq \mbox{mmse}(P_\mathcal{N}^*)$, where the
    Gaussian distribution $P_\mathcal{N}^*=\mathcal{N}(\mu_{P^*},\Sigma_{P^*})$ is given by the mean vector $\mu_{P^*}$ and covariance matrix  $\Sigma_{P^*}$ of $P^*$ and belongs to $\mathbb{B}(\hat{P})$ according to Lemma \ref{WDp}.
    Furthermore, since $P^*$ is the least favorable distribution in the sense of MMSE, we have $\mbox{mmse}(P^*) \geq \mbox{mmse}(P_\mathcal{N}^*)$.
    Then it follows that $\mbox{mmse}(P_\mathcal{N}^*)=
    \mbox{mmse}(P^*)$, which means that $P_\mathcal{N}^*$ is also the least favorable distribution in the sense of MMSE.
    
    Since (\ref{minmax}) admits a saddle point solution, according to the optimality of the linear estimator under the Gaussian distribution and
    the Cartesian product property of the saddle point 
    \citep[Theorem 6.2.9]{stoer2012convexity}, 
    we can deduce that there exists $f_\mathcal{N}^* \in \mathcal{F_L}$ such that $(f_\mathcal{N}^*,P_\mathcal{N}^*)$ is a saddle point solution of (\ref{minmax}). 
    Consequently, $(f_\mathcal{N}^*,P_\mathcal{N}^*)$ is also the optimal solution to $\sup_{P \in \mathbb{B}_\mathcal{N}(\hat{P})} \inf_{f \in \mathcal{F_L}} \mbox{mse}(f,P)$ according to the equality in (\ref{lemma 2.2}), and for given $f_\mathcal{N}^*$, $P_\mathcal{N}^*$ is the least favorable distribution on the constraint set $\mathbb{B}(\hat{P})$. Then since $P^*_\mathcal{N} \in \mathbb{B}_\mathcal{N}(\hat{P}) \subseteq \mathbb{B}(\hat{P})$, $P_\mathcal{N}^*$ is also the least favorable distribution on the smaller constraint set $\mathbb{B}_\mathcal{N}(\hat{P})$.
    Therefore, $(f_\mathcal{N}^*,P_\mathcal{N}^*)$ constitutes a saddle point solution of (\ref{minmaxG}).
\end{proof}

The above theorem shows that if a saddle point solution exists for (\ref{minmax}),  there must be a saddle point solution consisting of a linear estimator and a Gaussian distribution, which is also a saddle point solution for (\ref{minmaxG}).
%Similarly, since the KL divergence exhibits a property analogous to that detailed in Lemma 2.1 \citep[Lemma 3.1]{niu2023marginal}, it follows that Lemma 2.2 and Theorem 2.3 is also suitable to the anlysis for robust estimation problems where the uncertain sets are characterized by KL divergence.
Unfortunately, we have constructed a counterexample in which the saddle point solution of (\ref{minmaxG}) does not exist.

{\bf{Counterexample:}}
Consider the scalar case, i.e., $m=n=1$.
Let the nominal means and variances be 
$\hat{\mu}_x=\hat{\mu}_w=0$ and
$\hat{\Sigma}_x=\hat{\Sigma}_w=1$, respectively.
Take the Wasserstein radii $\rho_x=\rho_w=2$ and the observation matrix $H=1$.
Consider the problem
$$\sup_{P \in \mathbb{B}_\mathcal{N}\left( \hat{P}\right) } \inf_{f \in \mathcal{F_L}} \text{mse}\left( f,P\right) .$$
It can be parameterized as follows \citep[Theorem 3.1]{nguyen2023bridging}
\begin{equation} \begin{aligned}
\sup_{\mu_x,\mu_w,\Sigma_x,\Sigma_w} 
&\Sigma_x-\Sigma_x \left(\Sigma_x+\Sigma_w \right) ^{-1} \Sigma_x  \\
\mbox{s.t.} \ &\mu_x^2 + \Sigma_x+1-2 \sqrt{\Sigma_x} \leq 4, \ \Sigma_x \geq 0 ,\\
& \mu_w^2+ \Sigma_w+1-2 \sqrt{\Sigma_w} \leq 4, \ \Sigma_w \geq 0 .
\label{example} \end{aligned} \end{equation}
Then the objective function of (\ref{example}) can be expressed as
\begin{align*}
\Sigma_x-\Sigma_x \left( \Sigma_x+\Sigma_w \right) ^{-1} \Sigma_x 
&=\Sigma_x \left( \Sigma_x+\Sigma_w \right) ^{-1} \left( \Sigma_x+\Sigma_w \right) -\Sigma_x \left( \Sigma_x+\Sigma_w \right) ^{-1} \Sigma_x \\
&=\Sigma_x \left(  \Sigma_x+\Sigma_w \right) ^{-1} \Sigma_w \\
&=\left(\Sigma_w^{-1}+\Sigma_x^{-1} \right) ^{-1}  ,
\end{align*}
which is monotonically increasing with $\Sigma_x$ and $\Sigma_w$.
Therefore, (\ref{example}) has a unique optimal solution given by $\mu_x^*=\mu_w^*=0$, $\Sigma_x^*=\Sigma_w^*=9$, i.e., the least favorable distribution $P^*=\mathcal{N}(0,9) \times \mathcal{N}(0,9)$.
Then the corresponding optimal estimator is given by 
$$f^*(y)= \frac{H \Sigma_x^*}{H^2 \Sigma_x^*+\Sigma_w^*}y=\frac{1}{2}y,$$
and the mean square error is
$$\mbox{mse}(f^*,P^*)=\Sigma_x^*-\Sigma_x^* \left(\Sigma_x^*+\Sigma_w^* \right) ^{-1} \Sigma_x^*=\frac{9}{2}.$$
But in fact, 
for $\tilde{P} \hspace*{-1mm}=\hspace*{-1mm} \mathcal{N}\left( \sqrt{3},4\right)  \times \mathcal{N}\left( -\sqrt{3},4\right)$,
the type-2 Wasserstain Distance between $\tilde{P}$ and the nominal distribution is given by
$$W_2 \left(\mathcal{N}(\sqrt{3},4), \mathcal{N}(0,1)\right) = \sqrt{\|\sqrt{3} \hspace*{-1mm}-\hspace*{-1mm} 0\|^2 \hspace*{-1mm}+\hspace*{-1mm} \left[ 4+1 \hspace*{-1mm}-\hspace*{-1mm} 2\times 4^\frac{1}{2}\right] }
\hspace*{-1mm}=\hspace*{-1mm} 2,$$
and
$$W_2 \left(\mathcal{N}(-\sqrt{3},4), \mathcal{N}(0,1)\right) \hspace*{-1mm}=\hspace*{-1mm} \sqrt{\|\hspace*{-1mm}-\hspace*{-1mm}\sqrt{3}\hspace*{-1mm}-\hspace*{-1mm} 0\|^2 \hspace*{-1mm}+\hspace*{-1mm} \left[ 4 \hspace*{-1mm}+\hspace*{-1mm} 1 \hspace*{-1mm}-\hspace*{-1mm} 2 \hspace*{-1mm}\times\hspace*{-1mm} 4^\frac{1}{2}\right] } \hspace*{-1mm}=\hspace*{-1mm} 2.$$
It implies that $\tilde{P} \in \mathbb{B}_\mathcal{N}( \hat{P}) \subseteq \mathbb{B}(\hat{P})$.
And we have
\begin{align*}
\mbox{mse}( f^*,\tilde{P})
&=\int_{\mathbb{R} \times \mathbb{R}} \left\|  \frac{1}{2} (x+w) - x \right\| ^2 \mathrm{d} \tilde{P} \\
&=\frac{1}{4} \Sigma_x +\frac{1}{4} \Sigma_w
+\left( \frac{1}{2} \mu_x - \frac{1}{2}\mu_w \right) ^2 \\
&=5> \mbox{mse}\left( f^*,P^*\right) .\end{align*}
Hence, for the given $f^*$, the corresponding least favorable distribution is not $P^*$, which indicates that the unique solution to the maximin problem $\left( f^*,P^*\right) $ does not form a saddle point solution of (\ref{minmaxG}). 

\begin{remark}
The example above shows that, in general, the saddle point solution of (\ref{minmaxG}) is not guaranteed to exist without additional assumptions, which implies that the saddle point solution of (\ref{minmax}) may not exist either.
In addition, a multi-dimensional example is given by simulation 1 in Section 5.
\end{remark}

\subsection{A Necessary and Sufficient Condition for the Existence of the Saddle Point Solution for (\ref{minmax})}
When a saddle point solution exists, it is easy to deduce from Theorem \ref{equp} that solving the infinite-dimensional robust estimation problem (\ref{minmax}) is equivalent to solving the finite-dimensional minimax problem (\ref{minmaxG}).
Under such scenario, all inequalities in Lemma \ref{sort} become equalities.
Consequently, as detailed in \cite{nguyen2023bridging}, the optimal solution to (\ref{minmax}) can be effectively obtained by solving the tractable problem 
\begin{equation} \sup_{P \in \mathbb{B}_\mathcal{N}\left( \hat{P}\right) } 
\inf_{f \in \mathcal{F_L}} \mbox{mse}\left( f,P\right) . \label{maxmina} \end{equation}
Conversely, in the absence of the saddle point, the optimal value of (\ref{maxmina}) serves merely as a strict lower bound for the optimal value of (\ref{minmax}). 
It complicates the resolution of 
(\ref{minmax}) and motivates us to establish a necessary and sufficient condition for its existence in this case.

To identitify this condition, we commence our analysis with (\ref{maxmina}).
Since the distribution and the estimator are restricted to be a Gaussian distribution in (\ref{Guass class}) and a linear estimator in (\ref{linear class}) respectively, (\ref{maxmina}) can be reformulated as a finite-dimensional optimization problem, where MSE can be naturally parameterized as the following objective function by its definition, and the Wasserstein distance degenerates into the Gelbrich distance \cite{gelbrich1990formula}:
\begin{align}
\sup_{\mu_x,\mu_w,\Sigma_x,\Sigma_w} \inf_{A,b} \ \
&\mbox{Tr}\left[ \left( AH-I_n\right) \Sigma_x \left( AH-I_n\right) ^\top +A\Sigma_w A^\top \right]
+\left[ \left( AH-I_n \right) \mu_x+A\mu_w+b\right] ^\top 
\left[ \left( AH-I_n \right) \mu_x+A\mu_w+b\right] \nonumber  \\
\mbox{s.t.} \ 
&  \Sigma_x \succeq{\bf 0}, \quad \Sigma_w \succeq {\bf 0}, \nonumber \\
& \left\| \mu_x-\hat{\mu}_x \right\| ^2+
\mbox{Tr}\left[\Sigma_x + \hat{\Sigma}_x - 2 \left( \hat{\Sigma}_x^{\frac{1}{2}}\Sigma_x\hat{\Sigma}_x^{\frac{1}{2}} \right)^{\frac{1}{2}}\right] \leq \rho_x^2, \nonumber \\[0.5ex]
& \left\| \mu_w-\hat{\mu}_w \right\| ^2+
\mbox{Tr}\left[\Sigma_w + \hat{\Sigma}_w - 2 \left( \hat{\Sigma}_w^{\frac{1}{2}}\Sigma_w\hat{\Sigma}_w^{\frac{1}{2}} \right)^{\frac{1}{2}}\right] \leq \rho_w^2.
\label{maxminap} \end{align}
For fixed $\mu_x$, $\mu_w$, $\Sigma_x$, $\Sigma_w$ and $A$, the optimal solution to the inner minimization problem is given by 
$b^* = (I_n - A H)\mu_x - A \mu_w$,  
which minimizes the objective function by eliminating the quadratic term with respect to $b$.
%reducing the quadratic term to zero.  
By substituting $b^*$ into the above equation, problem (\ref{maxminap}) becomes
\begin{equation*}
\begin{aligned}
\sup \limits_{\mu_x,\mu_w,\Sigma_{x},\Sigma_w} \inf \limits_A \ \
&\mbox{Tr}\left[ \left( AH-I_n \right) \Sigma_x \left( AH-I_n \right) ^\top +A\Sigma_w A^\top \right]\\
\mbox{s.t.} \ 
&\Sigma_x \succeq{\bf 0}, \quad \Sigma_w \succeq {\bf 0}, \\%[0.5ex]
&\left\| \mu_x-\hat{\mu}_x\right\| ^2+\mbox{Tr}\left[\Sigma_x + \hat{\Sigma}_x - 2 \left( \hat{\Sigma}_x^{\frac{1}{2}}\Sigma_x\hat{\Sigma}_x^{\frac{1}{2}} \right)^{\frac{1}{2}}\right] \leq \rho_x^2, \\%[0.5ex]
&\left\| \mu_w-\hat{\mu}_w\right\| ^2+\mbox{Tr}\left[\Sigma_w + \hat{\Sigma}_w - 2 \left( \hat{\Sigma}_w^{\frac{1}{2}}\Sigma_w\hat{\Sigma}_w^{\frac{1}{2}} \right)^{\frac{1}{2}}\right] \leq \rho_w^2. \\%[0.5ex]
\end{aligned} \end{equation*}
It is straightforward to observe that the objective function is independent of $\mu_x$ and $\mu_w$.
Consequently, setting $\mu_x^*=\hat{\mu}_x$ and $\mu_w^*=\hat{\mu}_w$ can result in the largest feasible set of $\Sigma_x$ and $\Sigma_w$, which leads to the largest objective function value. 
%As a result, the feasible set of $\Sigma_x$ and $\Sigma_w$ is the largest when $\mu_x^*=\hat{\mu}_x$ and $\mu_w^*=\hat{\mu}_w$, which leads to the smallest optimal value.
It follows that the optimal solution must satisfy $\mu_x^*=\hat{\mu}_x$ and $\mu_w^*=\hat{\mu}_w$. 
Accordingly, the above problem can be further simplified to
\begin{equation}
\begin{aligned}
\sup \limits_{\Sigma_{x},\Sigma_w} \inf \limits_A \ \ 
&\mbox{Tr}\left[ (AH-I_n)\Sigma_x(AH-I_n)^\top +A\Sigma_w A^\top \right]\\
\mbox{s.t.} \
&  \Sigma_x \succeq{\bf 0}, \quad \Sigma_w \succeq {\bf 0}, \\%[0.5ex]
& \mbox{Tr}\left[\Sigma_x + \hat{\Sigma}_x - 2 \left( \hat{\Sigma}_x^{\frac{1}{2}}\Sigma_x\hat{\Sigma}_x^{\frac{1}{2}} \right)^{\frac{1}{2}}\right] \leq \rho_x^2, \\%[0.5ex]
& \mbox{Tr}\left[\Sigma_w + \hat{\Sigma}_w - 2 \left( \hat{\Sigma}_w^{\frac{1}{2}}\Sigma_w\hat{\Sigma}_w^{\frac{1}{2}} \right)^{\frac{1}{2}}\right] \leq \rho_w^2. \\%[0.5ex]
\end{aligned} \label{maxmin1} \end{equation}

For the sake of simplicity of notations, we denote the constraint set in (\ref{maxmin1}) by
$$\mathbb{B}_\Sigma=\left\lbrace \left( \Sigma_x,\Sigma_w \right) \left|
\begin{aligned}
&\Sigma_x \succeq{\bf 0}, \
\mbox{Tr}\left[\Sigma_x + \hat{\Sigma}_x - 2 \left( \hat{\Sigma}_x^{\frac{1}{2}}\Sigma_x\hat{\Sigma}_x^{\frac{1}{2}} \right)^{\frac{1}{2}}\right] \leq \rho_x^2 \\
&\Sigma_x \succeq{\bf 0}, \
\mbox{Tr}\left[\Sigma_w + \hat{\Sigma}_w - 2 \left( \hat{\Sigma}_w^{\frac{1}{2}}\Sigma_w\hat{\Sigma}_w^{\frac{1}{2}} \right)^{\frac{1}{2}}\right] \leq \rho_w^2
\end{aligned} \right. \right\rbrace .$$
Since the objective function of (\ref{maxmin1}) is convex in $A$ for each $(\Sigma_x,\Sigma_w)$, concave in $(\Sigma_x,\Sigma_w)$ for each $A$, and $\mathbb{B}_\Sigma$ is convex and compact \cite{kuhn2019wasserstein}, 
we can apply Sion's minimax theorem \cite{sion1958general} to conclude that
\begin{align}
&\sup_{(\Sigma_x,\Sigma_w) \in \mathbb{B}_\Sigma} \inf_A \quad
\left\langle \left( I_n-AH\right) ^\top \left( I_n-AH\right) ,\Sigma_{x} \right\rangle 
+\left\langle  A^\top A,\Sigma_w \right\rangle \label{maxminsi} \\
&=\inf_A \sup_{(\Sigma_x,\Sigma_w) \in \mathbb{B}_\Sigma}
\left\langle  \left( I_n-AH \right) ^\top \left( I_n-AH \right) ,\Sigma_{x} \right\rangle 
+\left\langle  A^\top A,\Sigma_w \right\rangle, \label{maxminis}
\end{align}
where (\ref{maxminsi}) is a reformulation of (\ref{maxmin1}) with simplified notations.
Furthermore, since $\mathbb{B}_\Sigma$ is compact, and there exists a pair of positive definite matrices $(\Sigma_x,\Sigma_w) \in \mathbb{B}_\Sigma$ such that the objective function is coercive in $A$, then the above problems admit the saddle point solution \citep[Theorem 10.2]{botelho2014functional}.

Then problem (\ref{maxminsi}) can be transformed into
\begin{align}
&\sup_{(\Sigma_x,\Sigma_w) \in \mathbb{B}_\Sigma} \inf_A \quad
\left\langle  \left( I_n-AH\right) ^\top \left( I_n-AH\right) ,\Sigma_{x} \right\rangle 
+\left\langle  A^\top A,\Sigma_w \right\rangle  \notag \\
&\begin{aligned}
=\sup_{\substack{\Sigma_x\succeq{\bf 0}\\ \Sigma_w\succeq{\bf 0}}}
\inf_{\substack {\gamma_x \geq 0 \\ \gamma_w \geq 0} } \ 
&\left[ \inf_A \quad
\left\langle  \left( I_n-AH\right) ^\top \left( I_n-AH\right) ,\Sigma_{x} \right\rangle +\left\langle  A^\top A,\Sigma_w \right\rangle \right]+ \notag \\
&\gamma_x \left\{ \rho_x^2-
\mbox{Tr}\left[\Sigma_x + \hat{\Sigma}_x - 2 \left( \hat{\Sigma}_x^{\frac{1}{2}}\Sigma_x\hat{\Sigma}_x^{\frac{1}{2}} \right)^{\frac{1}{2}}\right] \right\}+
\gamma_w \left\{ \rho_w^2-
\mbox{Tr}\left[\Sigma_w + \hat{\Sigma}_w - 2 \left( \hat{\Sigma}_w^{\frac{1}{2}}\Sigma_w\hat{\Sigma}_w^{\frac{1}{2}} \right)^{\frac{1}{2}}\right] \right\}
\end{aligned} \\
&\begin{aligned}
=\sup_{\substack{\Sigma_x\succeq{\bf 0}\\ \Sigma_w\succeq{\bf 0}}}
\inf_A \inf_{\substack {\gamma_x \geq 0 \\ \gamma_w \geq 0} } \quad
&\left\langle  \left(I_n-AH\right) ^\top \left(I_n-AH\right) ,\Sigma_{x} \right\rangle 
+\gamma_x \left\{ \rho_x^2-
\mbox{Tr}\left[\Sigma_x + \hat{\Sigma}_x - 2 \left( \hat{\Sigma}_x^{\frac{1}{2}}\Sigma_x\hat{\Sigma}_x^{\frac{1}{2}} \right)^{\frac{1}{2}}\right] \right\}  \\
&+\left\langle  A^\top A,\Sigma_w \right\rangle 
+\gamma_w \left\{ \rho_w^2-
\mbox{Tr}\left[\Sigma_w + \hat{\Sigma}_w - 2 \left( \hat{\Sigma}_w^{\frac{1}{2}}\Sigma_w\hat{\Sigma}_w^{\frac{1}{2}} \right)^{\frac{1}{2}}\right] \right\}, \label{maxminL}
\end{aligned} \end{align}
where the first equality arises from reformulating the constraints as penalty terms in the objective function.
Specifically, when $(\Sigma_x,\Sigma_w) \in \mathbb{B}_\Sigma$, 
the corresponding optimal multiplier $(\gamma_x,\gamma_w)$ satisfies
$$\gamma_x \left\lbrace \rho_x^2-\mbox{Tr}\left[\Sigma_x + \hat{\Sigma}_x - 2 \left( \hat{\Sigma}_x^{\frac{1}{2}} \Sigma_x \hat{\Sigma}_x^{\frac{1}{2}} \right)^{\frac{1}{2}}\right] \right\rbrace =0$$
and
$$\gamma_w \left\lbrace \rho_w^2-\mbox{Tr}\left[\Sigma_w + \hat{\Sigma}_w - 2 \left( \hat{\Sigma}_w^{\frac{1}{2}} \Sigma_w \hat{\Sigma}_w^{\frac{1}{2}} \right)^{\frac{1}{2}}\right] \right\rbrace =0,$$
thereby keeping the objective function value unchanged; otherwise, the  multiplier $\gamma_x$ or $\gamma_w$ tends to $+\infty$, driving the objective function to $-\infty$.

On the other hand, for any fixed $A$, consider the inner maximization problem in (\ref{maxminis}) 
\begin{equation}
\sup_{(\Sigma_x,\Sigma_w) \in \mathbb{B}_\Sigma}
\left\langle  \left( I_n-AH\right) ^\top \left( I_n-AH \right) ,\Sigma_{x} \right\rangle 
+\left\langle  A^\top A,\Sigma_w \right\rangle. \label{maxminis in}
\end{equation}
Problem (\ref{maxminis in}) maximizes a continuous concave function on a convex and compact set, and thus it has a finite optimal value.
Moreover, since the cartesian set of two cones of positive semidefinite matrices is convex, the objective function and the Gelbrich distance are also convex in $(\Sigma_x,\Sigma_w)$, then strong duality holds under Slater condition, i.e., (\ref{maxminis in}) is equivalent to its Lagrangian dual problem
\begin{equation} 
\begin{aligned}
\inf_{\substack {\gamma_x \geq 0 \\ \gamma_w \geq 0} } \sup_{\substack{\Sigma_x\succeq{\bf 0}\\ \Sigma_w\succeq{\bf 0}}} \quad
&\left\langle  \left( I_n-AH\right) ^\top \left( I_n-AH\right) ,\Sigma_{x} \right\rangle 
+\gamma_x \left\{ \rho_x^2-
\mbox{Tr}\left[\Sigma_x + \hat{\Sigma}_x - 2 \left( \hat{\Sigma}_x^{\frac{1}{2}}\Sigma_x\hat{\Sigma}_x^{\frac{1}{2}} \right)^{\frac{1}{2}}\right] \right\} \\
&+\left\langle  A^\top A,\Sigma_w \right\rangle 
+\gamma_w \left\{ \rho_w^2-
\mbox{Tr}\left[\Sigma_w + \hat{\Sigma}_w - 2 \left( \hat{\Sigma}_w^{\frac{1}{2}}\Sigma_w\hat{\Sigma}_w^{\frac{1}{2}} \right)^{\frac{1}{2}}\right] \right\}, 
\end{aligned} \label{maxminis in dual} \end{equation}
and the set of optimal dual solutions is nonempty \citep[Propositions 5.1.4 and 5.3.1]{bertsekas1997nonlinear}.
Thus, problem (\ref{maxminis}) is equivalent to
\begin{equation} 
\begin{aligned}
\inf_A \inf_{\substack {\gamma_x \geq 0 \\ \gamma_w \geq 0} } \sup_{\substack{\Sigma_x\succeq{\bf 0}\\ \Sigma_w\succeq{\bf 0}}} \quad
&\left\langle  \left( I_n-AH\right) ^\top \left( I_n-AH\right) ,\Sigma_{x} \right\rangle 
+\gamma_x \left\{ \rho_x^2-
\mbox{Tr}\left[\Sigma_x + \hat{\Sigma}_x - 2 \left( \hat{\Sigma}_x^{\frac{1}{2}}\Sigma_x\hat{\Sigma}_x^{\frac{1}{2}} \right)^{\frac{1}{2}}\right] \right\} \\
&+\left\langle  A^\top A,\Sigma_w \right\rangle 
+\gamma_w \left\{ \rho_w^2-
\mbox{Tr}\left[\Sigma_w + \hat{\Sigma}_w - 2 \left( \hat{\Sigma}_w^{\frac{1}{2}}\Sigma_w\hat{\Sigma}_w^{\frac{1}{2}} \right) ^{\frac{1}{2}} \right] \right\}. \label{maxminLe}
\end{aligned} \end{equation}
We denote the objective function of (\ref{maxminLe}), which is also the objective function of (\ref{maxminL}), as
\begin{align}
G(\Sigma_x,\Sigma_w,A,\gamma_x,\gamma_w) \triangleq
&\left\langle  \left( I_n-AH \right) ^\top \left( I_n-AH \right) ,\Sigma_{x} \right\rangle 
+\gamma_x \left\{ \rho_x^2-
\mbox{Tr}\left[\Sigma_x + \hat{\Sigma}_x - 2 \left( \hat{\Sigma}_x^{\frac{1}{2}}\Sigma_x\hat{\Sigma}_x^{\frac{1}{2}} \right)^{\frac{1}{2}}\right] \right\} \notag \\
&+\left\langle  A^\top A,\Sigma_w \right\rangle 
+\gamma_w \left\{ \rho_w^2-
\mbox{Tr}\left[\Sigma_w + \hat{\Sigma}_w - 2 \left( \hat{\Sigma}_w^{\frac{1}{2}}\Sigma_w\hat{\Sigma}_w^{\frac{1}{2}} \right)^{\frac{1}{2}} \right] \right\}. \label{G}
\end{align} 
Then since the optimal values of problems (\ref{maxminsi}) and (\ref{maxminis}) are equal, it follows that the optimal values of (\ref{maxminL}) and (\ref{maxminLe}) are also equal, i.e.,
$$\sup_{\substack{\Sigma_x\succeq{\bf 0}\\ \Sigma_w\succeq{\bf 0}}}
\inf_A \inf_{\substack {\gamma_x \geq 0 \\ \gamma_w \geq 0} } 
G(\Sigma_x,\Sigma_w,A,\gamma_x,\gamma_w)=
\inf_A \inf_{\substack {\gamma_x \geq 0 \\ \gamma_w \geq 0} } \sup_{\substack{\Sigma_x\succeq{\bf 0}\\ \Sigma_w\succeq{\bf 0}}}
G(\Sigma_x,\Sigma_w,A,\gamma_x,\gamma_w).$$

Furthermore, we shall give the relationship between saddle point solutions to problems (\ref{maxminsi}) and (\ref{maxminL}).

\begin{lemma}{(\citep[Theorem 6.9.8]{stoer2012convexity})} \label{spr}
    Under Assumption 1, the following statements hold:  
	
	(i)If problem (\ref{maxminL}) has a saddle point solution $(\Sigma_x^*,\Sigma_w^*,A^*,\gamma_x^*,\gamma_w^*)$, then $(\Sigma_x^*,\Sigma_w^*,A^*)$ is a saddle point solution to (\ref{maxminsi}). 
	
	(ii)If problem (\ref{maxminsi}) has a saddle point solution $(\Sigma_x^*,\Sigma_w^*,A^*)$ and Slater condition holds, then there are $\gamma_x^* \geq 0$ and $\gamma_w^* \geq 0$ such that $(\Sigma_x^*,\Sigma_w^*,A^*,\gamma_x^*,\gamma_w^*)$ is a saddle point solution to (\ref{maxminL}).
\end{lemma}

Subsequently, the saddle point solution of problem (\ref{maxminL}) enables us to derive a necessary and sufficient condition for the existence of a saddle point solution for the robust estimation problem (\ref{minmax}).

\begin{lemma} \label{snsd}
	Suppose that Assumption 1 holds. 
	The saddle point solution of (\ref{minmax}) exists if and only if (\ref{maxminL}) has a saddle point solution $(\Sigma_x^*,\Sigma_w^*,A^*,\gamma_x^*,\gamma_w^*)$ such that 
	$$\begin{bmatrix}
	(I_n - A^* H)^\top (I_n - A^* H) - \gamma_x^*I_n & (I_n - A^* H)^\top A^* \\
	(A^*)^\top (I_n - A^* H) & (A^*)^\top A^* - \gamma_w^*I_m
	\end{bmatrix} \preceq {\bf 0}.$$
\end{lemma}

\begin{proof}
For ease of notations, define $K^* \triangleq I_n - A^* H$. Notice that Theorem \ref{equp} shows that the saddle point solution of (\ref{minmax}) exists if and only if that of (\ref{minmaxG}) exists.
Then since (\ref{maxmina}) is obtained by exchanging the supremum and infimum of (\ref{minmaxG}) and it can be parameterized as (\ref{maxminap}), 
it suffices to prove that the saddle point solution of (\ref{maxminap}) exists if and only if  (\ref{maxminL}) has a saddle point solution $(\Sigma_x^*,\Sigma_w^*,A^*,\gamma_x^*,\gamma_w^*)$ such that 
$\begin{bmatrix}
(K^*)^\top (K^*) - \gamma_x^*I_n & (K^*)^\top A^* \\
(A^*)^\top (K^*) & (A^*)^\top A^* - \gamma_w^*I_m
\end{bmatrix} \preceq {\bf 0} $. 
%if and only if the saddle point solution of (\ref{maxminap}) exists.

``$\Longleftarrow$" {\bf{Sufficiency.}}

We first demonstrate that, 
if problem (\ref{maxminL}) has a saddle point solution $(\Sigma_x^*,\Sigma_w^*,A^*,\gamma_x^*,\gamma_w^*)$ such that 
$\begin{bmatrix}
(K^*)^\top (K^*) - \gamma_x^*I_n & (K^*)^\top A^* \\
(A^*)^\top (K^*) & (A^*)^\top A^* - \gamma_w^*I_m
\end{bmatrix} \preceq {\bf 0} $,
then $(A^*,b^*,\hat{\mu}_x,\hat{\mu}_w,\Sigma_x^*,\Sigma_w^*)$ 
constitutes a saddle point solution of (\ref{maxminap}),
where $\hat{\mu}_x$ and $\hat{\mu}_w$ are the nominal mean vectors of parameter and noise, respectively, and $b^* \triangleq (I_n - A^* H)\hat{\mu}_x - A^* \hat{\mu}_w$.

Since $(\Sigma_x^*,\Sigma_w^*,A^*,\gamma_x^*,\gamma_w^*)$ is a saddle point solution to (\ref{maxminL}), 
it follows from Lemma \ref{spr} that $(\Sigma_x^*, \Sigma_w^*,A^*)$ is an optimal solution to (\ref{maxminsi}), which implies that
$(A^*,b^*,\hat{\mu}_x,\hat{\mu}_w,\Sigma_x^*,\Sigma_w^*)$ 
is an optimal solution to (\ref{maxminap}).
To further establish that $(A^*,b^*,\hat{\mu}_x,\hat{\mu}_w,\Sigma_x^*,\Sigma_w^*)$ 
is a saddle point solution of (\ref{maxminap}), it suffices to show that, for given $A^*$ and $b^*$, the least favorable distribution is determined by the mean vectors $\hat{\mu}_x$, $\hat{\mu}_w$ and covariance matrices $\Sigma_x^*$, $\Sigma_w^*$.
Specifically, we intend to prove that 
$(\hat{\mu}_x,\hat{\mu}_w,\Sigma_x^*,\Sigma_w^*)$
is an optimal solution to the following problem
\begin{equation}
\begin{aligned}
\sup_{\mu_x,\mu_w,\Sigma_x,\Sigma_w}
&\left\langle  \left( K^*\right) ^\top K^*,\Sigma_{x} \right\rangle +
\left\langle  \left( A^*\right) ^\top A^*,\Sigma_w \right\rangle 
+\left\| K^*\left( \mu_x-\hat{\mu}_x \right) 
-A^*\left( \mu_w-\hat{\mu}_w\right) \right\| ^2\\
\mbox{s.t.} \ 
& \Sigma_x \succeq{\bf 0}, \quad \Sigma_w \succeq {\bf 0}, \\
& \|\mu_x-\hat{\mu}_x\|^2+
\mbox{Tr}\left[\Sigma_x + \hat{\Sigma}_x - 2 \left( \hat{\Sigma}_x^{\frac{1}{2}}\Sigma_x\hat{\Sigma}_x^{\frac{1}{2}} \right)^{\frac{1}{2}}\right] \leq \rho_x^2, \\[0.5ex]
& \|\mu_w-\hat{\mu}_w\|^2+
\mbox{Tr}\left[\Sigma_w + \hat{\Sigma}_w - 2 \left( \hat{\Sigma}_w^{\frac{1}{2}}\Sigma_w\hat{\Sigma}_w^{\frac{1}{2}} \right)^{\frac{1}{2}}\right] \leq \rho_w^2.
\end{aligned} \label{saddle judge}
\end{equation}
Then we denote the Lagrangian function of (\ref{saddle judge}) as
\begin{align}
L(\mu_x,\mu_w,\Sigma_x,\Sigma_w;\gamma_x,\gamma_w) \triangleq
&\left\langle  \left( K^*\right) ^\top K^*,\Sigma_{x} \right\rangle +
\left\langle  \left( A^*\right) ^\top A^*,\Sigma_w \right\rangle 
+\left\| K^*\left( \mu_x-\hat{\mu}_x\right) -A^*\left( \mu_w-\hat{\mu}_w\right) \right\| ^2 \notag \\
&+\gamma_x \left\{ \rho_x^2- \|\mu_x-\hat{\mu}_x\|^2-
\mbox{Tr}\left[\Sigma_x + \hat{\Sigma}_x - 2 \left( \hat{\Sigma}_x^{\frac{1}{2}}\Sigma_x\hat{\Sigma}_x^{\frac{1}{2}} \right)^{\frac{1}{2}}\right] \right\} \notag \\
&+\gamma_w \left\{ \rho_w^2- \|\mu_w-\hat{\mu}_w\|^2-
\mbox{Tr}\left[\Sigma_w + \hat{\Sigma}_w - 2 \left( \hat{\Sigma}_w^{\frac{1}{2}}\Sigma_w\hat{\Sigma}_w^{\frac{1}{2}} \right)^{\frac{1}{2}}\right] \right\} \notag \\
=& \left\langle  \left( K^*\right) ^\top K^*,\Sigma_{x} \right\rangle 
+\gamma_x \left\{ \rho_x^2-
\mbox{Tr}\left[\Sigma_x + \hat{\Sigma}_x - 2 \left( \hat{\Sigma}_x^{\frac{1}{2}}\Sigma_x\hat{\Sigma}_x^{\frac{1}{2}} \right)^{\frac{1}{2}}\right] \right\} \notag \\
&+\left\langle  \left( A^*\right) ^\top A^*,\Sigma_w \right\rangle 
+\gamma_w \left\{ \rho_w^2-
\mbox{Tr}\left[\Sigma_w + \hat{\Sigma}_w - 2 \left( \hat{\Sigma}_w^{\frac{1}{2}}\Sigma_w\hat{\Sigma}_w^{\frac{1}{2}} \right)^{\frac{1}{2}}\right] \right\} \notag \\
&+\begin{bmatrix}
\mu_x-\hat{\mu}_x \\ \mu_w-\hat{\mu}_w
\end{bmatrix}^\top
\begin{bmatrix}
(K^*)^\top K^*-\gamma_xI_n & (K^*)^\top A^* \\
(A^*)^\top K^* & (A^*)^\top A^* -\gamma_wI_m
\end{bmatrix}
\begin{bmatrix}
\mu_x-\hat{\mu}_x \\ \mu_w-\hat{\mu}_w
\end{bmatrix}, \label{F}
\end{align}
where the equality comes from rearranging items and combining like items. 
Next we discuss the relationship between $L$ and the objective function of (\ref{maxminL}), which is defined as $G$ in (\ref{G}).
Note that taking $\mu_x=\hat{\mu}_x$ and $\mu_w=\hat{\mu}_w$ eliminates the quadratic term with respect to $(\mu_x, \mu_w)$ in $L$, and setting $A=A^*$ in $G$ ensures that its inner probuct terms with respect to $\Sigma_x$ and $\Sigma_w$ are identical to that in $L$, i.e., for any $\Sigma_x \succeq{\bf 0}$, $\Sigma_w \succeq {\bf 0}$, $\gamma_x \geq 0$ and $\gamma_w \geq 0$,
\begin{equation}
L(\hat{\mu}_x,\hat{\mu}_w,\Sigma_x,\Sigma_w;\gamma_x,\gamma_w)=
G(\Sigma_x,\Sigma_w,A^*,\gamma_x,\gamma_w). 
\label{eqL} \end{equation}
Since $(\Sigma_x^*, \Sigma_w^*, A^*, \gamma_x^*, \gamma_w^*)$ is a saddle point solution of (\ref{maxminL}), 
%according to the definition of the saddle point, 
for any $\Sigma_x \succeq{\bf 0}$, $\Sigma_w \succeq {\bf 0}$, $A$, $\gamma_x \geq 0$ and $\gamma_w \geq 0$, we have
\begin{equation}
G(\Sigma_x,\Sigma_w,A^*,\gamma_x^*,\gamma_w^*) \leq 
G(\Sigma_x^*,\Sigma_w^*,A^*,\gamma_x^*,\gamma_w^*) \leq
G(\Sigma_x^*,\Sigma_w^*,A,\gamma_x,\gamma_w). 
\label{sd per} \end{equation}
Then for any $\mu_x$, $\mu_w$, $\Sigma_x \succeq{\bf 0}$, $\Sigma_w \succeq {\bf 0}$, $\gamma_x \geq 0$ and $\gamma_w \geq 0$, we obtain
\begin{align*}
&L(\mu_x,\mu_w,\Sigma_x,\Sigma_w;\gamma_x^*,\gamma_w^*) 
\leq L(\hat{\mu}_x,\hat{\mu}_w,\Sigma_x,\Sigma_w;\gamma_x^*,\gamma_w^*)
=G(\Sigma_x,\Sigma_w,A^*,\gamma_x^*,\gamma_w^*)\\
&\leq G(\Sigma_x^*,\Sigma_w^*,A^*,\gamma_x^*,\gamma_w^*)= L(\hat{\mu}_x,\hat{\mu}_w,\Sigma_x^*,\Sigma_w^*;\gamma_x^*,\gamma_w^*) \\
&\leq  G(\Sigma_x^*,\Sigma_w^*,A^*,\gamma_x,\gamma_w) = L(\hat{\mu}_x,\hat{\mu}_w,\Sigma_x^*,\Sigma_w^*;\gamma_x,\gamma_w),
\end{align*}
where the first inequality is due to
$\begin{bmatrix}
(K^*)^\top K^*-\gamma_x^*I_n & (K^*)^\top A^* \\
(A^*)^\top K^* & (A^*)^\top A^* -\gamma_w^*I_m
\end{bmatrix} \preceq {\bf 0}$;
the second inequality arises from the first inequality in (\ref{sd per});
the last inequality follows from the second inequality in (\ref{sd per}) with $A=A^*$ and all equalities are due to (\ref{eqL}).
Therefore, the Lagrangian function $L$ admits a saddle point $(\hat{\mu}_x,\hat{\mu}_w,\Sigma_x^*,\Sigma_w^*,\gamma_x^*,\gamma_w^*)$. 
Combined with saddle point theorem 
\citep[Proposition 5.1.6]{bertsekas1997nonlinear}, 
it further implies that $(\hat{\mu}_x,\hat{\mu}_w,\Sigma_x^*,\Sigma_w^*)$ 
is an optimal solution to (\ref{saddle judge}).

``$\Longrightarrow$" {\bf{Necessity.}}

Assuming that (\ref{maxminap}) has a saddle point solution $(A^*,b^*,\hat{\mu}_x,\hat{\mu}_w,\Sigma_x^*,\Sigma_w^*)$, it follows that: (i) (\ref{maxminsi}) has an optimal solution $(\Sigma_x^*,\Sigma_w^*,A^*)$,  and (ii) for the given $A^*$ and $b^*$, the least favorable distribution is determined by the mean vectors $\hat{\mu}_x$, $\hat{\mu}_w$ and covariance matrices $\Sigma_x^*$, $\Sigma_w^*$, that is, (\ref{saddle judge}) has an optimal solution $(\hat{\mu}_x,\hat{\mu}_w,\Sigma_x^*,\Sigma_w^*)$. 

We first consider the following equivalent problem of (\ref{saddle judge})
\begin{equation} \begin{aligned}
\sup_{\substack{\Sigma_x\succeq{\bf 0}\\ \Sigma_w\succeq{\bf 0}}}
%\sup_{\substack{(\Sigma_x,\Sigma_w) \in \mathbb{B}_\Sigma}} 
\sup_{\mu_x,\mu_w} 
\inf_{\substack {\gamma_x \geq 0 \\ \gamma_w \geq 0} } \
&\left\langle  \left( K^*\right) ^\top K^*,\Sigma_{x} \right\rangle +
\left\langle  \left( A^*\right) ^\top A^*,\Sigma_w \right\rangle 
+\|K^*(\mu_x-\hat{\mu}_x)-A^*(\mu_w-\hat{\mu}_w)\|^2\\
&+\gamma_x \left\{ \rho_x^2- \|\mu_x-\hat{\mu}_x\|^2-
\mbox{Tr}\left[\Sigma_x + \hat{\Sigma}_x - 2 \left( \hat{\Sigma}_x^{\frac{1}{2}}\Sigma_x\hat{\Sigma}_x^{\frac{1}{2}} \right)^{\frac{1}{2}}\right] \right\}  \\
&+\gamma_w \left\{ \rho_w^2- \|\mu_w-\hat{\mu}_w\|^2-
\mbox{Tr}\left[\Sigma_w + \hat{\Sigma}_w - 2 \left( \hat{\Sigma}_w^{\frac{1}{2}}\Sigma_w\hat{\Sigma}_w^{\frac{1}{2}} \right)^{\frac{1}{2}}\right] \right\}.  \\
\end{aligned} \label{saddle judge1} \end{equation}
For fixed $\Sigma_x$ and $\Sigma_w$, the maximization problem over $(\mu_x,\mu_w)$ in (\ref{saddle judge}) is a quadratically constrained quadratic program (QCQP) in which the two constraint functions and the objective function are all homogeneous quadratic. 
Under tha assumption that $\rho_x>0$ and $\rho_w>0$, Slater condition is readily verified.
Then according to Theorem 2.5 in \cite{ye2003new}, strong duality holds for the maximization problem over $(\mu_x,\mu_w)$ in (\ref{saddle judge}) and its dual.
Specifically, 
the maximization over $(\mu_x,\mu_w)$ and the minimization over $(\gamma_x,\gamma_w)$ in (\ref{saddle judge1}) can be interchanged. 
Then (\ref{saddle judge1}) is equivalent to %the following problem
\begin{equation} \begin{aligned}
\sup_{\substack{\Sigma_x\succeq{\bf 0}\\ \Sigma_w\succeq{\bf 0}}}
\inf_{\substack {\gamma_x \geq 0 \\ \gamma_w \geq 0} } \sup_{\mu_x,\mu_w} \ 
&\left\langle  \left( K^*\right) ^\top K^*,\Sigma_{x} \right\rangle 
+\gamma_x \left\{ \rho_x^2-
\mbox{Tr}\left[\Sigma_x + \hat{\Sigma}_x - 2 \left( \hat{\Sigma}_x^{\frac{1}{2}}\Sigma_x\hat{\Sigma}_x^{\frac{1}{2}} \right)^{\frac{1}{2}}\right] \right\} \\
&+\left\langle  \left( A^*\right) ^\top A^*,\Sigma_w \right\rangle 
+\gamma_w \left\{ \rho_w^2-
\mbox{Tr}\left[\Sigma_w + \hat{\Sigma}_w - 2 \left( \hat{\Sigma}_w^{\frac{1}{2}}\Sigma_w\hat{\Sigma}_w^{\frac{1}{2}} \right)^{\frac{1}{2}}\right] \right\}  \\
&+\begin{bmatrix}
\mu_x-\hat{\mu}_x \\ \mu_w-\hat{\mu}_w
\end{bmatrix}^\top
\begin{bmatrix}
(K^*)^\top K^*-\gamma_xI_n & (K^*)^\top A^* \\
(A^*)^\top K^* & (A^*)^\top A^* -\gamma_wI_m
\end{bmatrix}
\begin{bmatrix}
\mu_x-\hat{\mu}_x \\ \mu_w-\hat{\mu}_w
\end{bmatrix}.
\end{aligned} \label{saddle judge2} \end{equation}
For any pair of positive semidefinite matrices $(\Sigma_x,\Sigma_w)$, consider the inner minimax problem in (\ref{saddle judge2}).
It is easy to know that its optimal solution satisfies
\begin{equation}
(\gamma^*_x,\gamma^*_w) 
\in \mathcal{A} \triangleq \left\lbrace 
(\gamma_x,\gamma_w) \left|
\begin{bmatrix}
(K^*)^\top K^*-\gamma_xI_n & (K^*)^\top A^* \\
(A^*)^\top K^* & (A^*)^\top A^* -\gamma_wI_m
\end{bmatrix} \preceq \bf{0} \right. \right\rbrace 
\label{calA} \end{equation}
and the corresponding maxmizer with respect to $(\mu_x,\mu_w)$ makes the the quadratic term equal to zero, i.e.,
\begin{equation} \begin{bmatrix}
\mu_x^*-\hat{\mu}_x \\ \mu_w^*-\hat{\mu}_w
\end{bmatrix} \in \mbox{Null}
\left(  \begin{bmatrix}
(K^*)^\top K^*-\gamma^*_xI_n & (K^*)^\top A^* \\
(A^*)^\top K^* & (A^*)^\top A^* -\gamma^*_wI_m
\end{bmatrix} \right). \label{nulla} \end{equation}
Otherwise, if (\ref{calA}) does not hold, the matrix $\begin{bmatrix}
(K^*)^\top K^*-\gamma_xI_n & (K^*)^\top A^* \\
(A^*)^\top K^* & (A^*)^\top A^* -\gamma_wI_m
\end{bmatrix}$ has at least one positive eigenvalue.
In this case, let $\begin{bmatrix}
\mu_x^*-\hat{\mu}_x \\ \mu_w^*-\hat{\mu}_w
\end{bmatrix}$ take the direction of the corresponding eigenvector, and as the norm tends to $\infty$, the objective function value tends to $+\infty$.
On the other hand, based on (\ref{calA}), it is obvious that (\ref{nulla}) holds.
Thus, (\ref{saddle judge2}) is further equivalent to
\begin{equation} \begin{aligned}
\sup_{\substack{\Sigma_x\succeq{\bf 0}\\ \Sigma_w\succeq{\bf 0}}} \inf_{(\gamma_x,\gamma_w)\in \mathcal{A}}  \
&\left\langle  \left( K^*\right) ^\top K^*,\Sigma_{x} \right\rangle 
+\gamma_x \left\{ \rho_x^2-
\mbox{Tr}\left[\Sigma_x + \hat{\Sigma}_x - 2 \left( \hat{\Sigma}_x^{\frac{1}{2}}\Sigma_x\hat{\Sigma}_x^{\frac{1}{2}} \right)^{\frac{1}{2}}\right] \right\} \\
&+\left\langle  \left( A^*\right) ^\top A^*,\Sigma_w \right\rangle 
+\gamma_w \left\{ \rho_w^2-
\mbox{Tr}\left[\Sigma_w + \hat{\Sigma}_w - 2 \left( \hat{\Sigma}_w^{\frac{1}{2}}\Sigma_w\hat{\Sigma}_w^{\frac{1}{2}} \right)^{\frac{1}{2}}\right] \right\}.
\end{aligned} \label{saddle judge3} \end{equation}
Note that the objective function of (\ref{saddle judge3}) is linear in $(\gamma_x,\gamma_w)$ for each $(\Sigma_x,\Sigma_w)$, concave in $(\Sigma_x,\Sigma_w)$ for each $(\gamma_x,\gamma_w)$, and the constraint sets $\mathbb{S}_+^n \times \mathbb{S}_+^m$ and $\mathcal{A}$ are convex,  closed and nonempty.
Moreover, 
the objective function of (\ref{saddle judge3}) with $\Sigma_x=\hat{\Sigma}_x\succeq{\bf 0}$ and $\Sigma_w=\hat{\Sigma}_w\succeq{\bf 0}$ 
tends to $+\infty$ as $\gamma_x$ or $\gamma_w$ tends to $+\infty$.
On the other hand, there exists $(\bar{\gamma}_x,\bar{\gamma}_w) \in \mathcal{A}$ such that $\left( K^*\right) ^\top K^*-\bar{\gamma}_x I_n \prec \bf{0}$ and $\left( A^*\right) ^\top A^*-\bar{\gamma}_w I_m \prec \bf{0}$.
Without loss of generality, we only consider the spectral norm below, while other norms can be handled analogously via the equivalence of norms in finite-dimensional spaces.
Then the objective function of (\ref{saddle judge3}) with $(\bar{\gamma}_x,\bar{\gamma}_w)$ can be transformed to
\begin{align}
&\hspace*{-4mm} -\left\langle  \bar{\gamma}_x I_n-\left( K^*\right) ^\top K^*,\Sigma_{x} \right\rangle 
+2 \bar{\gamma}_x \mbox{Tr} \left[  \left( \hat{\Sigma}_x^{\frac{1}{2}}\Sigma_x\hat{\Sigma}_x^{\frac{1}{2}} \right)^{\frac{1}{2}}\right]
+\bar{\gamma}_x \left[ \rho_x^2- \mbox{Tr}\left( \hat{\Sigma}_x \right) \right] \notag \\
&-\left\langle \bar{\gamma}_w I_m- \left( A^*\right) ^\top A^*,\Sigma_w \right\rangle 
+2 \bar{\gamma}_w \mbox{Tr}\left[ \left( \hat{\Sigma}_w^{\frac{1}{2}}\Sigma_w\hat{\Sigma}_w^{\frac{1}{2}} \right)^{\frac{1}{2}}\right]
+\bar{\gamma}_w \left[ \rho_w^2- \mbox{Tr}\left( \hat{\Sigma}_w \right) \right] 
\notag \\
\leq &-\lambda_{\min}\left[ \bar{\gamma}_x I_n-\left( K^*\right) ^\top K^* \right]  \mbox{Tr}\left( \Sigma_x\right) +
2 \bar{\gamma}_x \sum_{i=1}^n \lambda_i^\frac{1}{2}
\left( \hat{\Sigma}_x \Sigma_x \right) +\bar{\gamma}_x \left[ \rho_x^2- \mbox{Tr}\left( \hat{\Sigma}_x \right) \right] \notag \\
&-\lambda_{\min} \left[ \bar{\gamma}_w I_m- \left( A^*\right) ^\top A^* \right] \mbox{Tr}\left( \Sigma_w\right) +
2 \bar{\gamma}_w \sum_{i=1}^m \lambda_i^\frac{1}{2}
\left( \hat{\Sigma}_w \Sigma_w \right) +\bar{\gamma}_w \left[ \rho_w^2- \mbox{Tr}\left( \hat{\Sigma}_w \right) \right] \notag \\
\leq &-\lambda_{\min}\left[ \bar{\gamma}_x I_n-\left( K^*\right) ^\top K^* \right]  \left\|  \Sigma_x \right\| +
2 \bar{\gamma}_x n \lambda_{\max}^\frac{1}{2}\left( \hat{\Sigma}_x \right) 
\left\|  \Sigma_x \right\|^\frac{1}{2}
+\bar{\gamma}_x \left[ \rho_x^2- \mbox{Tr}\left( \hat{\Sigma}_x \right) \right] \notag \\
&-\lambda_{\min}\left[ \bar{\gamma}_w I_m-\left( A^*\right) ^\top A^* \right]  \left\|  \Sigma_w \right\| +
2 \bar{\gamma}_w m \lambda_{\max}^\frac{1}{2}\left( \hat{\Sigma}_w \right) 
\left\|  \Sigma_w \right\|^\frac{1}{2}
+\bar{\gamma}_w \left[ \rho_w^2- \mbox{Tr}\left( \hat{\Sigma}_w \right) \right], \notag
\end{align}
where the first inequality is due to  $\mbox{Tr} \left[ \left( B^\frac{1}{2} D B^\frac{1}{2}\right)^\frac{1}{2} \right]=\sum_{i=1}^n \lambda_i^\frac{1}{2}\left( BD\right) $ for $B,D \in \mathbb{S}^n_+$ \citep[Fact 10.14.22]{bernstein2018scalar}, 
and the second inequality follows from $\lambda_i(BD) \leq \lambda_{\max}(BD) \leq
\lambda_{\max}(B) \lambda_{\max}(D)$ for $B,D \in \mathbb{S}^n_+$ \citep[Fact 10.22.28]{bernstein2018scalar}.
Hence, The objective function of (\ref{saddle judge3}) with $(\bar{\gamma}_x,\bar{\gamma}_w)$ tends to $-\infty$ as $\left\|  \Sigma_x \right\|$ or $\left\|  \Sigma_w \right\|$ tends to $+\infty$, owing to the positive definiteness of $\bar{\gamma}_x I_n-\left( K^*\right) ^\top K^*$ and $\bar{\gamma}_w I_m-\left( A^*\right) ^\top A^*$.

Therefore, according to Theorem 10.2 in \cite{botelho2014functional}, problem (\ref{saddle judge3}) is equivalent to
\begin{equation} \begin{aligned}
\inf_{(\gamma_x,\gamma_w) \in \mathcal{A}} \ 
\sup_{\substack{\Sigma_x\succeq{\bf 0}\\ \Sigma_w\succeq{\bf 0}}} \
&\left\langle  \left( K^*\right) ^\top K^*,\Sigma_{x} \right\rangle 
+\gamma_x \left\{ \rho_x^2-
\mbox{Tr}\left[\Sigma_x + \hat{\Sigma}_x - 2 \left( \hat{\Sigma}_x^{\frac{1}{2}}\Sigma_x\hat{\Sigma}_x^{\frac{1}{2}} \right)^{\frac{1}{2}}\right] \right\} \\
&+\left\langle  \left( A^*\right) ^\top A^*,\Sigma_w \right\rangle 
+\gamma_w \left\{ \rho_w^2-
\mbox{Tr}\left[\Sigma_w + \hat{\Sigma}_w - 2 \left( \hat{\Sigma}_w^{\frac{1}{2}}\Sigma_w\hat{\Sigma}_w^{\frac{1}{2}} \right)^{\frac{1}{2}}\right] \right\}.
\end{aligned} \label{mid} \end{equation}
Furthermore, notice that problem (\ref{mid}) is equivalent to 
\begin{equation} \begin{aligned}
\inf_{\substack {\gamma_x \geq 0 \\ \gamma_w \geq 0} }
\sup_{\substack{\Sigma_x\succeq{\bf 0}\\ \Sigma_w\succeq{\bf 0}}} 
\sup_{\mu_x,\mu_w}  \ 
&\left\langle  \left( K^*\right) ^\top K^*,\Sigma_{x} \right\rangle 
+\gamma_x \left\{ \rho_x^2-
\mbox{Tr}\left[\Sigma_x + \hat{\Sigma}_x - 2 \left( \hat{\Sigma}_x^{\frac{1}{2}}\Sigma_x\hat{\Sigma}_x^{\frac{1}{2}} \right)^{\frac{1}{2}}\right] \right\} \\
&+\left\langle  \left( A^*\right) ^\top A^*,\Sigma_w \right\rangle 
+\gamma_w \left\{ \rho_w^2-
\mbox{Tr}\left[\Sigma_w + \hat{\Sigma}_w - 2 \left( \hat{\Sigma}_w^{\frac{1}{2}}\Sigma_w\hat{\Sigma}_w^{\frac{1}{2}} \right)^{\frac{1}{2}}\right] \right\}  \\
&+\begin{bmatrix}
\mu_x-\hat{\mu}_x \\ \mu_w-\hat{\mu}_w
\end{bmatrix}^\top
\begin{bmatrix}
(K^*)^\top K^*-\gamma_xI_n & (K^*)^\top A^* \\
(A^*)^\top K^* & (A^*)^\top A^* -\gamma_wI_m
\end{bmatrix}
\begin{bmatrix}
\mu_x-\hat{\mu}_x \\ \mu_w-\hat{\mu}_w
\end{bmatrix},
\end{aligned} \label{saddle judge22} \end{equation}
since when $(\gamma_x,\gamma_w) \in \mathcal{A}$, as defined in (\ref{calA}), the supremum of the quadratic term with respect to $(\mu_x,\mu_w)$ in (\ref{saddle judge22}) is zero and thus the objective function values of (\ref{mid}) and (\ref{saddle judge22}) are identical; otherwise, the supremum is $+\infty$ as $\mu_x$ or $\mu_w$ tends to $\infty$.
Consequently, based on the above argument that (\ref{saddle judge1}) is equivalent to (\ref{saddle judge22}), we further have
$$%\sup_{\substack{\mu_x,\mu_w \\(\Sigma_x,\Sigma_w) \in \mathbb{B}_\Sigma}}
\sup_{\substack{\Sigma_x \succeq {\bf 0} \\ \Sigma_w\succeq {\bf 0}}}
\sup_{\mu_x,\mu_w}
\inf_{\substack {\gamma_x \geq 0 \\ \gamma_w \geq 0}} 
L(\mu_x,\mu_w,\Sigma_x,\Sigma_w;\gamma_x,\gamma_w)=
\inf_{\substack {\gamma_x \geq 0 \\ \gamma_w \geq 0} }
%\sup_{\substack{\mu_x,\mu_w \\(\Sigma_x,\Sigma_w) \in \mathbb{B}_\Sigma}}
\sup_{\substack{\Sigma_x \succeq {\bf 0} \\ \Sigma_w\succeq {\bf 0}}}
\sup_{\mu_x,\mu_w} 
L(\mu_x,\mu_w,\Sigma_x,\Sigma_w;\gamma_x,\gamma_w),
$$
where the left side of the above equality is (\ref{saddle judge1}) and the right side is (\ref{saddle judge22}) by the definition of $L$ in (\ref{F}). %because the maximum value of the quadratic term with respect to $(\mu_x,\mu_w)$ in $F$ is $+\infty$ when $(\gamma_x,\gamma_w) \notin \mathcal{A}$, and 0 when $(\gamma_x,\gamma_w) \in \mathcal{A}$.
Then for the optimal solution $(\hat{\mu}_x,\hat{\mu}_w,\Sigma_x^*,\Sigma_w^*)$ to (\ref{saddle judge}), there exists an optimal dual solution $(\gamma^*_x$, $\gamma^*_w) \in \mathcal{A}$ such that for any $\mu_x$, $\mu_w$, $\Sigma_x \succeq {\bf 0}$, $\Sigma_w \succeq {\bf 0}$, $\gamma_x \geq 0$ and $\gamma_w \geq 0$, 
the Lagrangian function of (\ref{saddle judge}), defined as $L$ in (\ref{F}), satisfies
\begin{equation}
L(\mu_x,\mu_w,\Sigma_x,\Sigma_w;\gamma^*_x,\gamma^*_w) \leq
L(\hat{\mu}_x,\hat{\mu}_w,\Sigma_x^*,\Sigma_w^*;
\gamma^*_x,\gamma^*_w) \leq  
L(\hat{\mu}_x,\hat{\mu}_w,\Sigma_x^*,\Sigma_w^*;\gamma_x,\gamma_w).
\label{sd perd} \end{equation}
Then for any $\Sigma_x \succeq {\bf 0}$, $\Sigma_w \succeq {\bf 0}$, $A$, $\gamma_x \geq 0$ and $\gamma_w \geq 0$, the objective function of (\ref{maxminL}),  defined as $G$ in (\ref{G}), satisfies
\begin{align*}
&G(\Sigma_x,\Sigma_w,A^*,\gamma^*_x,\gamma^*_w)=
L(\hat{\mu}_x,\hat{\mu}_w,\Sigma_x,\Sigma_w;\gamma^*_x,\gamma^*_w)\\
&\leq  L(\hat{\mu}_x,\hat{\mu}_w,\Sigma_x^*,\Sigma_w^*;
\gamma^*_x,\gamma^*_w)=
G(\Sigma_x^*,\Sigma_w^*,A^*,\gamma^*_x,\gamma^*_w)\\
&\leq L(\hat{\mu}_x,\hat{\mu}_w,\Sigma_x^*,\Sigma_w^*;\gamma_x,\gamma_w)
=G(\Sigma_x^*,\Sigma_w^*,A^*,\gamma_x,\gamma_w)
\leq G(\Sigma_x^*,\Sigma_w^*,A,\gamma_x,\gamma_w),
\end{align*}
where the first inequality follows from the first inequality in (\ref{sd perd}) for $\mu_x=\hat{\mu}_x$ and $\mu_w=\hat{\mu}_w$; the second inequality arises from the second inequality in (\ref{sd perd}); the last inequality is due to the fact that $(\Sigma_x^*,\Sigma_w^*,A^*)$ is an optimal solution of (\ref{maxminsi}) and all equalities follow from (\ref{eqL}).
Therefore, $(\Sigma_x^*,\Sigma_w^*,A^*,\gamma_x^*,\gamma_w^*)$
is a saddle point solution of
(\ref{maxminL}) such that
$\begin{bmatrix}
(K^*)^\top K^*-\gamma^*_xI_n & (K^*)^\top A^* \\
(A^*)^\top K^* & (A^*)^\top A^* -\gamma^*_wI_m
\end{bmatrix} \preceq {\bf{0}}$.
\end{proof}

Nevertheless, the condition presented in Lemma \ref{snsd} not easily checked in certain cases. 
Specifically, if the matrix $H\Sigma_x^* H^\top + \Sigma_w^*$ is not positive definite, the uniqueness of the optimal solution $A^*$ for (\ref{maxmin1}) is no longer guaranteed and it becomes challenging to identify which solution can form a saddle point solution for (\ref{maxminL}).  
However, Section 4 proposes that for this case, the existence of a saddle point can be determined numerically by solving two SDPs.
Therefore, we shall only focus on a simplified case under a mild assumption below.

{\bf{Assumption 2.} }

i) The Wasserstein radii $\rho_x>0$ and $\rho_w>0$. 

ii) The nominal covariance
matrices $\hat{\Sigma}_x$ and  $\hat{\Sigma}_w$ are positive definite.

Refocusing on problem (\ref{maxmin1}), it is well known that
for fixed $\Sigma_x$ and $\Sigma_w$, the optimal solution to the inner minimization problem is given by $A^* = \Sigma_x H^\top (H \Sigma_x H^\top + \Sigma_w)^{-}$.
By substituting $A^*$ into (\ref{maxmin1}), we can derive its equivalent problem as follows:
\begin{equation}
\begin{array}{cl}
\sup \limits_{\Sigma_{x},\Sigma_w}
& \mbox{Tr}\left[\Sigma_x - \Sigma_x H^\top \left( H \Sigma_x H^\top +
\Sigma_w \right)^{-} H \Sigma_x\right] \\
\mbox{s.t.} &  \Sigma_x \succeq{\bf 0}, \quad \Sigma_w \succeq {\bf 0}, \\%[0.5ex]
& \mbox{Tr}\left[\Sigma_x + \hat{\Sigma}_x - 2 \left( \hat{\Sigma}_x^{\frac{1}{2}}\Sigma_x\hat{\Sigma}_x^{\frac{1}{2}} \right)^{\frac{1}{2}}\right] \leq \rho_x^2, \\[0.5ex]
& \mbox{Tr}\left[\Sigma_w + \hat{\Sigma}_w - 2 \left( \hat{\Sigma}_w^{\frac{1}{2}}\Sigma_w\hat{\Sigma}_w^{\frac{1}{2}} \right)^{\frac{1}{2}}\right] \leq \rho_w^2. \\[0.5ex]
\end{array} \label{maxmin} \end{equation}

The objective function of (\ref{maxmin}), i.e.,
the inner minimization problem of (\ref{maxmin1}) 
$$\inf_A \quad
\left\langle \left(I_n-AH\right)^\top \left(I_n-AH\right),\Sigma_{x} \right\rangle 
+\left\langle  A^\top A,\Sigma_w \right\rangle$$
can be interpreted as minimizing a set of linear functions over the variable $(\Sigma_x,\Sigma_w)$,  which implies that it is concave in $(\Sigma_x,\Sigma_w)$ \cite{boyd2004convex}. 
Consequently, (\ref{maxmin}) is a convex optimization problem in which strong duality holds under Slater condition.
As detailed in Corollary 3.1 of \cite{nguyen2023bridging}, problem (\ref{maxmin}) can be transformed into the following SDP problem 
\begin{equation}
\begin{aligned}
\sup \limits_{\Sigma_x,\Sigma_w,V_x,V_w,U} 
&\mbox{Tr}(\Sigma_x)-\mbox{Tr}(U) \\
\mbox{s.t.       } \ \ \ \
&\Sigma_x \succeq {\bf 0}, \Sigma_w \succeq {\bf 0}, \ V_x  \succeq{\bf 0}, \ V_w  \succeq{\bf 0},\\
&\left[\begin{array} {cc}
\hat{\Sigma}_x^{\frac{1}{2}} \Sigma_x \hat{\Sigma}_x^{\frac{1}{2}} &
V_x \\ V_x &I_n
\end{array}
\right]\succeq {\bf 0}, \quad
\left[\begin{array} {cc}
\hat{\Sigma}_w^{\frac{1}{2}} \Sigma_w \hat{\Sigma}_w^{\frac{1}{2}} &
V_w \\ V_w &I_m
\end{array}
\right]\succeq {\bf 0},\\
&\mbox{Tr}\left( \Sigma_x+\hat{\Sigma}_x-2V_x \right) \succeq {\bf 0}, \quad
\mbox{Tr}\left( \Sigma_w+\hat{\Sigma}_w-2V_w \right) \succeq {\bf 0},\\
&\begin{bmatrix}
U & \Sigma_x H^\top \\ H \Sigma & H \Sigma_x H^\top+\Sigma_w
\end{bmatrix} \succeq {\bf 0}.
\end{aligned} \label{maxmin SDP}
\end{equation} 
In the following theorem, we will demonstrate that under a mild condition, the condition in Lemma \ref{snsd} can be simplified, with all parameters determined through the convex problem (\ref{maxmin}) and its dual.

\begin{theorem} \label{snsdu}
	Suppose that Assumption 2 holds
	and problem (\ref{maxmin}) has a primal and dual optimal solution pair
	$(\Sigma_x^*, \Sigma_w^*,\gamma_x^*, \gamma_w^*)$.
	Define $A^* \triangleq \Sigma_x^* H^\top \left( H \Sigma_x^* H^\top+\Sigma_w^*\right) ^{-1}$ and 
	$K^* \triangleq I_n - A^* H$. 
	Then the saddle point solution of (\ref{minmax}) exists if and only if 
	$\begin{bmatrix}
	(K^*)^\top K^* - \gamma_x^*I_n & (K^*)^\top A^* \\
	(A^*)^\top K^* & (A^*)^\top A^* - \gamma_w^*I_m
	\end{bmatrix} \preceq {\bf{0}}$.
\end{theorem}

\begin{proof}
	Under the assumption that the nominal covariance matrices $\hat{\Sigma}_x$ and  $\hat{\Sigma}_w$ are positive definite, it follows from Theorem 3.1 in \cite{nguyen2023bridging} that (\ref{maxmin}) admits an optimal solution $(\Sigma_x^*,\Sigma_w^*)$ with 
	\begin{equation} \Sigma_x^* \succeq \lambda_{\min}\left( \hat{\Sigma}_x\right)  I_n \succ {\bf 0} \label{pdsx} \end{equation} 
	and 
	\begin{equation} \Sigma_w^* \succeq \lambda_{\min}\left( \hat{\Sigma}_w\right)  I_m \succ {\bf 0}, \label{pdsw} \end{equation} 
	which implies that $H \Sigma_x^* H^\top+\Sigma_w^* \succ {\bf 0}$. Then $A^*=\Sigma_x^* H^\top \left( H \Sigma_x^* H^\top+\Sigma_w^*\right) ^{-1}$ is well-defined.
	
	Subsequently, we shall demonstrate that $(\Sigma_x^*, \Sigma_w^*, A^*,\gamma_x^*, \gamma_w^*)$ is a saddle point solution of (\ref{maxminL}).
	We denote the Lagrangian function of (\ref{maxmin}) as
	\begin{align}
	g(\Sigma_x,\Sigma_w;\gamma_x,\gamma_w) \triangleq &
	\mbox{Tr}\left[\Sigma_x - \Sigma_x H^\top \left( H \Sigma_x H^\top +
	\Sigma_w \right)^{-1} H \Sigma_x\right] \notag \\
	&+\gamma_x \left\{ \rho_x^2-
	\mbox{Tr}\left[\Sigma_x + \hat{\Sigma}_x - 2 \left( \hat{\Sigma}_x^{\frac{1}{2}}\Sigma_x\hat{\Sigma}_x^{\frac{1}{2}} \right)^{\frac{1}{2}}\right] \right\} 
	+\gamma_w \left\{ \rho_w^2-
	\mbox{Tr}\left[\Sigma_w + \hat{\Sigma}_w - 2 \left( \hat{\Sigma}_w^{\frac{1}{2}}\Sigma_w\hat{\Sigma}_w^{\frac{1}{2}} \right)^{\frac{1}{2}}\right] \right\} \notag \\
	=& \inf_A \ \left\langle  \left(I_n-AH \right) ^\top \left(I_n-AH \right) ,\Sigma_{x} \right\rangle 
	+\left\langle  A^\top A,\Sigma_w \right\rangle  \notag \\
	&+\gamma_x \left\{ \rho_x^2-
	\mbox{Tr}\left[\Sigma_x + \hat{\Sigma}_x - 2 \left( \hat{\Sigma}_x^{\frac{1}{2}}\Sigma_x\hat{\Sigma}_x^{\frac{1}{2}} \right)^{\frac{1}{2}}\right] \right\} 
	+\gamma_w \left\{ \rho_w^2-
	\mbox{Tr}\left[\Sigma_w + \hat{\Sigma}_w - 2 \left( \hat{\Sigma}_w^{\frac{1}{2}}\Sigma_w\hat{\Sigma}_w^{\frac{1}{2}} \right)^{\frac{1}{2}}\right] \right\}  \notag \\
	=& \inf_A \ G(\Sigma_x,\Sigma_w, A,\gamma_x,\gamma_w), \label{regG}
	\end{align} 
	where $G$ is the objective function of (\ref{maxminL}) defined in (\ref{G}).
	According to the assumption that (\ref{maxmin}) has a primal and dual optimal solution pair
	$(\Sigma_x^*, \Sigma_w^*,\gamma_x^*, \gamma_w^*)$, we have \citep[Proposition 5.1.5]{bertsekas1997nonlinear}
	\begin{equation}
	g(\Sigma_x^*,\Sigma_w^*;\gamma_x^*,\gamma_w^*)=\max_{\Sigma_x \succeq {\bf 0},\Sigma_w \succeq {\bf 0}}  g(\Sigma_x,\Sigma_w;\gamma_x^*,\gamma_w^*)
	\label{gmmin} \end{equation}
	and the following complementary slackness conditions
	\begin{subequations}
	\begin{align}
	&\gamma_x^* \left\lbrace \rho_x^2-\mbox{Tr}\left[\Sigma_x^* + \hat{\Sigma}_x - 2 \left( \hat{\Sigma}_x^{\frac{1}{2}} \Sigma_x^* \hat{\Sigma}_x^{\frac{1}{2}} \right)^{\frac{1}{2}}\right] \right\rbrace =0,  \\
	&\gamma_w^* \left\lbrace \rho_w^2-\mbox{Tr}\left[\Sigma_w^* + \hat{\Sigma}_w - 2 \left( \hat{\Sigma}_w^{\frac{1}{2}} \Sigma_w^* \hat{\Sigma}_w^{\frac{1}{2}} \right)^{\frac{1}{2}}\right] \right\rbrace =0.
	\end{align}  \label{cs2}
	\end{subequations}
	Moreover, for the optimal solution $(\Sigma_x^*, \Sigma_w^*)$ to (\ref{maxmin}),
	$A^*=\Sigma_x^* H^\top \left( H \Sigma_x^* H^\top+\Sigma_w^*\right) ^{-1}$ is the unique minimizer of the inner minimization problem in (\ref{maxmin1}), which implies that $(\Sigma_x^*,\Sigma_w^*,A^*)$ is a saddle point solution of (\ref{maxmin1}).
	Combined with the equality of the optimal values of (\ref{maxmin1}) and (\ref{maxminL}) and the complementary slackness conditions (\ref{cs2}), it follows that $(\Sigma_x^*, \Sigma_w^*, A^*,\gamma_x^*, \gamma_w^*)$ is an optimal solution of (\ref{maxminL}).
	
	On the other hand, for given $(A^*,\gamma_x^*, \gamma_w^*)$, we obtain
	$$\bigtriangledown_{\Sigma_x} G(\Sigma_x^*,\Sigma_w^*,A^*,\gamma_x^*,\gamma_w^*)=
	\bigtriangledown_{\Sigma_x} g(\Sigma_x^*,\Sigma_w^*;\gamma_x^*,\gamma_w^*)=\bf{0},$$
	where the first equality is due to the relationship between $g$ and $G$ given by (\ref{regG}), Danskin's theorem \citep[Proposition B.25]{bertsekas1997nonlinear} and the uniqueness of the minimizer $A^*$, and the second equality follows from (\ref{gmmin}) and the positive definiteness of $\Sigma_x^*$ given by (\ref{pdsx}).
	Similarly, we have
	$$\bigtriangledown_{\Sigma_w} G(\Sigma_x^*,\Sigma_w^*,A^*,\gamma_x^*,\gamma_w^*)=
	\bigtriangledown_{\Sigma_w} g(\Sigma_x^*,\Sigma_w^*;\gamma_x^*,\gamma_w^*)=\bf{0}.$$
	Therefore, $(\Sigma_x^*,\Sigma_w^*)$ is the maximizer of $G(\Sigma_x,\Sigma_w,A^*,\gamma_x^*,\gamma_w^*)$, which implies that
	$(\Sigma_x^*, \Sigma_w^*, A^*,\gamma_x^*, \gamma_w^*)$ is a saddle point solution of (\ref{maxminL}).
	
	Finally, we prove by contradiction that (\ref{maxminL}) has a unique saddle point solution.
	If $(\bar{\Sigma}_x^*, \bar{\Sigma}_w^*, \bar{A}^*, \bar{\gamma}_x^*, \bar{\gamma}_w^*)$ is an another saddle point solution of (\ref{maxminL}),
	$(\Sigma_x^*, \Sigma_w^*, \bar{A}^*, \bar{\gamma}_x^*, \bar{\gamma}_w^*)$ is also a saddle point solution of (\ref{maxminL}) \citep[Theorem 6.2.9]{stoer2012convexity}.
	Then the uniqueness of the minimizer $A^*$ for $G(\Sigma_x^*,\Sigma_w^*, A,\gamma_x,\gamma_w)$ gives rise to $\bar{A}^*=A^*$.
	Furthermore, for given $A^*$, it follows from Proposition A.2 in \cite{nguyen2023bridging} that
	there exists a unique minimizer $(\gamma_x^*,\gamma_w^*)$ 
	and the unique optimal solution $(\Sigma_x^*, \Sigma_w^*)$ to the inner maximization problem in (\ref{maxminLe}).
	Then it contradicts the assumption that $(\bar{\Sigma}_x^*, \bar{\Sigma}_w^*,  \bar{\gamma}_x^*, \bar{\gamma}_w^*)$ is also an optimal solution to the inner minimax problem in (\ref{maxminLe}) for given $\bar{A}^*=A^*$.
	
	Consequently, the primal and dual optimal solution pair $(\Sigma_x^*, \Sigma_w^*,\gamma_x^*, \gamma_w^*)$ of (\ref{maxmin}) and the corresponding $A^*$ yield the unique saddle point solution of (\ref{maxminL}). Hence, it follows directly from Lemma \ref{snsd} that the saddle point solution of (\ref{minmax}) exists if and only if 
	$\begin{bmatrix}
	(K^*)^\top K^* - \gamma_x^*I_n & (K^*)^\top A^* \\
	(A^*)^\top K^* & (A^*)^\top A^* - \gamma_w^*I_m
	\end{bmatrix} \preceq {\bf{0}}$.
\end{proof}

\subsection{A Sufficient Condition for the Existence of the Saddle Point Solution for (\ref{minmax})}
Section 3.2 establishes a necessary and sufficient condition for the existence of the saddle point solution in (\ref{minmax}).  
However, verifying this condition requires determining both the primal and dual optimal solutions of (\ref{maxmin}), which may be computationally intractable for large-scale problems.
In this section, we will present a more direct and computationally simple sufficient condition for the existence of a saddle point solution to (\ref{minmax}), which intuitively indicates that when the Wasserstein radii $\rho_x$ and $\rho_w$ are small enough, the saddle point always exists.

\begin{theorem} \label{sd}
	Suppose that Assumption 2 holds.
	If the Wasserstein radii $\rho_x$ and $\rho_w$ satisfy the inequality 
	$\rho_x \rho_w \leq \lambda_{\min}^{\frac{1}{2}}(\hat{\Sigma}_x) \lambda_{\min}^{\frac{1}{2}}(\hat{\Sigma}_w)$, 
	%the saddle point of the original problem (\ref{minmax}) exists.
	problem (\ref{minmax}) has a saddle point solution.
\end{theorem}

\begin{proof}
	Under the assumption that the nominal covariance matrices $\hat{\Sigma}_x$ and $\hat{\Sigma}_w$ are both positive definite,
	we obtain from the proof of Theorem \ref{snsdu} that (\ref{maxminL}) admits a unique saddle point solution, denoted by $(\Sigma_x^*, \Sigma_w^*, A^*,\gamma_x^*, \gamma_w^*)$.  
	This solution is also optimal for (\ref{maxminLe}).
	Then it follows from Proposition A.2 in \cite{nguyen2023bridging} that 
    for given $A^*$ and $K^*=I_n-A^*H$, 
    the unique minimizer $(\gamma_x^*,\gamma_w^*)$ satisfies
	$\gamma_x^* > \lambda_{\max}\left( \left( K^*\right) ^\top K^* \right) $ and $\gamma_w^* > \lambda_{\max}\left( \left(A^*\right) ^\top A^* \right) $, 
	and the unique optimal solution $(\Sigma_x^*, \Sigma_w^*)$ to the inner maximization problem in (\ref{maxminLe}) is given by
	\begin{equation}
	\Sigma_x^*=\left( \gamma_x^*\right) ^2 
	\left[ \gamma_x^*I_n-\left( K^*\right) ^\top K^*\right] ^{-1}
	\hat{\Sigma}_x  \left[ \gamma_x^* I_n-\left( K^*\right) ^\top K^*\right] ^{-1}
	%\succ {\bf 0} 
	\label{opsx} \end{equation}
	and
	\begin{equation}
	\Sigma_w^*=\left( \gamma_w^*\right) ^2 
	\left[ \gamma_w^*I_m-\left( A^*\right) ^\top A^*\right] ^{-1}
	\hat{\Sigma}_w  \left[ \gamma_w^* I_m-\left( A^*\right) ^\top A^*\right] ^{-1}
	%\succ {\bf 0} 
	.\label{opsw} \end{equation} 
	
	Due to the strict positivity of $\gamma_x^*$ and $\gamma_w^*$, the
	Wasserstein distance constraints with respect to $(\Sigma_x,\Sigma_w)$  are active.
	%optimal solution $(\Sigma_x^*,\Sigma_w^*)$ is on the boundary of the Wasserstein distance constraints. 
	Then we have
	\begin{equation}
	\rho_x^2-\mbox{Tr}\left[\Sigma_x^* + \hat{\Sigma}_x - 2 \left( \hat{\Sigma}_x^{\frac{1}{2}} \Sigma_x^* \hat{\Sigma}_x^{\frac{1}{2}} \right)^{\frac{1}{2}}\right]=0 \label{bound1} 
	\end{equation}
	and
	\begin{equation}
	\rho_w^2-\mbox{Tr}\left[\Sigma_w^* + \hat{\Sigma}_w - 2 \left( \hat{\Sigma}_w^{\frac{1}{2}} \Sigma_w^* \hat{\Sigma}_w^{\frac{1}{2}} \right)^{\frac{1}{2}}\right]=0. \label{bound2}
	\end{equation}
	Substituting (\ref{opsx}) into (\ref{bound1}), we obtain
	$$ \rho_x^2-\mbox{Tr}\left( \hat{\Sigma}_x\right) +2\gamma_x^* \mbox{Tr}
	\left\lbrace \left[ \gamma_x^*I_n-\left( K^*\right) ^\top K^*\right] ^{-1} \hat{\Sigma}_x\right\rbrace -(\gamma_x^*)^2
	\mbox{Tr} \left\lbrace \left[ \gamma_x^*I_n-\left( K^*\right) ^\top K^*\right] ^{-1} \hat{\Sigma}_x \left[ \gamma_x^*I_n-\left( K^*\right) ^\top K^*\right] ^{-1} \right\rbrace =0.$$
	Rearranging all the terms, we derive
	\begin{align*}
	&\mbox{Tr} \left\lbrace \left[ \gamma_x^*I_n-\left( K^*\right) ^\top \hspace*{-1mm} K^*\right] ^{-1} \hspace*{-1mm} \left[ 
	\left( \gamma_x^*\right) ^2 \hat{\Sigma}_x 
	-\gamma_x^* \hat{\Sigma}_x \left( \gamma_x^*I_n-\left( K^*\right) ^\top \hspace*{-1mm} K^*\right) 
	-\gamma_x^* \left( \gamma_x^*I_n-\left( K^*\right) ^\top \hspace*{-1mm} K^*\right)  \hat{\Sigma}_x
	\right] \hspace*{-1mm}  \left[ \gamma_x^*I_n-\left( K^*\right) ^\top \hspace*{-1mm} K^*\right] ^{-1} \right\rbrace 
	\\&=\rho_x^2-\mbox{Tr}\left( \hat{\Sigma}_x\right). \end{align*}
	By expanding the brackets in the middle terms of the product under the trace and combining like terms, we get
	$$\mbox{Tr} \left\lbrace  \left[ \gamma_x^*I_n-\left( K^*\right) ^\top \hspace*{-1mm} K^*\right] ^{-1} \left[ 
	-\left( \gamma_x^*\right) ^2 \hat{\Sigma}_x 
	+\gamma_x^* \hat{\Sigma}_x \left( K^* \right) ^\top \hspace*{-1mm} K^*
	+\gamma_x^* \left( K^*\right) ^\top \hspace*{-1mm} K^* \hat{\Sigma}_x
	\right]  \left[ \gamma_x^*I_n-\left( K^*\right) ^\top \hspace*{-1mm} K^*\right] ^{-1} \right\rbrace 
	=\rho_x^2-\mbox{Tr}\left( \hat{\Sigma}_x\right) . $$
	Furthermore, completing the square for the middle terms in the product under the trace gives
	\begin{align*}
	&\mbox{Tr} \left\lbrace  
	\left[ \gamma_x^*I_n \hspace*{-1mm} - \hspace*{-1mm} \left( K^*\right)^\top \hspace*{-1mm} K^*\right] ^{-1} 
	\hspace*{-1mm} \left\lbrace \hspace*{-1mm} - \hspace*{-1mm} \left[ \gamma_x^*I_n \hspace*{-1mm}-\hspace*{-1mm} \left( K^*\right) ^\top \hspace*{-1mm} K^*\right]  
	\hspace*{-1mm} \hat{\Sigma}_x \hspace*{-1mm} \left[ \gamma_x^*I_n \hspace*{-1mm}-\hspace*{-1mm} \left( K^*\right) ^\top \hspace*{-1mm} K^*\right] 
	\hspace*{-1mm}+\hspace*{-0.5mm}
	\left( K^*\right) ^\top \hspace*{-1mm} K^* \hat{\Sigma}_x (K^*)^\top \hspace*{-1mm} K^* \right\rbrace \hspace*{-1mm}  
	\left[ \gamma_x^*I_n \hspace*{-1mm}-\hspace*{-1mm} \left( K^*\right) ^\top \hspace*{-1mm} K^*\right] ^{-1} \right\rbrace \notag \\ &=
	\rho_x^2\hspace*{-1mm}-\hspace*{-1mm}\mbox{Tr}\left( \hat{\Sigma}_x\right) . \end{align*}
	Eliminating $-\mbox{Tr}(\hat{\Sigma}_x)$ which appears on both sides of the equation, we obtain 
	$$\mbox{Tr} \left\lbrace
	\left[ \gamma_x^*I_n-\left( K^*\right) ^\top K^*\right] ^{-1} 
	\left[ \left( K^*\right) ^\top K^* \hat{\Sigma}_x \left( K^*\right) ^\top K^* \right]  
	\left[ \gamma_x^*I_n-\left( K^*\right) ^\top K^*\right] ^{-1} \right\rbrace 
	=\rho_x^2.$$
	Upon dividing both sides by $\rho_x^2$, we have
	\begin{equation*} 
	\mbox{Tr} \left\lbrace \left[ \gamma_x^*I_n-\left( K^*\right) ^\top K^*\right] ^{-1} \left[ \left( K^*\right) ^\top K^* \frac{\hat{\Sigma}_x}{\rho_x^2} \left( K^*\right) ^\top K^*
	\right]  \left[ \gamma_x^* I_n-\left( K^*\right) ^\top K^*\right] ^{-1} \right\rbrace  =1.
	\end{equation*}
	Then the matrix under the trace operation is positive semidefinite and the sum of its eigenvalues equal to one. Consequently, all its eigenvalues are not greater than one, which leads to the conclusion that   
	\begin{equation} 
	\left[ \gamma_x^*I_n-\left( K^*\right) ^\top K^*\right] ^{-1} \left[ 
	\left( K^*\right) ^\top K^* \frac{\hat{\Sigma}_x}{\rho_x^2} \left( K^*\right) ^\top K^*
	\right]  \left[ \gamma_x^* I_n-\left( K^*\right) ^\top K^*\right] ^{-1} \preceq I_n.
	\label{loose1} \end{equation}
	Multiplying both sides by $\gamma_x^*I_n-(K^*)^\top K^*$ gives rise to
	$$\left[ \gamma_x^*I_n-\left( K^*\right) ^\top K^*\right] ^{2} \succeq
	\left( K^*\right) ^\top K^* \frac{\hat{\Sigma}_x}{\rho_x^2} \left( K^*\right) ^\top K^*.$$
    Due to $\gamma_x^*>\lambda_{\max}\left[ \left( K^*\right) ^\top K^* \right]$, the matrix $\gamma_x^*I_n-\left( K^*\right) ^\top K^*$ is positive definite. 
    According to the Löwner-Heinz inequality \citep[Theorem 1.1]{zhan2004matrix}, it follows that
	\begin{equation*} 
	\gamma_x^*I_n-\left( K^*\right) ^\top K^* \succeq
	\left[ \left( K^*\right) ^\top K^* \frac{\hat{\Sigma}_x}{\rho_x^2} \left( K^*\right) ^\top K^* \right] ^\frac{1}{2}.  \end{equation*}
	Combining with $\hat{\Sigma}_x \succeq \lambda_{\min}\left( \hat{\Sigma}_x\right) I_n$, we further have
	\begin{equation} 
	\gamma_x^*I_n-\left( K^*\right) ^\top K^* \succeq
	\frac{\lambda_{\min}^{\frac{1}{2}}\left( \hat{\Sigma}_x\right) }{\rho_x}
	\left( K^*\right) ^\top K^*. 
	\label{loose2}  \end{equation}

	Similarly, based on (\ref{opsw}) and (\ref{bound2}), we derive
	\begin{equation}
	\gamma_w^*I_m-\left( A^*\right) ^\top A^* \succeq
	\frac{\lambda_{\min}^{\frac{1}{2}}\left( \hat{\Sigma}_w\right) }{\rho_w}
	\left( A^*\right) ^\top A^*. 
	\label{loosew} \end{equation}
	
	Subsequently, by utilizing equations (\ref{loose2}) and (\ref{loosew}), we deduce that:
	$$ \begin{aligned}
	\begin{bmatrix}
	\left( K^*\right)^\top K^*-\gamma_x^*I_n & \left( K^*\right) ^\top A^* \\
	\left( A^*\right) ^\top K^* & \left( A^*\right) ^\top A^* -\gamma_w^*I_m
	\end{bmatrix} & \preceq
	\begin{bmatrix}
	-\frac{\lambda_{\min}^{\frac{1}{2}}\left( \hat{\Sigma}_x\right) }{\rho_x} \left( K^*\right) ^\top K^* 
	& \left( K^*\right)^\top A^* \\ \left(A^*\right)^\top \left( K^*\right) &
	- \frac{\lambda_{\min}^{\frac{1}{2}}\left( \hat{\Sigma}_w\right) }{\rho_w}
	\left( A^*\right) ^\top A^*
	\end{bmatrix}\\
	&=\begin{bmatrix}
	K^* & {\bf{0}} \\ {\bf{0}} & A^*
	\end{bmatrix} ^\top
	\begin{bmatrix}
	-\frac{\lambda_{\min}^{\frac{1}{2}}\left( \hat{\Sigma}_x\right) }{\rho_x}I_n
	&I_n \\I_n & 
	- \frac{\lambda_{\min}^{\frac{1}{2}}\left( \hat{\Sigma}_w\right) }{\rho_w}I_n
	\end{bmatrix}
	\begin{bmatrix}
	K^* & {\bf{0}} \\ {\bf{0}} & A^*
	\end{bmatrix}.
	\end{aligned} $$
	Here, the matrix $\begin{bmatrix}
	-\frac{\lambda_{\min}^{\frac{1}{2}}\left( \hat{\Sigma}_x\right) }{\rho_x}I_n
	&I_n \\I_n & 
	- \frac{\lambda_{\min}^{\frac{1}{2}}\left( \hat{\Sigma}_w\right) }{\rho_w}I_n \end{bmatrix}$ 
	is negative semidefinite if and only if its schur complement $-\frac{\lambda_{\min}^{\frac{1}{2}}\left( \hat{\Sigma}_x\right) }{\rho_x}I_n+\frac{\rho_w}{\lambda_{\min}^{\frac{1}{2}}\left( \hat{\Sigma}_w\right) }I_n$ is negative semidefinite \cite{zhang2006schur}.
	Consequently, if $$\rho_x \rho_w \leq \lambda_{\min}^{\frac{1}{2}}\left( \hat{\Sigma}_x\right) 
	\lambda_{\min}^{\frac{1}{2}}\left( \hat{\Sigma}_w\right) ,$$
	the matrix $\begin{bmatrix}
	\left(K^*\right) ^\top K^*-\gamma_x^*I_n & \left(K^*\right) ^\top A^* \\
	\left(A^*\right) ^\top K^* & \left(A^*\right)^\top A^* -\gamma_w^*I_m
	\end{bmatrix}$ is negative semidefinite, which implies that 
	the saddle point solution of (\ref{minmax}) exists by Lemma \ref{snsd}.
\end{proof}

\begin{remark}
	Based on the aforementioned sufficient condition, we can easily obtain the following conclusions:
	\begin{itemize}
		\item In the scalar case, i.e., when both the parameter $x$ and the observation $y$ are scalars, the relaxations (\ref{loose1}) and (\ref{loose2}) are tight. This indicates that Theorem \ref{sd} provides a necessary and sufficient condition for the existence of the saddle point solution for (\ref{minmax}).
		\item The sufficient condition given in Theorem \ref{sd} constrains the product of the two Wasserstein radii, which implies that when one is sufficiently small, the other can be larger.
		Furthermore, if one of the Wasserstein radii is zero, i.e.,  
		$\rho_x \hspace*{-1mm}=\hspace*{-0.5mm} 0$ or 
		$\rho_w \hspace*{-1mm}=\hspace*{-0.5mm} 0$, 
		a saddle point solution to (\ref{minmax}) always exists, regardless of the size of the other radius.
		%\item When the radii of the uncertainty set $\rho_x$ and $\rho_w$ are sufficiently small, a saddle point always exists.
	\end{itemize}
\end{remark}

\section{The Robust Linear Estimator}
In Section 3.1, Theorem \ref{equp} shows that the absence of the saddle point solution in the finite-dimensional optimization problem (\ref{minmaxG}) implies the absence of the saddle point solution in the infinite-dimensional optimization problem (\ref{minmax}). 
In this case, although the exact optimal solution of (\ref{minmax}) is unavailable by solving (\ref{minmaxG}), the optimal value of (\ref{minmaxG}) can still provide an upper bound on that of (\ref{minmax}), as detailed in Lemma \ref{sort}.
Furthermore, due to (\ref{rlr}), the optimal solution to (\ref{minmaxG}) is also the optimal solution to 
\begin{equation} \label{infLsup}
\inf_{f \in \mathcal{F_L}} 
\sup_{P \in \mathbb{B}(\hat{P})} \mbox{mse}(f,P),
\end{equation} 
which indicates that this solution yields an optimal robust estimator in the class of linear estimators.

Consequently, the focus of this section is on problem (\ref{minmaxG}).
Specifically, we demonstrate that (\ref{minmaxG}) is equivalent to a convex relaxation problem, and thus equivalent to an SDP problem. Furthermore, based on the primal and dual optimal solutions of the SDP problem, we construct an optimal solution to (\ref{minmaxG}), which is also the optimal solution to (\ref{infLsup}) and thus provides a robust linear estimator.

\subsection{A Tight Convex Relaxation of (\ref{minmaxG})}
In light of the above discussion, we now focus on (\ref{minmaxG}). Analogous to parameterizing (\ref{maxmina}) into (\ref{maxminap}) in Section 3.2, problem (\ref{minmaxG}) can be parameterized into a finite-dimensional form as follows:
\begin{align}
\inf_{A,b}
\sup_{\mu_x,\mu_w,\Sigma_x,\Sigma_w}
&\mbox{Tr}\left[ \left( AH-I_n\right) \Sigma_x\left( AH-I_n\right) ^\top+A\Sigma_w A^\top \right]+\left[ \left( AH-I_n\right) \mu_x+A\mu_w+b\right] ^\top \left[ \left( AH-I_n\right) \mu_x+A\mu_w+b \right] \notag  \\
\mbox{s.t. } \
& \Sigma_x \succeq{\bf 0}, \quad \Sigma_w \succeq {\bf 0}, \notag \\
&\left\| \mu_x-\hat{\mu}_x \right\| ^2+\mbox{Tr} \left[ \Sigma_x+\hat{\Sigma}_x-
2 \left( \hat{\Sigma}_x^{\frac{1}{2}} \Sigma_x \hat{\Sigma}_x^{\frac{1}{2}}\right) ^{\frac{1}{2}} \right] \leq \rho_x^2, \label{minmaxGG} \\
&\left\| \mu_w-\hat{\mu}_w\right\| ^2+\mbox{Tr} \left[ \Sigma_w+\hat{\Sigma}_w-
2 \left( \hat{\Sigma}_w^{\frac{1}{2}} \Sigma_w \hat{\Sigma}_w^{\frac{1}{2}}\right) ^{\frac{1}{2}} \right] 
\leq \rho_w^2. \notag 
\end{align}
For given $A$ and $b$, the objective function of (\ref{minmaxGG}) is convex in $(\mu_x,\mu_w)$. 
Consequently, the inner maximization problem over $(\mu_x,\mu_w)$ is non-convex, which brings a significant challenge in identifying the optimal solution.
For ease of notations, we denote
$\tilde{\mu}_x \triangleq \mu_x-\hat{\mu}_x$, 
$\tilde{\mu}_w \triangleq \mu_w-\hat{\mu}_w$
and $\tilde{b} \triangleq b+(AH-I_n)\hat{\mu}_x+A\hat{\mu}_w$, respectively.
Then problem (\ref{minmaxGG}) can be reformulated as follows
\begin{align} 
\inf_{A,\tilde{b}} \
\sup_{\tilde{\mu}_x,\tilde{\mu}_w,\Sigma_x,\Sigma_w}
&\mbox{Tr}\left[ \left( AH-I_n\right) \Sigma_x\left( AH-I_n\right) ^\top+A\Sigma_w A^\top \right]
+\left[\left( AH-I_n\right) \tilde{\mu}_x+A\tilde{\mu}_w+\tilde{b}\right] ^\top \left[\left( AH-I_n \right) \tilde{\mu}_x+A\tilde{\mu}_w+\tilde{b}\right] \notag  \\
\mbox{s.t. } \
& \Sigma_x \succeq{\bf 0}, \quad \Sigma_w \succeq {\bf 0}, \notag \\
&\left\| \tilde{\mu}_x\right\| ^2+\mbox{Tr} \left[ \Sigma_x+\hat{\Sigma}_x-
2 \left( \hat{\Sigma}_x^{\frac{1}{2}} \Sigma_x \hat{\Sigma}_x^{\frac{1}{2}}\right) ^{\frac{1}{2}} \right] \leq \rho_x^2, 
\label{3minimax} \\
&\left\| \tilde{\mu}_w\right\| ^2+\mbox{Tr} \left[ \Sigma_w+\hat{\Sigma}_w-
2 \left( \hat{\Sigma}_w^{\frac{1}{2}} \Sigma_w \hat{\Sigma}_w^{\frac{1}{2}}\right) ^{\frac{1}{2}} \right] 
\leq \rho_w^2. \notag  
\end{align}
Notice that for given $A$, $\tilde{b}$, $\Sigma_x$ and $\Sigma_w$, the inner maximization problem over $(\tilde{\mu}_x,\tilde{\mu}_w)$ in (\ref{3minimax}) is a nonhomogeneous QCQP problem with two homogeneous constraints.
If we directly apply the SDP relaxation method, we can not assert that the relaxed problem is equivalent to problem (\ref{3minimax}) \cite{ye2003new,ai2024tightness}.
To address this difficulty,  we first give a special structure of the optimal solution to (\ref{3minimax}) in the following theorem.

\begin{theorem} \label{minimax oss}
	If $(A^*,\tilde{b}^*,\tilde{\mu}^*_x,\tilde{\mu}^*_w,\Sigma^*_x,\Sigma^*_w)$ is an optimal solution to (\ref{3minimax}), then $\tilde{b}^*={\bf 0}$.
\end{theorem}
\begin{proof}
	For the sake of simplicity of notations, we denote the objective function of (\ref{3minimax}) as
	\begin{align*}
	\psi(A,\tilde{b},\tilde{\mu}_x,\tilde{\mu}_w,\Sigma_x,\Sigma_w)
	=&\mbox{Tr}\left[ \left( AH-I_n\right) \Sigma_x\left( AH-I_n\right) ^\top+A\Sigma_w A^\top \right]+  \\ 
	&\left[\left( AH-I_n\right) \tilde{\mu}_x+A\tilde{\mu}_w\right] ^\top \left[\left( AH-I_n\right) \tilde{\mu}_x+A\tilde{\mu}_w\right]+
	2\left[\left( AH-I_n\right) \tilde{\mu}_x+A\tilde{\mu}_w\right] ^\top \tilde{b}
	%+\tilde{b}^\top \left[(AH-I_n)\tilde{\mu}_x+A\tilde{\mu}_w\right]
	+\tilde{b}^\top \tilde{b}.
	\end{align*}
	In addition, for given $A$ and $\tilde{b}$, the optimal solution to the inner maximization problem in (\ref{3minimax}) is denoted as $(\tilde{\mu}_x^*(A,\tilde{b}),\tilde{\mu}_w^*(A,\tilde{b}),\Sigma_x^*(A,\tilde{b}),\Sigma_w^*(A,\tilde{b}))$. Then for any $A$ and $\tilde{b} \neq {\bf 0}$, we have
	\begin{align*}
	&\psi(A,{\bf 0},\tilde{\mu}_x^*(A,{\bf 0}),\tilde{\mu}_w^*(A,{\bf 0}),\Sigma_x^*(A,{\bf 0}),\Sigma_w^*(A,{\bf 0})) \\
	&< \max \left\lbrace  
	\psi(A,\tilde{b},\tilde{\mu}_x^*(A,{\bf 0}),\tilde{\mu}_w^*(A,{\bf 0}),\Sigma_x^*(A,{\bf 0}),\Sigma_w^*(A,{\bf 0})),
	\psi(A,\tilde{b},-\tilde{\mu}_x^*(A,{\bf 0}),-\tilde{\mu}_w^*(A,{\bf 0}),\Sigma_x^*(A,{\bf 0}),\Sigma_w^*(A,{\bf 0})) \right\rbrace \\
	&\leq \psi\left( A,\tilde{b},\tilde{\mu}_x^*(A,\tilde{b}),\tilde{\mu}_w^*(A,\tilde{b}),\Sigma_x^*(A,\tilde{b}),\Sigma_w^*(A,\tilde{b})\right) ,
	\end{align*}
	where the first inequality holds because:
	i) the first two terms of the objective function $\psi$ remain unchanged;
	ii) $(\tilde{\mu}_x^*(A,{\bf 0}),\tilde{\mu}_w^*(A,{\bf 0}),\Sigma_x^*(A,{\bf 0}),\Sigma_w^*(A,{\bf 0}))$ is a feasible solution and so is $(-\tilde{\mu}_x^*(A,{\bf 0}),-\tilde{\mu}_w^*(A,{\bf 0}),\Sigma_x^*(A,{\bf 0}),\Sigma_w^*(A,{\bf 0}))$;
	iii) at least one of these two feasible solutions satisfies
	$\left[(AH-I_n)\tilde{\mu}_x+A\tilde{\mu}_w\right] ^\top \tilde{b} \geq 0$;
	iv) $\tilde{b}^\top \tilde{b}>0$,
	and the second inequality follows from the optimality of $(\tilde{\mu}_x^*(A,\tilde{b}),\tilde{\mu}_w^*(A,\tilde{b}),\Sigma_x^*(A,\tilde{b}),\Sigma_w^*(A,\tilde{b}))$ for $A$ and $\tilde{b}$.
	Therefore, the optimal solution to (\ref{3minimax}) must satisfy $\tilde{b}^*=\bf 0$.
\end{proof}

With the help of Theorem \ref{minimax oss}, (\ref{3minimax}) can be equivalently transformed into the following problem:
\begin{align}
\inf_{A}
\sup_{\tilde{\mu}_x,\tilde{\mu}_w,\Sigma_x,\Sigma_w}
&\mbox{Tr}\left[ \left( AH-I_n\right) \Sigma_x\left( AH-I_n\right) ^\top+A\Sigma_w A^\top \right]
+\left[\left( AH-I_n\right) \tilde{\mu}_x+A\tilde{\mu}_w \right] ^\top \left[\left( AH-I_n\right) \tilde{\mu}_x+A\tilde{\mu}_w \right] \notag  \\
\mbox{s.t. } \
& \Sigma_x \succeq{\bf 0}, \quad \Sigma_w \succeq {\bf 0}, \notag \\
&\left\| \tilde{\mu}_x\right\| ^2+\mbox{Tr} \left[ \Sigma_x+\hat{\Sigma}_x-
2 \left( \hat{\Sigma}_x^{\frac{1}{2}} \Sigma_x \hat{\Sigma}_x^{\frac{1}{2}}\right) ^{\frac{1}{2}} \right] \leq \rho_x^2, \label{2minimax} \\
&\left\| \tilde{\mu}_w\right\| ^2+\mbox{Tr} \left[ \Sigma_w+\hat{\Sigma}_w-
2 \left( \hat{\Sigma}_w^{\frac{1}{2}} \Sigma_w \hat{\Sigma}_w^{\frac{1}{2}}\right) ^{\frac{1}{2}} \right] 
\leq \rho_w^2. \notag 
\end{align} 
Notice that for given $A$, $\Sigma_x$ and $\Sigma_w$, the inner maximization problem over $(\tilde{\mu}_x,\tilde{\mu}_w)$ in (\ref{2minimax}) is a homogeneous QCQP problem with two homogeneous constraints, which can be be equivalently reformulated as a convex problem through the SDP relaxation method as described in the following theorem.

\begin{theorem} \label{minmax eq}
	The optimal value of (\ref{3minimax}) is equal to the optimal value of the following SDP problem
	\begin{align}
	\begin{array}{cl}
	\sup \limits_{Q,S,\Sigma_x,\Sigma_w,V_x,V_w}
	&\mbox{Tr}(Q) \\
	\mbox{s.t.}&  S \succeq{\bf 0}, \ \Sigma_x \succeq{\bf 0}, \ \Sigma_w  \succeq{\bf 0}, \ V_x  \succeq{\bf 0},
	\ V_w  \succeq{\bf 0},\\
	&\begin{bmatrix} I_n & {\bf 0} \\H & I_m \end{bmatrix} 
	\left( 
	\begin{bmatrix} \Sigma_x & {\bf 0} \\ {\bf 0} & \Sigma_w \end{bmatrix}+S
	\right) 
	\begin{bmatrix} I_n & H^\top \\ {\bf 0} & I_m \end{bmatrix}-\begin{bmatrix}
	Q & {\bf 0} \\ {\bf 0} & {\bf 0}
	\end{bmatrix} \succeq {\bf 0}\\
	&\left[\begin{array} {cc}
	\hat{\Sigma}_x^{\frac{1}{2}} \Sigma_x \hat{\Sigma}_x^{\frac{1}{2}} &
	V_x \\ V_x &I_n
	\end{array}
	\right]\succeq {\bf 0}, \quad
	\left[\begin{array} {cc}
	\hat{\Sigma}_w^{\frac{1}{2}} \Sigma_w \hat{\Sigma}_w^{\frac{1}{2}} &
	V_w \\ V_w &I_m
	\end{array}
	\right]\succeq {\bf 0},\\
	& \mbox{Tr}\left(\begin{bmatrix}
	I_n & \boldsymbol{0} \\ \boldsymbol{0} & \boldsymbol{0} 
	\end{bmatrix} S \right) 
	+\mbox{Tr}\left( \Sigma_x + \hat{\Sigma}_x - 2 V_x \right) 
	\leq \rho_x^2 ,\\
	&\mbox{Tr} \left( \begin{bmatrix}
	\boldsymbol{0} & \boldsymbol{0} \\ \boldsymbol{0} &I_m
	\end{bmatrix} S \right)
	+\mbox{Tr}\left(\Sigma_w + \hat{\Sigma}_w - 2 V_w\right)  \leq \rho_w^2.
	\end{array} \label{minmax SDP}
	\end{align}
\end{theorem}
\begin{proof}
	Since
	\begin{align*}
	\left\lbrace \left.
	\begin{bmatrix}
	\tilde{\mu}_x \\ \tilde{\mu}_x
	\end{bmatrix} 
	\begin{bmatrix}
	\tilde{\mu}_x^\top & \tilde{\mu}_x^\top
	\end{bmatrix} \right| 
	\tilde{\mu}_x \in \mathbb{R}^n, 
	\tilde{\mu}_w \in \mathbb{R}^m
	\right\rbrace 
	=\left\lbrace S \left| S \in \mathbb{S}^{m+n}_+, \mbox{rank}(S)=1 
	\right. \right\rbrace 
	\subseteq \left\lbrace S \left| S \in \mathbb{S}^{m+n}_+
	\right. \right\rbrace, 
	\end{align*}
	then, by ignoring the constraints $\mbox{rank}(S)=1$ and thus relaxing the inner maximization problem of (\ref{2minimax}), we obtain the following relaxation of (\ref{2minimax}):
	%we consider the following relaxation problem of (\ref{2minimax})
	\begin{align}\label{least SDPs}
	\begin{array}{cl}
	\inf \limits_{A}
	\sup \limits_{S,\Sigma_x,\Sigma_w} 
	&\mbox{Tr}\left[\left(AH-I_n\right)\Sigma_x\left(AH-I_n\right)^\top
	+A \Sigma_w A^\top\right]+
	\mbox{Tr} \left(\begin{bmatrix} (AH-I_n)^\top \\  A^\top  \end{bmatrix}
	\begin{bmatrix} AH-I_n & A \end{bmatrix} S \right) \\
	\mbox{s.t.}& S \succeq{\bf 0}, \ \Sigma_x \succeq{\bf 0}, \ \Sigma_w  \succeq{\bf 0},\\
	& \mbox{Tr} \left(\begin{bmatrix}
	I_n & \boldsymbol{0} \\ \boldsymbol{0} & \boldsymbol{0} 
	\end{bmatrix} S \right) 
	+ \mbox{Tr} \left[ \Sigma_x+\hat{\Sigma}_x-
	2 \left( \hat{\Sigma}_x^{\frac{1}{2}} \Sigma_x \hat{\Sigma}_x^{\frac{1}{2}}\right) ^{\frac{1}{2}} \right]  \leq \rho_x^2,\\
	&\mbox{Tr} \left(\begin{bmatrix}
	\boldsymbol{0} & \boldsymbol{0} \\ \boldsymbol{0} &I_m
	\end{bmatrix} S \right)
	+ \mbox{Tr} \left[ \Sigma_w+\hat{\Sigma}_w-
	2 \left( \hat{\Sigma}_w^{\frac{1}{2}} \Sigma_w \hat{\Sigma}_w^{\frac{1}{2}}\right) ^{\frac{1}{2}}\right]  \leq \rho_w^2.
	\end{array}
	\end{align}
	Notice that for any $A$, $\Sigma_x$ and $\Sigma_w$, the inner maximization problem on $S$ in (\ref{least SDPs}) always admits a rank-one optimal solution \citep[Theorem 2.5]{ye2003new}. This implies that the relaxation is tight, i.e., problems (\ref{2minimax}) and (\ref{least SDPs}) are equivalent. %and is equivalent to the inner maximization problem on $\begin{bmatrix} \tilde{\mu}_x \\ \tilde{\mu}_w \end{bmatrix}$ in (\ref{2minimax}) \cite{ye2003new},.
	
	Note that the objective function of (\ref{least SDPs}) is convex in $A$ for each $(S,\Sigma_x,\Sigma_w)$, concave in $(S,\Sigma_x,\Sigma_w)$ for each $A$, and the constraint set under the supremum is convex and compact. Then according to Sion's minimax theorem \cite{sion1958general},
	(\ref{least SDPs}) is equivalent to the following problem obtained by exchanging minimization and maximization:
	\begin{align}\label{least SDPss}
	\begin{array}{cl}
	\sup \limits_{S,\Sigma_x,\Sigma_w} \inf \limits_{A}
	&\mbox{Tr}\left[\left(AH-I_n\right)\Sigma_x\left(AH-I_n\right)^\top
	+A\Sigma_w A^\top\right]+
	\mbox{Tr} \left(\begin{bmatrix} \left( AH-I_n\right) ^\top \\  A^\top  \end{bmatrix}
	\begin{bmatrix} AH-I_n & A \end{bmatrix} S \right) \\
	\mbox{s.t.}& S \succeq{\bf 0}, \ \Sigma_x \succeq{\bf 0}, \ \Sigma_w  \succeq{\bf 0},\\
	& \mbox{Tr} \left(\begin{bmatrix}
	I_n & \boldsymbol{0} \\ \boldsymbol{0} & \boldsymbol{0} 
	\end{bmatrix} S \right)
	+ \mbox{Tr} \left[ \Sigma_x+\hat{\Sigma}_x-
	2 \left( \hat{\Sigma}_x^{\frac{1}{2}} \Sigma_x \hat{\Sigma}_x^{\frac{1}{2}}\right) ^{\frac{1}{2}}\right]  \leq \rho_x^2,\\
	&\mbox{Tr} \left(\begin{bmatrix}
	\boldsymbol{0} & \boldsymbol{0} \\ \boldsymbol{0} &I_m
	\end{bmatrix} S \right)
	+ \mbox{Tr} \left[ \Sigma_w+\hat{\Sigma}_w-
	2 \left( \hat{\Sigma}_w^{\frac{1}{2}} \Sigma_w \hat{\Sigma}_w^{\frac{1}{2}}\right) ^{\frac{1}{2}} \right]  \leq \rho_w^2.
	\end{array}
	\end{align}
	The objective function of (\ref{least SDPss}) can be reformulated as
	\begin{align*}
	&\mbox{Tr} \left( \begin{bmatrix} AH-I_n & A \end{bmatrix} 
	\left( 
	\begin{bmatrix} \Sigma_x & {\bf 0} \\ {\bf 0} & \Sigma_w \end{bmatrix}+S
	\right) 
	\begin{bmatrix} (AH-I_n)^\top \\  A^\top  \end{bmatrix} \right) \\
	&=\mbox{Tr} \left( \begin{bmatrix} I_n & -A \end{bmatrix}
	\begin{bmatrix} I_n & {\bf 0} \\H & I_m \end{bmatrix} 
	\left( 
	\begin{bmatrix} \Sigma_x & {\bf 0} \\ {\bf 0} & \Sigma_w \end{bmatrix}+S
	\right) 
	\begin{bmatrix} I_n & H^\top \\ {\bf 0} & I_m \end{bmatrix} 
	\begin{bmatrix} I_n \\  -A^\top  \end{bmatrix} \right).
	\end{align*}
	Then since the constraints of (\ref{least SDPss}) are independent of $A$, we first consider the inner unconstrained minimization problem
	\begin{equation}
	\inf \limits_{A} \hspace*{1mm}
	\mbox{Tr} \left( \begin{bmatrix} I_n & -A \end{bmatrix}
	\begin{bmatrix} I_n & {\bf 0} \\H & I_m \end{bmatrix} 
	\left( 
	\begin{bmatrix} \Sigma_x & {\bf 0} \\ {\bf 0} & \Sigma_w \end{bmatrix}+S
	\right) 
	\begin{bmatrix} I_n & H^\top \\ {\bf 0} & I_m \end{bmatrix} 
	\begin{bmatrix} I_n \\  -A^\top  \end{bmatrix} \right).
	\label{least SDP min} \end{equation}
	Let
	\begin{align*}
	R \triangleq \begin{bmatrix} R_{11} & R_{12} \\ R_{12}^\top & R_{22} \end{bmatrix} \triangleq
	\begin{bmatrix} I_n & {\bf 0} \\H & I_m \end{bmatrix} 
	\left( 
	\begin{bmatrix} \Sigma_x & {\bf 0} \\ {\bf 0} & \Sigma_w \end{bmatrix}+S
	\right) 
	\begin{bmatrix} I_n & H^\top \\ {\bf 0} & I_m \end{bmatrix}. \end{align*}
	It is well-known that the optimal value of (\ref{least SDP min}) is  $\mbox{Tr}\left( R_{11}-R_{12} R_{22}^\dagger R_{12}^\top\right) $ \citep[Exercise 9.2.4.1]{Albert1972Regression}.
	Consequently, (\ref{least SDPss}) is equivalent to
	\begin{align}
	\begin{array}{cl}
	\sup \limits_{S,\Sigma_x,\Sigma_w}
	&\mbox{Tr}\left( R_{11}-R_{12} R_{22}^\dagger R_{12}^\top\right) \\
	\mbox{s.t.}& S \succeq{\bf 0}, \ \Sigma_x \succeq{\bf 0}, \ \Sigma_w  \succeq{\bf 0},\\
	&R=\begin{bmatrix} I_n & {\bf 0} \\H & I_m \end{bmatrix} 
	\left( 
	\begin{bmatrix} \Sigma_x & {\bf 0} \\ {\bf 0} & \Sigma_w \end{bmatrix}+S
	\right) 
	\begin{bmatrix} I_n & H^\top \\ {\bf 0} & I_m \end{bmatrix},\\
	& \mbox{Tr} \left(\begin{bmatrix}
	I_n & \boldsymbol{0} \\ \boldsymbol{0} & \boldsymbol{0} 
	\end{bmatrix} S \right) 
	+ \mbox{Tr} \left[ \Sigma_x+\hat{\Sigma}_x-
	2\left( \hat{\Sigma}_x^{\frac{1}{2}} \Sigma_x \hat{\Sigma}_x^{\frac{1}{2}}\right) ^{\frac{1}{2}}\right]  \leq \rho_x^2,\\
	&\mbox{Tr} \left(\begin{bmatrix}
	\boldsymbol{0} & \boldsymbol{0} \\ \boldsymbol{0} &I_m
	\end{bmatrix} S \right)
	+ \mbox{Tr} \left[ \Sigma_w+\hat{\Sigma}_w-
	2 \left( \hat{\Sigma}_w^{\frac{1}{2}} \Sigma_w \hat{\Sigma}_w^{\frac{1}{2}}\right) ^{\frac{1}{2}} \right]  \leq \rho_w^2.
	\end{array} \label{least SDPsss}
	\end{align}
	By introducing a slack variable $Q \in \mathbb{S}^n$, applying the Schur complement theorem, and eliminating $R$, problem (\ref{least SDPsss}) is equivalent to
	\begin{align}
	\begin{array}{cl}
	\sup \limits_{Q,S,\Sigma_x,\Sigma_w}
	&\mbox{Tr}(Q) \\
	\mbox{s.t.}&  S \succeq{\bf 0}, \ \Sigma_x \succeq{\bf 0}, \ \Sigma_w  \succeq{\bf 0}, \\
	&\begin{bmatrix} I_n & {\bf 0} \\H & I_m \end{bmatrix} 
	\left( 
	\begin{bmatrix} \Sigma_x & {\bf 0} \\ {\bf 0} & \Sigma_w \end{bmatrix}+S
	\right) 
	\begin{bmatrix} I_n & H^\top \\ {\bf 0} & I_m \end{bmatrix}-\begin{bmatrix}
	Q & {\bf 0} \\ {\bf 0} & {\bf 0}
	\end{bmatrix} \succeq {\bf 0},\\
	& \mbox{Tr} \left(\begin{bmatrix}
	I_n & \boldsymbol{0} \\ \boldsymbol{0} & \boldsymbol{0} 
	\end{bmatrix} S \right) 
	+ \mbox{Tr} \left[ \Sigma_x+\hat{\Sigma}_x-
	2 \left( \hat{\Sigma}_x^{\frac{1}{2}} \Sigma_x \hat{\Sigma}_x^{\frac{1}{2}}\right) ^{\frac{1}{2}}\right]  \leq \rho_x^2,\\
	&\mbox{Tr} \left(\begin{bmatrix}
	\boldsymbol{0} & \boldsymbol{0} \\ \boldsymbol{0} &I_m
	\end{bmatrix} S \right)
	+ \mbox{Tr} \left[ \Sigma_w+\hat{\Sigma}_w-
	2 \left( \hat{\Sigma}_w^{\frac{1}{2}} \Sigma_w \hat{\Sigma}_w^{\frac{1}{2}}\right) ^{\frac{1}{2}} \right]  \leq \rho_w^2.
	\end{array} \label{minmax SDPs}
	\end{align}
	Finally, by Proposition 2 in \cite{malago2018wasserstein}, introducing auxiliary variables $V_x \in \mathbb{S}_+^n$ and $V_w \in \mathbb{S}_+^m$, (\ref{minmax SDPs}) can be transformed into the SDP problem (\ref{minmax SDP}), which completes the proof.
\end{proof}

\begin{remark}
	The optimal value of (\ref{minmax SDP}) is equal to that of (\ref{minmaxG}), while the optimal value of (\ref{maxmin SDP}) is equal to that of problem ({\ref{maxmina}}) obtained by exchanging the supremum and infimum in (\ref{minmaxG}). Therefore, according to Theorem \ref{equp}, if the optimal values of the two SDP problems (\ref{minmax SDP}) and (\ref{maxmin SDP}) are equal, the robust estimation problem (\ref{minmax}) has a saddle point solution.
\end{remark}

\subsection{An Optimal Solution to (\ref{minmaxG}) and (\ref{infLsup})}
Section 4.1 shows that the optimal value of (\ref{minmaxG}) is equal to that of an SDP problem (\ref{minmax SDP}). 
However, in practical applications, we often need to obtain the optimal solution of (\ref{minmaxG}), specifically the robust linear estimator and the corresponding least favorable distribution. 
When the saddle point solution does not exist, the robust linear estimator is not the optimal estimator corresponding to its least favorable distribution, %obtained by the SDP problem (\ref{least SDPss}), 
and thus it can not be directly calculated.

Notice that problem (\ref{minmaxG}) is equivalent to (\ref{least SDPs}), while the SDP problem (\ref{minmax SDP}) is equivalent to (\ref{least SDPss}) obtained by exchanging minimization and maximization in problem (\ref{least SDPs}). 
Therefore, if the saddle point solution of (\ref{least SDPss}) can be constructed by the primal and dual optimal solutions of (\ref{minmax SDP}), then this saddle point solution is also an optimal solution to (\ref{least SDPs}), and thus the optimal solution of (\ref{minmaxG}) can be further obtained.

Therefore, we first consider the Lagrangian function of (\ref{minmax SDP}) denoted by
\begin{align*}
&\mathfrak{L}(Q,S,\Sigma_x,\Sigma_w,V_x,V_w,G_S,G_x,G_w,G_{vx},G_{vw},W,T_x,T_w,\alpha_x,\alpha_w)\\
&\begin{aligned}
=&-\mbox{Tr}(Q)-\mbox{Tr}(G_S^\top S)-\mbox{Tr}(G_x^\top \Sigma_{x})
-\mbox{Tr}(G_w^\top \Sigma_w)-\mbox{Tr}(G_{vx}^\top V_x)-\mbox{Tr}(G_{vw}^\top V_w)\\
&-\mbox{Tr}\left[ W^\top \left( \begin{bmatrix} I_n & {\bf 0} \\H & I_m \end{bmatrix} 
\left( 
\begin{bmatrix} \Sigma_x & {\bf 0} \\ {\bf 0} & \Sigma_w \end{bmatrix}+S
\right) 
\begin{bmatrix} I_n & H^\top \\ {\bf 0} & I_m \end{bmatrix}-\begin{bmatrix}
Q & {\bf 0} \\ {\bf 0} & {\bf 0}
\end{bmatrix} \right) \right] \\
&-\mbox{Tr}\left(T_x^\top \begin{bmatrix}
\hat{\Sigma}^\frac{1}{2}_x \Sigma_x \hat{\Sigma}^\frac{1}{2}_x & V_x \\
V_x & I_n \end{bmatrix} \right) 
-\mbox{Tr}\left(T_w^\top \begin{bmatrix}
\hat{\Sigma}^\frac{1}{2}_w \Sigma_w \hat{\Sigma}^\frac{1}{2}_w & V_w \\
V_w & I_m \end{bmatrix} \right)\\
&+\alpha_x \left[\mbox{Tr} \left(\begin{bmatrix}
I_n & \boldsymbol{0} \\ \boldsymbol{0} & \boldsymbol{0} 
\end{bmatrix} S \right)
+\mbox{Tr}\left( \Sigma_x + \hat{\Sigma}_x - 2 V_x \right)
-\rho_x^2 \right]\\
&+\alpha_w \left[ \mbox{Tr} \left(\begin{bmatrix}
\boldsymbol{0} & \boldsymbol{0} \\ \boldsymbol{0} &I_m
\end{bmatrix} S \right)
+\mbox{Tr}\left(\Sigma_w + \hat{\Sigma}_w - 2 V_w\right)
-\rho_w^2 \right], 
\end{aligned} 
\end{align*}
where the dual variables $G_x,G_{vx} \in \mathbb{S}_+^n$; $G_w,G_{vw} \in \mathbb{S}_+^m$; $G_S,W \in \mathbb{S}_+^{m+n}$; $T_x \in \mathbb{S}_+^{2n}$; $T_w \in \mathbb{S}_+^{2m}$ and 
$\alpha_x,\alpha_w \in \mathbb{R}_+$.
Since (\ref{minmax SDP}) is a convex optimization problem, its optimal solution must stasify KKT system, i.e.,
\begin{subequations} 
\begin{align}
&\bigtriangledown \mathfrak{L}_Q=-I_n+W_{11}={\bf 0},\label{KKT1} \\
&\bigtriangledown \mathfrak{L}_S=-G_S-\begin{bmatrix} I_n & H^\top \\ {\bf 0} & I_m \end{bmatrix} W
\begin{bmatrix} I_n & {\bf 0} \\ H & I_m \end{bmatrix}+\alpha_x \begin{bmatrix} I_n & \boldsymbol{0} \\ \boldsymbol{0} & \boldsymbol{0} 
\end{bmatrix}+\alpha_w \begin{bmatrix}
\boldsymbol{0} & \boldsymbol{0} \\ \boldsymbol{0} &I_m
\end{bmatrix}={\bf 0}, \label{KKT2}\\
&\bigtriangledown \mathfrak{L}_{\Sigma_x}=-G_x-\left( W_{11}+W_{12} H+H^\top W_{12}^\top+H^\top W_{22} H\right) -
\hat{\Sigma}_x^\frac{1}{2} T^{11}_x \hat{\Sigma}_x^\frac{1}{2}+
\alpha_x I_n={\bf 0}, \label{KKT3}\\
&\bigtriangledown \mathfrak{L}_{\Sigma_w}=
-G_w-W_{22}-\hat{\Sigma}_w^\frac{1}{2} T^{11}_w \hat{\Sigma}_w^\frac{1}{2}+\alpha_w I_m={\bf 0},\label{KKT4}\\
&\bigtriangledown \mathfrak{L}_{V_x}=
-G_{vx}-T^{12}_x-\left( T^{12}_x\right) ^\top-2\alpha_x I_n={\bf 0},\label{KKT5}\\
&\bigtriangledown \mathfrak{L}_{V_w}=
-G_{vw}-T^{12}_w-\left( T^{12}_w\right) ^\top-2\alpha_w I_m={\bf 0},\label{KKT6} \\
&{\bf 0} \preceq G_S \perp S \succeq {\bf 0}, \label{KKT7} \\
&{\bf 0} \preceq G_x \perp \Sigma_x \succeq {\bf 0}, \label{KKT8} \\
&{\bf 0} \preceq G_w \perp \Sigma_w \succeq {\bf 0}, \label{KKT9} \\
&{\bf 0} \preceq G_{vx} \perp V_x \succeq {\bf 0},\label{KKT10} \\
&{\bf 0} \preceq G_{vw} \perp V_w \succeq {\bf 0},\label{KKT11} \\
&{\bf 0} \preceq W \perp \left( 
\begin{bmatrix} I_n & {\bf 0} \\H & I_m \end{bmatrix} 
\left( \begin{bmatrix} \Sigma_x & {\bf 0} \\ {\bf 0} & \Sigma_w \end{bmatrix}+S \right) 
\begin{bmatrix} I_n & H^\top \\ {\bf 0} & I_m \end{bmatrix}-\begin{bmatrix}
Q & {\bf 0} \\ {\bf 0} & {\bf 0}
\end{bmatrix} \right)  \succeq {\bf 0}, \label{KKT12}\\
&{\bf 0} \preceq T_x \perp \begin{bmatrix}
\hat{\Sigma}^\frac{1}{2}_x \Sigma_x \hat{\Sigma}^\frac{1}{2}_x & V_x \\
V_x & I_n \end{bmatrix} \succeq {\bf 0}, \label{KKT13} \\
&{\bf 0} \preceq T_w \perp \begin{bmatrix}
\hat{\Sigma}^\frac{1}{2}_w \Sigma_w \hat{\Sigma}^\frac{1}{2}_w & V_w \\
V_w & I_m \end{bmatrix} \succeq {\bf 0}, \label{KKT14} \\
&{\bf 0} \preceq \alpha_x \perp \left[\rho_x^2-\mbox{Tr} \left(\begin{bmatrix}
I_n & \boldsymbol{0} \\ \boldsymbol{0} & \boldsymbol{0} 
\end{bmatrix} S \right) 
-\mbox{Tr}\left( \Sigma_x + \hat{\Sigma}_x - 2 V_x \right) \right] \succeq {\bf 0}, \label{KKT15} \\
&{\bf 0} \preceq \alpha_w \perp \left[\rho_w^2-\mbox{Tr} \left(\begin{bmatrix}
\boldsymbol{0} & \boldsymbol{0} \\ \boldsymbol{0} &I_m
\end{bmatrix} S \right)
-\mbox{Tr}\left(\Sigma_w + \hat{\Sigma}_w - 2 V_w\right) \right] \succeq {\bf 0}, \label{KKT16}
\end{align} \label{KKTs}
\end{subequations}
where some dual variables are appropriately partitioned as
$$W=\begin{bmatrix} W_{11} & W_{12} \\ W_{12}^\top & W_{22} \end{bmatrix}, T_x=\begin{bmatrix}
T^{11}_x & T^{12}_x \\ (T^{12}_x)^\top & T^{22}_x
\end{bmatrix}, T_w=\begin{bmatrix}
T^{11}_w & T^{12}_w \\ (T^{12}_w)^\top & T^{22}_w
\end{bmatrix},$$
with the first diagonal submatrices $W_{11}, T^{11}_x \in \mathbb{S}_+^n$, $T^{11}_w \in \mathbb{S}_+^m$. 
Then we can formulate the saddle point solution of (\ref{least SDPss}) in the following theorem.

\begin{theorem}
	Assuming that $(Q^*,S^*,\Sigma_x^*,\Sigma_w^*,V_x^*,V_w^*,G_S^*,G_x^*,G_w^*,G_{vx}^*,G_{vw}^*,W^*,T_x^*,T_w^*,\alpha_x^*,\alpha_w^*)$ is a solution of KKT system (\ref{KKTs}), then $A^*=-W_{12}^*$ and $(S^*,\Sigma_{x}^*,\Sigma_w^*)$ constitute a saddle point solution of (\ref{least SDPss}).
\end{theorem}
\begin{proof}
	First, we prove that $A^*=-W_{12}^*$ and $(S^*,\Sigma_{x}^*,\Sigma_w^*)$ constitute an optimal solution to (\ref{least SDPss}), i.e., 
	\begin{equation}
	\mbox{Tr}(Q^*)=\mbox{Tr} \left( \begin{bmatrix} I_n & W_{12}^* \end{bmatrix}
	\begin{bmatrix} I_n & {\bf 0} \\H & I_m \end{bmatrix} 
	\left( 
	\begin{bmatrix} \Sigma_x^* & {\bf 0} \\ {\bf 0} & \Sigma_w^* \end{bmatrix}+S^*
	\right) 
	\begin{bmatrix} I_n & H^\top \\ {\bf 0} & I_m \end{bmatrix} 
	\begin{bmatrix} I_n \\  (W_{12}^*)^\top  \end{bmatrix} \right).
	\label{T33.1} \end{equation}
	According to (\ref{KKT1}) and (\ref{KKT12}), we have
    \begin{equation}
	\begin{bmatrix} I_n & W_{12}^* \\ (W_{12}^*)^\top & W_{22}^* \end{bmatrix} \left(  
	\begin{bmatrix} I_n & {\bf 0} \\H & I_m \end{bmatrix} 
	\left( \begin{bmatrix} \Sigma_x^* & {\bf 0} \\ {\bf 0} & \Sigma_w^* \end{bmatrix}+S^* \right) 
	\begin{bmatrix} I_n & H^\top \\ {\bf 0} & I_m \end{bmatrix}-\begin{bmatrix}
	Q^* & {\bf 0} \\ {\bf 0} & {\bf 0}
	\end{bmatrix} \right) ={\bf 0}, 
	\label{T33.11} \end{equation}
	which implies that
	$$\begin{bmatrix} I_n & W_{12}^* \end{bmatrix} \left(  
	\begin{bmatrix} I_n & {\bf 0} \\H & I_m \end{bmatrix} 
	\left( \begin{bmatrix} \Sigma_x^* & {\bf 0} \\ {\bf 0} & \Sigma_w^* \end{bmatrix}+S^* \right) 
	\begin{bmatrix} I_n & H^\top \\ {\bf 0} & I_m \end{bmatrix}-\begin{bmatrix}
	Q^* & {\bf 0} \\ {\bf 0} & {\bf 0}
	\end{bmatrix} \right) ={\bf 0}. $$
	Then we derive
	\begin{align}
	\begin{bmatrix} I_n \\ (W_{12}^*)^\top \end{bmatrix}
	\begin{bmatrix} I_n & W_{12}^* \end{bmatrix}  
	\begin{bmatrix} I_n & {\bf 0} \\H & I_m \end{bmatrix} 
	\left( \begin{bmatrix} \Sigma_x^* & {\bf 0} \\ {\bf 0} & \Sigma_w^* \end{bmatrix}+S^* \right) 
	\begin{bmatrix} I_n & H^\top \\ {\bf 0} & I_m \end{bmatrix}
	=\begin{bmatrix} I_n \\ (W_{12}^*)^\top \end{bmatrix}
	\begin{bmatrix} I_n & W_{12}^* \end{bmatrix}
	\begin{bmatrix} Q^* & {\bf 0} \\ {\bf 0} & {\bf 0} \end{bmatrix}
	=\begin{bmatrix}
	Q^* & {\bf 0} \\ (W_{12}^*)^\top Q^* & {\bf 0}
	\end{bmatrix}. \label{T33.2}
    \end{align}
    Taking trace operation on both sides, equation (\ref{T33.1}) is proved.
    
    Subsequently, it is sufficient to prove that for given $A^*=-W_{12}^*$, $(S^*,\Sigma_{x}^*,\Sigma_w^*)$  is the maximizer of the inner maximization problem in (\ref{least SDPs}), which is equivalent to proving that $(S^*,\Sigma_{x}^*,\Sigma_w^*,V_x^*,V_w^*)$ is an optimal solution to the following problem
    \begin{align}
    \begin{array}{cl}
    \sup \limits_{S,\Sigma_x,\Sigma_w,V_x,V_w} 
    &\mbox{Tr}\left[\left(W_{12}^*H+I_n\right)\Sigma_x\left(W_{12}^*H+I_n\right)^\top
    +W_{12}^* \Sigma_w (W_{12}^*)^\top\right]+
    \mbox{Tr} \left(\begin{bmatrix} (W_{12}^*H+I_n)^\top \\  (W_{12}^*)^\top  \end{bmatrix}
    \begin{bmatrix} W_{12}^*H+I_n & W_{12}^* \end{bmatrix} S \right) \\
    \mbox{s.t.}& S \succeq{\bf 0}, \ \Sigma_x \succeq{\bf 0}, \ \Sigma_w  \succeq{\bf 0}, \ V_x \succeq{\bf 0},\ V_w \succeq{\bf 0}\\
    &\left[\begin{array} {cc}
    \hat{\Sigma}_x^{\frac{1}{2}} \Sigma_x \hat{\Sigma}_x^{\frac{1}{2}} &
    V_x \\ V_x &I_n
    \end{array}
    \right]\succeq {\bf 0}, \quad
    \left[\begin{array} {cc}
    \hat{\Sigma}_w^{\frac{1}{2}} \Sigma_w \hat{\Sigma}_w^{\frac{1}{2}} &
    V_w \\ V_w &I_m
    \end{array}
    \right]\succeq {\bf 0},\\
    & \mbox{Tr} \left(\begin{bmatrix}
    I_n & \boldsymbol{0} \\ \boldsymbol{0} & \boldsymbol{0} 
    \end{bmatrix} S \right) 
    +\mbox{Tr}\left( \Sigma_x + \hat{\Sigma}_x - 2 V_x \right) 
    \leq \rho_x^2 ,\\
    &\mbox{Tr} \left(\begin{bmatrix}
    \boldsymbol{0} & \boldsymbol{0} \\ \boldsymbol{0} &I_m
    \end{bmatrix} S \right)
    +\mbox{Tr}\left(\Sigma_w + \hat{\Sigma}_w - 2 V_w\right)  \leq \rho_w^2.
    \end{array} \label{T33.3}
    \end{align}
    Let the Lagrangian function of (\ref{T33.3}) be
    \begin{align*}
    &\bar{\mathfrak{L}}(S,\Sigma_x,\Sigma_w,V_x,V_w,\hat{G}_S,\hat{G}_x,\hat{G}_w,\hat{G}_{vx},\hat{G}_{vw},\hat{T}_x,\hat{T}_w,\hat{\alpha}_x,\hat{\alpha}_w)\\
    &\begin{aligned}
    =&-\mbox{Tr}\left[\left(W_{12}^*H+I_n\right)\Sigma_x\left(W_{12}^*H
    +I_n\right)^\top+W_{12}^* \Sigma_w \left( W_{12}^*\right) ^\top\right]-
    \mbox{Tr} \left(\begin{bmatrix} (W_{12}^*H+I_n)^\top \\  (W_{12}^*)^\top  \end{bmatrix}
    \begin{bmatrix} W_{12}^*H+I_n & W_{12}^* \end{bmatrix} S \right)\\
    &-\mbox{Tr}\left( \hat{G}_S^\top S\right) 
    -\mbox{Tr}\left( \hat{G}_x^\top \Sigma_{x}\right) 
    -\mbox{Tr}\left( \hat{G}_w^\top \Sigma_w\right) 
    -\mbox{Tr}\left( \hat{G}_{vx}^\top V_x\right) 
    -\mbox{Tr}\left( \hat{G}_{vw}^\top V_w\right) \\
    &-\mbox{Tr}\left(\hat{T}_x^\top \begin{bmatrix}
    \hat{\Sigma}^\frac{1}{2}_x \Sigma_x \hat{\Sigma}^\frac{1}{2}_x & V_x \\
    V_x & I_n \end{bmatrix} \right) 
    -\mbox{Tr}\left(\hat{T}_w^\top \begin{bmatrix}
    \hat{\Sigma}^\frac{1}{2}_w \Sigma_w \hat{\Sigma}^\frac{1}{2}_w & V_w \\
    V_w & I_m \end{bmatrix} \right)\\
    &+\hat{\alpha}_x \left[\mbox{Tr} \left(\begin{bmatrix}
    I_n & \boldsymbol{0} \\ \boldsymbol{0} & \boldsymbol{0} 
    \end{bmatrix} S \right)
    +\mbox{Tr}\left( \Sigma_x + \hat{\Sigma}_x - 2 V_x \right)
    -\rho_x^2 \right]\\
    &+\hat{\alpha}_w \left[\mbox{Tr} \left(\begin{bmatrix}
    \boldsymbol{0} & \boldsymbol{0} \\ \boldsymbol{0} &I_m
    \end{bmatrix} S \right)
    +\mbox{Tr}\left(\Sigma_w + \hat{\Sigma}_w - 2 V_w\right)
    -\rho_w^2 \right], 
    \end{aligned} 
    \end{align*}
    where the dual variables $\hat{G}_S \in \mathbb{S}_+^{m+n}$; $\hat{G}_x,\hat{G}_{vx} \in \mathbb{S}_+^n$; $\hat{G}_w,\hat{G}_{vw} \in \mathbb{S}_+^m$; $\hat{T}_x \in \mathbb{S}_+^{2n}$; $\hat{T}_w \in \mathbb{S}_+^{2m}$ and 
    $\hat{\alpha}_x,\hat{\alpha}_w \in \mathbb{R}_+$.
    Since (\ref{T33.3}) is a convex optimization problem, its optimal solution stasifies the KKT conditions, i.e.,
    \begin{subequations}
    \begin{align}
    &\bigtriangledown \bar{\mathfrak{L}}_{S}=-\hat{G}_S-
    \begin{bmatrix} I_n & H^\top \\ {\bf 0} & I_m \end{bmatrix}
    \begin{bmatrix} I_n \\ (W_{12}^*)^\top \end{bmatrix}
    \begin{bmatrix} I_n & W_{12}^* \end{bmatrix}
    \begin{bmatrix} I_n & {\bf 0} \\ H & I_m \end{bmatrix}+\alpha_x \begin{bmatrix} I_n & \boldsymbol{0} \\ \boldsymbol{0} & \boldsymbol{0} 
    \end{bmatrix}+\alpha_w \begin{bmatrix}
    \boldsymbol{0} & \boldsymbol{0} \\ \boldsymbol{0} &I_m
    \end{bmatrix}={\bf 0}, \label{KKTs1}\\
    &\bigtriangledown \bar{\mathfrak{L}}_{\Sigma_x}=-\hat{G}_x-\left[ I_n+W_{12} H+H^\top \left( W_{12}^*\right) ^\top+H^\top \left( W_{12}^*\right) ^\top W_{12}^* H\right] -
    \hat{\Sigma}_x^\frac{1}{2} \hat{T}^{11}_x \hat{\Sigma}_x^\frac{1}{2}+
    \hat{\alpha}_x I_n={\bf 0}, \label{KKTs2}\\
    &\bigtriangledown \bar{\mathfrak{L}}_{\Sigma_w}=
    -\hat{G}_w-\left( W_{12}^*\right) ^\top W_{12}^*-\hat{\Sigma}_w^\frac{1}{2} \hat{T}^{11}_w \hat{\Sigma}_w^\frac{1}{2}+\hat{\alpha}_w I_m={\bf 0},\label{KKTs3}\\
    &\bigtriangledown \bar{\mathfrak{L}}_{V_x}=
    -\hat{G}_{vx}-\hat{T}^{12}_x-\left( \hat{T}^{12}_x\right) ^\top-2\hat{\alpha}_x I_n={\bf 0},\label{KKTs4}\\
    &\bigtriangledown \bar{\mathfrak{L}}_{V_w}=
    -\hat{G}_{vw}-\hat{T}^{12}_w-\left( \hat{T}^{12}_w\right) ^\top-2\hat{\alpha}_w I_m={\bf 0},\label{KKTs5} \\
    &{\bf 0} \preceq \hat{G}_S \perp S \succeq {\bf 0}, \label{KKTs6} \\
    &{\bf 0} \preceq \hat{G}_x \perp \Sigma_x \succeq {\bf 0}, \label{KKTs7} \\
    &{\bf 0} \preceq \hat{G}_w \perp \Sigma_w \succeq {\bf 0}, \label{KKTs8} \\
    &{\bf 0} \preceq \hat{G}_{vx} \perp V_x \succeq {\bf 0},\label{KKTs9} \\
    &{\bf 0} \preceq \hat{G}_{vw} \perp V_w \succeq {\bf 0},\label{KKTs10} \\
    &{\bf 0} \preceq \hat{T}_x \perp \begin{bmatrix}
    \hat{\Sigma}^\frac{1}{2}_x \Sigma_x \hat{\Sigma}^\frac{1}{2}_x & V_x \\
    V_x & I_n \end{bmatrix} \succeq {\bf 0}, \label{KKTs11} \\
    &{\bf 0} \preceq \hat{T}_w \perp \begin{bmatrix}
    \hat{\Sigma}^\frac{1}{2}_w \Sigma_w \hat{\Sigma}^\frac{1}{2}_w & V_w \\
    V_w & I_m \end{bmatrix} \succeq {\bf 0}, \label{KKTs12} \\
    &{\bf 0} \preceq \hat{\alpha}_x \perp \left[\rho_x^2-\mbox{Tr} \left(\begin{bmatrix}
    I_n & \boldsymbol{0} \\ \boldsymbol{0} & \boldsymbol{0} 
    \end{bmatrix} S \right)
    -\mbox{Tr}\left( \Sigma_x + \hat{\Sigma}_x - 2 V_x \right) \right] \succeq {\bf 0}, \label{KKTs13} \\
    &{\bf 0} \preceq \hat{\alpha}_w \perp \left[\rho_w^2-\mbox{Tr} \left(\begin{bmatrix}
    \boldsymbol{0} & \boldsymbol{0} \\ \boldsymbol{0} &I_m
    \end{bmatrix} S \right) 
    -\mbox{Tr}\left(\Sigma_w + \hat{\Sigma}_w - 2 V_w\right) \right] \succeq {\bf 0}, \label{KKTs14}
    \end{align} \label{KKTss}
    \end{subequations}
where some dual variables are appropriately partitioned as
$$\hat{T}_x=\begin{bmatrix}
\hat{T}^{11}_x & \hat{T}^{12}_x \\ \left( \hat{T}^{12}_x\right) ^\top & \hat{T}^{22}_x
\end{bmatrix}, \hat{T}_w=\begin{bmatrix}
\hat{T}^{11}_w & \hat{T}^{12}_w \\ \left( \hat{T}^{12}_w\right) ^\top & \hat{T}^{22}_w
\end{bmatrix},$$
with the first diagonal submatrices $\hat{T}^{11}_x \in \mathbb{S}_+^n$ and  $\hat{T}^{11}_w \in \mathbb{S}_+^m$. 

Then we intend to prove that $(S^*,\Sigma_x^*,\Sigma_w^*,V^*_x,V^*_w,\hat{G}^*_S,\hat{G}^*_x,\hat{G}^*_w,\hat{G}^*_{vx},\hat{G}^*_{vw},\hat{T}^*_x,\hat{T}^*_w,\hat{\alpha}^*_x,\hat{\alpha}^*_w)$ 
stasifies the KKT system (\ref{KKTss}), where
\begin{subequations}
\begin{align*}
&\hat{G}^*_S=G_S^*+\begin{bmatrix} I_n & H^\top \\ {\bf 0} & I_m \end{bmatrix}
\begin{bmatrix} {\bf 0} & {\bf 0} \\ {\bf 0} & 
W_{22}^*-(W_{12}^*)^\top W_{12}^* \end{bmatrix}
\begin{bmatrix} I_n & {\bf 0} \\ H & I_m \end{bmatrix},\\
&\hat{G}^*_x=G_x^*+H^\top \left[ W_{22}^*-\left( W_{12}^*\right) ^\top W_{12}^* \right] H,\\
&\hat{G}^*_w=G_w^*+W_{22}^*-(W_{12}^*)^\top W_{12}^*,\\
&\hat{G}^*_{vx}=G^*_{vx}, \ \hat{G}^*_{vw}=G^*_{vw}, \
\hat{T}^*_x=T_x^*, \ \hat{T}^*_w=T_w^*, \
\hat{\alpha}^*_x=\alpha_x^*,\ \hat{\alpha}^*_w=\alpha_w^*.
\end{align*} 
\end{subequations}
It is obvious that (\ref{KKTs1}-\ref{KKTs5}) and (\ref{KKTs9}-\ref{KKTs14}) hold.
Since $W^* \succeq {\bf 0}$, it holds that $W_{22}^*-(W_{12}^*)^\top W_{12}^* \succeq {\bf 0}$ by schur complement theorem. 
Then due to the positive semidefiniteness of $G_S^*$, $G_X^*$ and $G_w^*$, we obtain that $\hat{G}^*_S$, $\hat{G}^*_x$ and $\hat{G}^*_w$ are also positive semidefinite. 
Therefore, it only remains for us to demonstrate that $\hat{G}^*_S S^*={\bf 0}$, $\hat{G}^*_x \Sigma_x^*={\bf 0}$ and $\hat{G}^*_w \Sigma_w^*={\bf 0}$.
According to (\ref{T33.11}) and (\ref{T33.2}), we have
\begin{equation*}
\begin{bmatrix} {\bf 0} & {\bf 0} \\ {\bf 0} & W_{22}^*-(W_{12}^*)^\top W_{12}^* \end{bmatrix} 
\begin{bmatrix} I_n & {\bf 0} \\H & I_m \end{bmatrix} 
\left( \begin{bmatrix} \Sigma_x^* & {\bf 0} \\ {\bf 0} & \Sigma_w^* \end{bmatrix}+S^* \right) 
\begin{bmatrix} I_n & H^\top \\ {\bf 0} & I_m \end{bmatrix} ={\bf 0},
\end{equation*}
which implies that
\begin{equation*}
\begin{bmatrix} I_n & H^\top \\ {\bf 0} & I_m \end{bmatrix} 
\begin{bmatrix} {\bf 0} & {\bf 0} \\ {\bf 0} & W_{22}^*-(W_{12}^*)^\top W_{12}^* \end{bmatrix} 
\begin{bmatrix} I_n & {\bf 0} \\H & I_m \end{bmatrix} 
\left( \begin{bmatrix} \Sigma_x^* & {\bf 0} \\ {\bf 0} & \Sigma_w^* \end{bmatrix}+S^* \right) ={\bf 0}.
\end{equation*}
Since the matrices 
$\begin{bmatrix} I_n & H^\top \\ {\bf 0} & I_m \end{bmatrix} 
\begin{bmatrix} {\bf 0} & {\bf 0} \\ {\bf 0} & W_{22}^*-(W_{12}^*)^\top W_{12}^* \end{bmatrix} 
\begin{bmatrix} I_n & {\bf 0} \\H & I_m \end{bmatrix} $,
$\begin{bmatrix} \Sigma_x^* & {\bf 0} \\ {\bf 0} & \Sigma_w^* \end{bmatrix}$ and $S^*$ are positive semidefinite, we have
\begin{equation}
\begin{bmatrix} I_n & H^\top \\ {\bf 0} & I_m \end{bmatrix} 
\begin{bmatrix} {\bf 0} & {\bf 0} \\ {\bf 0} & W_{22}^*-(W_{12}^*)^\top W_{12}^* \end{bmatrix} 
\begin{bmatrix} I_n & {\bf 0} \\H & I_m \end{bmatrix}
\begin{bmatrix} \Sigma_x^* & {\bf 0} \\ {\bf 0} & \Sigma_w^* \end{bmatrix}={\bf 0}, \label{T33.5} 
\end{equation}
and
\begin{equation}
\begin{bmatrix} I_n & H^\top \\ {\bf 0} & I_m \end{bmatrix} 
\begin{bmatrix} {\bf 0} & {\bf 0} \\ {\bf 0} & W_{22}^*-(W_{12}^*)^\top W_{12}^* \end{bmatrix} 
\begin{bmatrix} I_n & {\bf 0} \\H & I_m \end{bmatrix}
S^*={\bf 0}, \label{T33.6}
\end{equation}
which follows from the fact that for any matrices $A,B,C \succeq {\bf 0}$,  the equation $A(B+C)={\bf 0}$ implies $AB={\bf 0}$ and $AC={\bf 0}$ and the fact is a direct consequence of Fact 10.14.5 in \cite{bernstein2018scalar}.
%In accordance with the diagonal blocks of the left matrix being zero matrices in (\ref{T33.5}), we have
In accordance with equation (\ref{T33.5}), we have
\begin{equation}
H^\top \left[ W_{22}^*-\left( W_{12}^*\right) ^\top W_{12}^* \right] H \Sigma_x^*={\bf 0} \label{T33.7} \end{equation}
and
\begin{equation}
\left[ W_{22}^*-\left( W_{12}^*\right) ^\top W_{12}^*\right] \Sigma_w^*={\bf 0}. \label{T33.8} \end{equation}
Thus, combining (\ref{T33.6}-\ref{T33.8}) with (\ref{KKT7}-\ref{KKT9}),
it follows that (\ref{KKT6}-\ref{KKT8}) hold.

Therefore, $(S^*,\Sigma_x^*,\Sigma_w^*,V^*_x,V^*_w,\hat{G}^*_S,\hat{G}^*_x,\hat{G}^*_w,\hat{G}^*_{vx},\hat{G}^*_{vw},\hat{T}^*_x,\hat{T}^*_w,\hat{\alpha}^*_x,\hat{\alpha}^*_w)$ 
stasifies the KKT system (\ref{KKTss}).
Then for given $A^*=-W_{12}^*$, $(S^*,\Sigma_{x}^*,\Sigma_w^*)$  is the maximizer of the inner maximization problem in (\ref{least SDPs}), which implies that $A^*=-W_{12}^*$ and $(S^*,\Sigma_{x}^*,\Sigma_w^*)$ constitute a saddle point solution of (\ref{least SDPss}).
\end{proof}

For given $A^*=-W_{12}^*$, $\Sigma_{x}^*$ and $\Sigma_{w}^*$, we can obtain a rank-one optimal solution 
$S^*=\begin{bmatrix} \tilde{\mu}_x^* \\ \tilde{\mu}_w^* \end{bmatrix}
\begin{bmatrix} (\tilde{\mu}_x^*)^\top & (\tilde{\mu}_w^*)^\top \end{bmatrix}$
to the inner maximization problem over $S$ in (\ref{least SDPs}) as outlined in the proof of Theorem 2.5 in \cite{ye2003new}.
Then $(A^*,\tilde{\mu}^*_x,\tilde{\mu}^*_w,\Sigma^*_x,\Sigma^*_w)$ is an optimal solution to (\ref{2minimax}), which further implies that 
$(A^*,-(A^*H-I_n)\hat{\mu}_x-A^* \hat{\mu}_w,\tilde{\mu}^*_x+\hat{\mu}_x,\tilde{\mu}^*_w+\hat{\mu}_w,\Sigma^*_x,\Sigma^*_w)$ is an optimal solution to (\ref{minmaxGG}).
That is, $f^*(y)=A^*y-(A^*H-I_n)\hat{\mu}_x-A^* \hat{\mu}_w$ and $P^*=\mathcal{N}(\tilde{\mu}^*_x+\hat{\mu}_x,\Sigma^*_x) \times \mathcal{N}(\tilde{\mu}^*_w+\hat{\mu}_w,\Sigma^*_w)$ constitute an optimal solution to (\ref{minmaxG}) and (\ref{infLsup}).

\section{Simulation}
In this section, we intend to verify the effectiveness of our theory  through numerical experiments.
All experiments are implemented in MATLAB R2024a on a PC with AMD Ryzen 7  9800X3D processors (4.7GHz) and 64 GB of RAM.
In all experiments, the SDP problems are numerically solved by SDPT3 solver through CVX interface \cite{CVX}, and all parameters are set to default values when solving optimization problems.

\subsection{The Nonexistence of the Saddle Point}
In the first experiment, we aim to illustrate that the saddle point may not exist in high-dimensional case when the parameter and noise dimensions are fixed to $n=m=d$.
We take the nominal mean vectors to be $\hat{\mu}_x=\hat{\mu}_w={\bf 0}$ and draw the elements of the observation matrix $H$ independently from the standard Gaussian distribution. 
The nominal covariance matrices $\hat{\Sigma}_x$ and $\hat{\Sigma}_w$ are constructed as follows: first, we sample the elements of matrices $Q_x$ and $Q_w$ independently from the standard Gaussian distribution and denote $R_x$ and $R_w$ the orthogonal matrices whose columns are the orthogonal eigenvectors of $Q_x+Q_x^\top$ and $Q_w+Q_w^\top$, respectively. 
Then we define $\hat{\Sigma}_x=R_x \Lambda_x R_x^\top$ and $\hat{\Sigma}_w=R_w \Lambda_w R_w^\top$, where $\Lambda_x$ and $\Lambda_w$ are diagonal with entries sampled uniformly from [1,5] and [1,2], respectively. 
Finally, we set the Wasserstein radiu of the parameter distribution to $\rho_x=3$ and vary the Wasserstein radiu of the noise distribution $\rho_w$ across the interval [1,10] with a stepsize of 0.5.

For a given $\rho_w$, we first calculate the optimal solution to (\ref{maxmina}), which is equal to that of the SDP problem (\ref{maxmin SDP}).
From the proof of Theorem \ref{snsdu}, the optimal solution to (\ref{maxmina}) denoted by $(f^*,P^*)$ is unique due to the positive definiteness of the nominal covariance matrices $\hat{\Sigma}_x$ and $\hat{\Sigma}_w$. 
Subsequently, the least favourable distribution corresponding to $f^*$ will be calculated by solving the problem
\begin{equation} \sup_P \mbox{mse}(f^*,P), \label{least} \end{equation}
where $f^*=A^* y+b^*$ and $b^*= (I_n-A^*H)\hat{\mu}_x - A^* \hat{\mu}_w={\bf 0}$.
Then (\ref{least}) can be parameterized as
\begin{equation} \label{least 2} 
\begin{aligned}
\sup_{\mu_x,\mu_w,\Sigma_x,\Sigma_w}
&\mbox{Tr}\left[\left( A^*H-I_n\right) \Sigma_x\left( A^*H-I_n\right) ^\top+A^*\Sigma_w \left( A^*\right) ^\top \right]+\left[ \left( A^*H-I_n\right) \mu_x+A^*\mu_w\right] ^\top 
\left[ \left( A^*H-I_n\right) \mu_x+A^*\mu_w \right]  \\
\mbox{s.t. } \
& \Sigma_x \succeq{\bf 0}, \quad \Sigma_w \succeq {\bf 0}, \\
&\left\| \mu_x \right\| ^2+\mbox{Tr} \left[ \Sigma_x+\hat{\Sigma}_x-
2 \left( \hat{\Sigma}_x^{\frac{1}{2}} \Sigma_x \hat{\Sigma}_x^{\frac{1}{2}}\right) ^{\frac{1}{2}}\right]  \leq \rho_x^2, \\
&\left\| \mu_w\right\| ^2+\mbox{Tr} \left[ \Sigma_w+\hat{\Sigma}_w-
2 \left( \hat{\Sigma}_w^{\frac{1}{2}} \Sigma_w \hat{\Sigma}_w^{\frac{1}{2}}\right) ^{\frac{1}{2}}\right] 
\leq \rho_w^2.
\end{aligned} \end{equation}
For given $\Sigma_x$ and $\Sigma_w$, the maximization problem over $(\mu_x,\mu_w)$ is a QCQP in which the two constraint functions and the objective function are all homogeneous quadratic.
Then the SDP relaxation of problem (\ref{least 2}) is
\begin{align}
\begin{array}{cl}
%\phi(A,b) \triangleq
\sup \limits_{S,\Sigma_x,\Sigma_w} 
&\mbox{Tr}\left[\left(A^*H-I_n\right)\Sigma_x\left(A^*H-I_n\right)^\top
+A^*\Sigma_w \left( A^*\right) ^\top\right]+
\mbox{Tr} \left(\begin{bmatrix} (A^*H-I_n)^\top \\  (A^*)^\top  \end{bmatrix}
\begin{bmatrix} A^*H-I_n & A^* \end{bmatrix} S \right) \\
\mbox{s.t.}
&\Sigma_x \succeq{\bf 0}, \quad \Sigma_w \succeq {\bf 0}, \\
&\mbox{Tr} \left(\begin{bmatrix}
I_n & \boldsymbol{0} \\ \boldsymbol{0} & \boldsymbol{0} 
\end{bmatrix} S \right)
+\mbox{Tr} \left[ \Sigma_x+\hat{\Sigma}_x-
2 \left( \hat{\Sigma}_x^{\frac{1}{2}} \Sigma_x \hat{\Sigma}_x^{\frac{1}{2}}\right) ^{\frac{1}{2}}\right]  \leq \rho_x^2, \\
&\mbox{Tr} \left(\begin{bmatrix}
\boldsymbol{0} & \boldsymbol{0} \\ \boldsymbol{0} &I_m
\end{bmatrix} S \right) 
+\mbox{Tr} \left[ \Sigma_w+\hat{\Sigma}_w-
2 \left( \hat{\Sigma}_w^{\frac{1}{2}} \Sigma_w \hat{\Sigma}_w^{\frac{1}{2}}\right) ^{\frac{1}{2}}\right] 
\leq \rho_w^2.
\end{array} \label{least SDP1}
\end{align}
According to Theorem 2.5 in \cite{ye2003new}, problem (\ref{least SDP1}) has a rank-one optimal solution, which implies that the SDP relaxation is tight.
Thus, combined with Proposition 2 in \cite{malago2018wasserstein}, introducing auxiliary variables $V_x \in \mathbb{S}_+^n$ and $V_w \in \mathbb{S}_+^m$, problem (\ref{least 2}) is equivalent to the following SDP problem 
\begin{align}\label{least SDP}
\begin{array}{cl}
%\phi(A,b) \triangleq
\sup \limits_{S,\Sigma_x,\Sigma_w,V_x,V_w} 
&\mbox{Tr}\left[\left(A^*H-I_n\right)\Sigma_x\left(A^*H-I_n\right)^\top
+A^*\Sigma_w \left( A^*\right) ^\top\right]+
\mbox{Tr} \left(\begin{bmatrix} (A^*H-I_n)^\top \\  (A^*)^\top  \end{bmatrix}
\begin{bmatrix} A^*H-I_n & A^* \end{bmatrix} S \right) \\
\mbox{s.t.}& S \succeq{\bf 0}, \ \Sigma_x \succeq{\bf 0}, \ \Sigma_w  \succeq{\bf 0}, \ V_x  \succeq{\bf 0}, \ V_w  \succeq{\bf 0},\\
&\left[\begin{array} {cc}
\hat{\Sigma}_x^{\frac{1}{2}} \Sigma_x \hat{\Sigma}_x^{\frac{1}{2}} &
V_x \\ V_x &I_n
\end{array}
\right]\succeq {\bf 0}, \quad
\left[\begin{array} {cc}
\hat{\Sigma}_w^{\frac{1}{2}} \Sigma_w \hat{\Sigma}_w^{\frac{1}{2}} &
V_w \\ V_w &I_m
\end{array}
\right]\succeq {\bf 0},\\
& \mbox{Tr} \left(\begin{bmatrix}
I_n & \boldsymbol{0} \\ \boldsymbol{0} & \boldsymbol{0} 
\end{bmatrix} S \right)
+\mbox{Tr}\left( \Sigma_x + \hat{\Sigma}_x - 2 V_x \right) 
\leq \rho_x^2 ,\\
&\mbox{Tr} \left(\begin{bmatrix}
\boldsymbol{0} & \boldsymbol{0} \\ \boldsymbol{0} &I_m
\end{bmatrix} S \right) 
+\mbox{Tr}\left(\Sigma_w + \hat{\Sigma}_w - 2 V_w\right)  \leq \rho_w^2.
\end{array}
\end{align}
We denote the least favorable distribution corresponding to $f^*$ obtained by (\ref{least SDP}) as $\tilde{P}$,
and then compare $\mbox{mse}(f^*,P^*)$ and $\mbox{mse}(f^*,\tilde{P})$ which are the optimal values of (\ref{maxmin SDP}) and (\ref{least SDP}), respectively.
If %$\mbox{mse}(f^*,P^*) \neq \mbox{mse}(f^*,\tilde{P})$, in fact, 
$\mbox{mse}(f^*,P^*) < \mbox{mse}(f^*,\tilde{P})$, it can be inferred that $P^*$ is not the least favorable distribution corresponding to $f^*$, which indicates that the saddle point solution of (\ref{minmax}) does not exist.

\begin{figure}[htbp]
	\centering
	\includegraphics[width=0.6\linewidth]{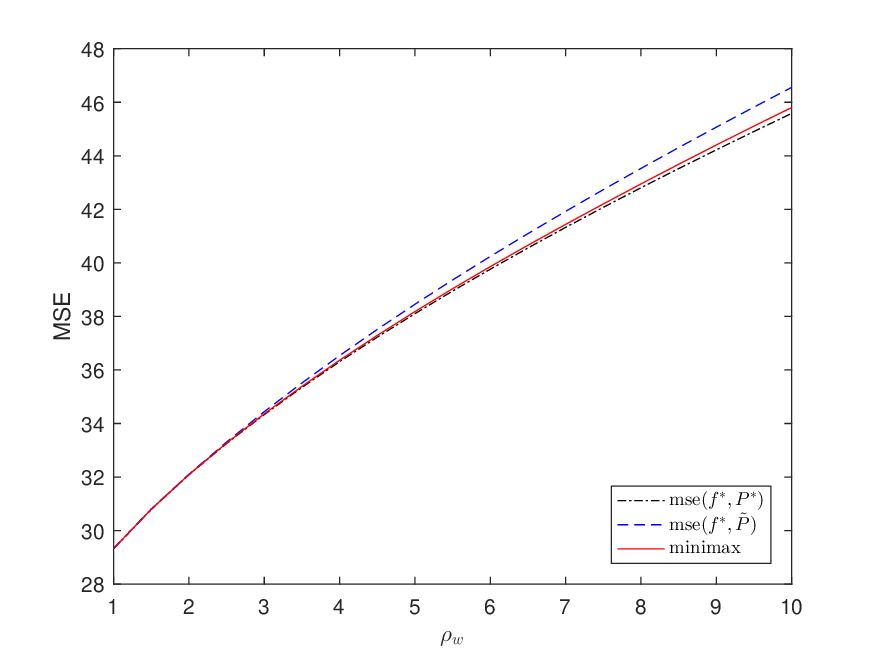}
	\caption{MSE vs $\rho_w$ when $\rho_x=3$ for $d=20$.}
\end{figure}

In Figure 5.1, $\mbox{mse}(f^*,P^*)$, $\mbox{mse}(f^*,\tilde{P})$ and the optimal value of (\ref{minmax SDP}), which is also the optimal value of the minimax problem (\ref{infLsup}), are plotted versus the quantile of the radiu $\rho_w$ for $d=20$.
It is evident that the saddle point exists when $\rho_w$ is sufficiently small. 
However, as $\rho_w$ gradually increases, $\mbox{mse}(f^*,P^*)$ is smaller than $\mbox{mse}(f^*,\tilde{P})$, which indicates that $P^*$ is no longer the least favorable distribution corresponding to $f^*$. 
Consequently, the saddle point solution of (\ref{minmax}) may not exist.
Moreover, irrespective of the existence of the saddle point, the optimal value of (\ref{infLsup}) is always greater than $\mbox{mse}(f^*,P^*)$ but less than $\mbox{mse}(f^*,\tilde{P})$, which is consistent with the theoretical result.
Furthermore, the optimal value of (\ref{infLsup}) and $\mbox{mse}(f^*,P^*)$ provide upper and lower bounds on the optimal value of the original problem (\ref{minmax}), respectively.

\subsection{The Validity of Sufficient Condition}
In the second experiment, we aim to verify the validity of the sufficient condition by comparing the lower bound on the existence of the saddle point determined by the sufficient condition in Theorem \ref{sd} with the actual bound determined by the necessary and sufficient condition in Theorem \ref{snsdu}.
The nominal mean vectors $\hat{\mu}_x$ and $\hat{\mu}_w$, nominal covariance matrices $\hat{\Sigma}_x$ and $\hat{\Sigma}_w$, and observation matrix $H$ are generated in the same way as in the first experiment. We vary the Wasserstein radiu of the parameter distribution $\rho_x$ across the interval [1,10] with a stepsize of 0.1 and calculate the lower bound determined by the sufficient condition in Theorem \ref{sd} through
$$ \rho_w^L=\frac{\lambda_{\min}^{\frac{1}{2}}\left( \hat{\Sigma}_x\right) 
	\lambda_{\min}^{\frac{1}{2}}\left(\hat{\Sigma}_w\right) }{\rho_x}.$$
In addition, for a given $\rho_x$, to determine the actual bound of the existence of the saddle point, we vary the Wasserstein radiu of the noise distribution $\rho_w$ across the interval [1,10] with a stepsize of 0.02.
For given $\rho_x$ and $\rho_w$, we slove (\ref{maxmin SDP}) and verify whether the matrix in Theorem \ref{snsdu} is negative semidefinite.
If the largest eigenvalue of the matrix is larger than 0, we assert that the saddle point solution of (\ref{minmax}) does not exist by Theorem \ref{snsdu} and record the smallest $\rho_w^*$ that satisfies this condition for each $\rho_x$.
Then the actual bound of the existence of the saddle point can be given by $\rho_w^A=\rho_w^*-0.02$.

\begin{figure}[htbp]
	\centering
	\includegraphics[width=0.6\linewidth]{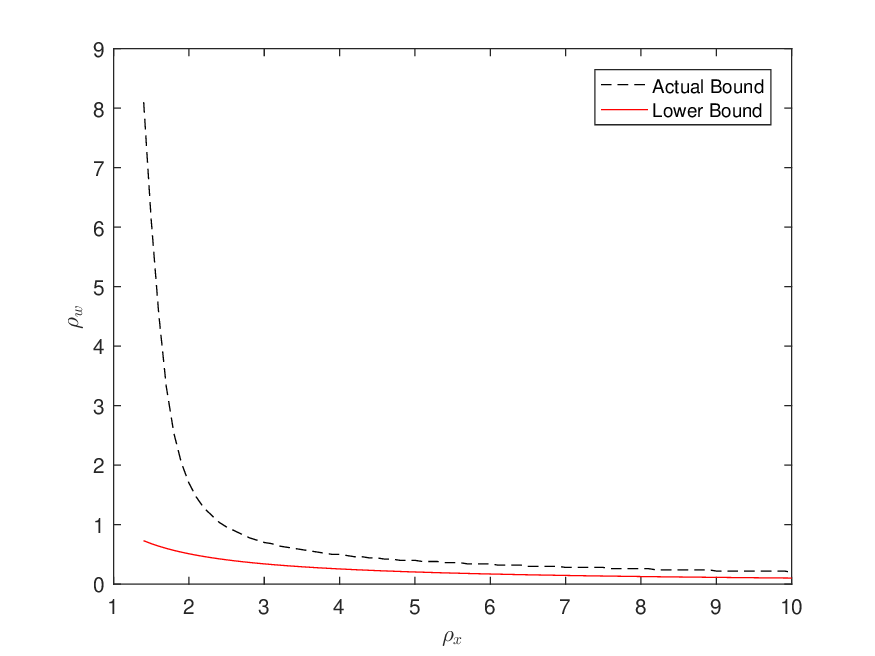}
	\caption{The actual bound and the lower bound for the existence of the saddle point.}
\end{figure}

In Figure 5.2, the lower bound calculated by the sufficient condition $\rho_w^L$ and the actual bound obtained by the traversal algorithm $\rho_w^A$ are plotted versus the quantile of the radiu $\rho_x$ for $d=20$. 
It is evident that the lower bound is consistently inferior to the actual bound. %and the gap narrows as $\rho_x$ increases.
As $\rho_x$ increases, the discrepancy between the two bounds decreases.

\subsection{The Robustness of the Robust Linear Estimator}
In the third experiment, we aim to verify the robustness of the linear estimator obtained from the upper bound problem (\ref{infLsup}) in Section 4.
In this experiment, we assume that the parameter and noise dimensions are equal, denoted by $n=m=d$. 
Furthermore, we take the observation matrix to be $H=I_d$.
Without loss of generality, the true mean vectors and the nominal mean vectors are all set to be zeros, i.e., $\mu_x=\mu_w=\hat{\mu}_x=\hat{\mu}_w={\bf 0}$.
The experiment comprises 3000 simulation runs. 
In each run, the true covariance matrices $\Sigma_x$ and $\Sigma_w$ is randomly generated in the same way as in the first experiment.
Then the nominal covariance matrices $\hat{\Sigma}_x$ and $\hat{\Sigma}_w$ are defined as the sample covariance matrices corresponding to 100 independent samples from $\mathcal{N}({\bf 0},\Sigma_x)$ and $\mathcal{N}({\bf 0},\Sigma_w)$.
%Finally, we set the Wasserstein radii of the uncertainty sets $\rho_x$ and $\rho_w$ to be the Wasserstein distances between the true distributions and the nominal distributions.
Finally, we design the Wasserstein radii of the uncertainty sets $\rho_x$ and $\rho_w$.
The simulation is repeated 1,000 times, with 100 samples drawn from the true distributions on each occasion. The Wasserstein distances between the sample covariance matrices and the true covariance matrices are then calculated and recorded. The Wasserstein distances obtained from the 1,000 simulations are sorted in ascending order, and the 0.95-quantile is taken as the Wasserstein radii $\rho_x$ and $\rho_w$ to ensure that the sample covariance matrices can be included in the uncertainty sets in most cases.

In this framework, two distributions are considered: the true distribution $P=\mathcal{N}({\bf 0},\Sigma_x) \times \mathcal{N}({\bf 0},\Sigma_w)$ and the nominal distribution $\hat{P}=\mathcal{N}({\bf 0},\hat{\Sigma}_x) \times \mathcal{N}({\bf 0},\hat{\Sigma}_w)$. 
This allows us to obtain three estimators: the optimal estimator of the nominal distribution $f^*(\hat{P})$, the optimal estimator of the true distribution $f^*(P)$, and the robust linear estimator $\tilde{f}^*$ which is proposed in Section 4. 
To verify the robustness of the estimator $\tilde{f}^*$, we can compare the relative mean square errors (RMSE) $\mbox{mse}(\tilde{f}^*,P)-\mbox{mse}(f^*(P),P)$ and 
$\mbox{mse}(f^*(\hat{P}),P)-\mbox{mse}(f^*(P),P)$. %which are the mean square errors of the robust estimator and the nominal optimal estimator under the true distribution minus the minimum mean square error of the true distribution.

Figure 5.3 presents the frequency histograms of the relative mean square errors of the robust linear estimator and the nominal optimal estimator in 3,000 repeated experiments for $d=10$ and $d=20$, respectively.
It is evident that the relative mean square error of the robust linear estimator is superior to that of the optimal estimator of the nominal distribution, and this advantage becomes increasingly apparent as the dimensions of the parameter and noise increase.

\begin{figure}[htbp]
	\centering
	\begin{minipage}{0.49\linewidth}
		\centering
		\includegraphics[scale=0.6]{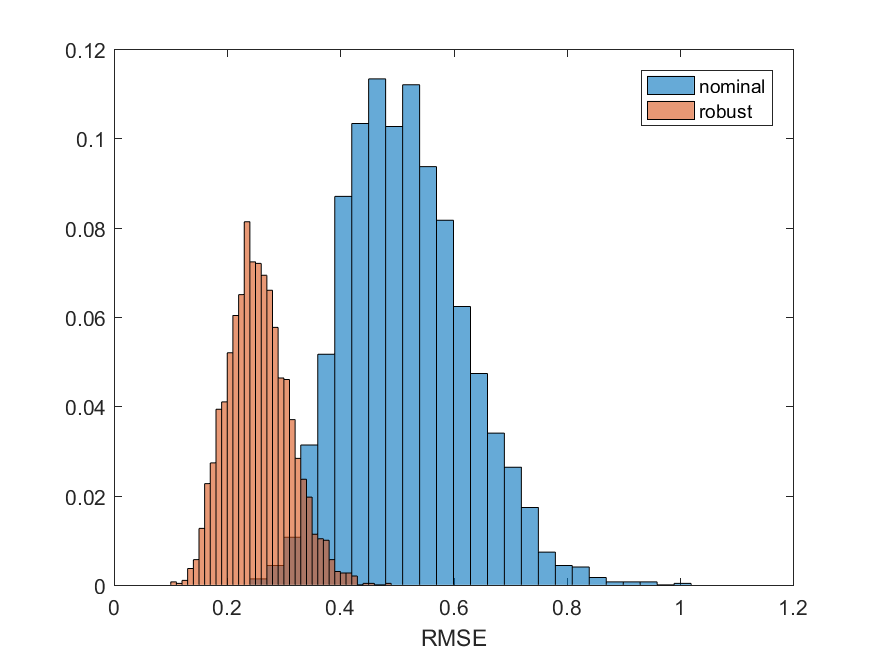}
		\centerline{(a) d=10}
	\end{minipage}
    \begin{minipage}{0.49\linewidth}
    	\centering
    	\includegraphics[scale=0.6]{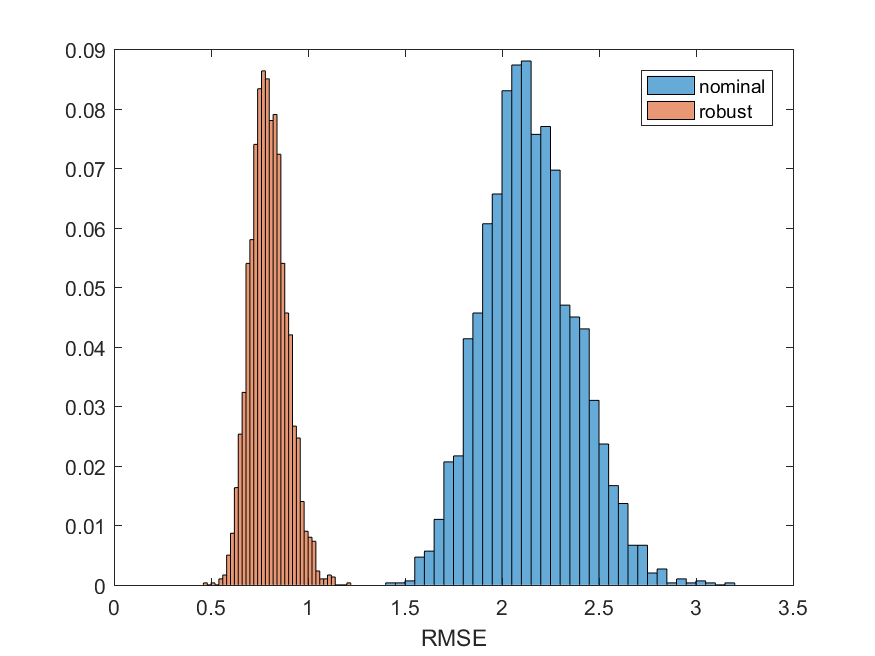}
    	\centerline{(b) d=20}
    \end{minipage}
    \caption{The frequency histograms of the RMSE.}
\end{figure}

\section{Conclusion}
In this paper, we consider a robust estimation problem in the linear measurement model with additive noise, where the parameter and noise are constrained by bounded Wasserstein-distance balls, respectively.
This robust estimation problem can be formulated as an infinite-dimensional nonconvex minimax problem whose saddle point may not exist.
By transforming the existence of its saddle point to that in a finite-dimensional minimax problem, we provide a verifiable necessary and sufficient condition and a simplified sufficient condition.
When a saddle point exists, the original infinite-dimensional minimax problem reduces to a SDP problem.
Conversely, when the saddle point is absent, the problem becomes intractable.
This fact motivates us to consider an upper-bound problem where the estimator is restricted to be linear.
By demonstrating the tightness of the SDP relaxation for the upper-bound problem, we prove that its optimal value coincides with that of a SDP problem. 
Furthermore, the optimal solution of this upper-bound problem is constructed and yields a robust linear estimator.

%%%%%%%%%%%%%%%%%%%%%%%%%%%%%%%%%%%%%%%%%%%%%%%%%%%%%%%%%%%%%
\section{Acknowledgements}
The authors would like to thank Prof. Zhi-Quan (Tom) Luo from The Chinese University of Hong Kong, Shenzhen, for the helpful discussions on this work.

\nocite{*}
\bibliographystyle{unsrt}
\bibliography{BibWDRE}
\end{document}